\documentclass[12pt]{article}
\usepackage[vmargin=1.18in,hmargin=1.1in]{geometry}				
\usepackage{amssymb,amsfonts,amsthm,mathtools,bm,bbm,mathrsfs} 	
\usepackage[dvipsnames]{xcolor} 								
\usepackage[hidelinks]{hyperref} 								
\usepackage{graphicx} 											
\usepackage{caption} 											
\usepackage{verbatim}
\usepackage[misc]{ifsym} 										
\usepackage[ruled, linesnumbered]{algorithm2e} 					
\usepackage{layouts} 
\usepackage[uniquename=false, 
uniquelist = false, 
url = false, 
doi = false, 
isbn = false, 
natbib = true, 
backend = bibtex, 
style = authoryear, 
maxbibnames = 5]{biblatex} 							

\newtheorem{theorem}{Theorem}
\newtheorem*{theorem*}{Theorem}
\newtheorem{proposition}{Proposition}
\newtheorem{lemma}{Lemma}
\newtheorem{corollary}{Corollary}

\theoremstyle{definition}
\newtheorem{definition}{Definition}
\theoremstyle{remark}

\makeatletter
\newcommand\newtarget[2]{\Hy@raisedlink{\hypertarget{#1}{}}#2}
\makeatother

\newcommand{\E}{\mathbb E}								
\newcommand{\V}{\mathrm{Var}}							
\renewcommand{\P}{\mathbb{P}}							
\newcommand{\Q}{\mathbb{Q}}								
\newcommand{\R}{\mathbb{R}}								
\newcommand{\indicator}{\mathbbm 1}						
\newcommand{\norm}[1]{\left\lVert{#1}\right\rVert}		
\newcommand{\independent}{{\perp \! \! \! \perp}}		
\newcommand{\iidsim}{\stackrel{\mathrm{i.i.d.}}{\sim}} 	
\newcommand{\convp}{\overset p \rightarrow}             
\newcommand{\convd}{\overset d \rightarrow}             
\newcommand{\convas}{\overset {a.s.} \rightarrow}       
\newcommand{\argmin}[1]{\underset{#1}{\arg \min}}       

\newcommand{\prx}{\bm X}								
\newcommand{\srx}{X}									
\newcommand{\prz}{\bm Z}								
\newcommand{\srz}{Z}									

\newcommand{\srxk}{\widetilde X}						
\newcommand{\pry}{{\bm Y}}								
\newcommand{\sry}{Y}									
\newcommand{\law}{\mathcal L}							
\newcommand{\nulllaws}{\mathscr L^0}					
\newcommand{\regclass}{\mathscr R}					    
\newcommand{\lawhat}{\widehat{\mathcal L}}				
\newcommand{\dCRT}{\textnormal{dCRT}} 					
\newcommand{\GCM}{\textnormal{GCM}}						
\newcommand{\dCRThat}{\widehat{\textnormal{dCRT}}}		
\newcommand{\MXtwohat}{\widehat{\textnormal{MX(2)}}}		
\newcommand{\ndCRThat}{\widehat{\textnormal{ndCRT}}}	
\renewcommand{\H}{\mathcal H}		 					
\newcommand{\MXtwo}{\textnormal{MX(2)}}                 
\newcommand{\convdp}{\overset {d,p} \longrightarrow}    
\newcommand{\convpp}{\overset {p,p} \longrightarrow}    

\let\oldnl\nl
\newcommand{\nonl}{\renewcommand{\nl}{\let\nl\oldnl}} 

\begin{document}
	

	\title{Reconciling model-X and doubly robust approaches to conditional independence testing}
	\author{Ziang Niu$^{*1}$, Abhinav Chakraborty$^{*1}$,
		Oliver Dukes$^2$, Eugene Katsevich$\textsuperscript{1 \Letter}$}

	\renewcommand{\thefootnote}{*} 
	\footnotetext{These authors contributed equally to this work.}
	\renewcommand{\thefootnote}{1}
	\footnotetext{Department of Statistics and Data Science, University of Pennsylvania}
	\renewcommand{\thefootnote}{2}
	\footnotetext{Department of Applied Mathematics, Computer Science and Statistics, Ghent University}
	\renewcommand{\thefootnote}{\Letter} 
	\footnotetext{\texttt{ekatsevi@wharton.upenn.edu}}
	\renewcommand{\thefootnote}{3} 
	
	\maketitle
	\thispagestyle{empty}
	\begin{abstract}
		Model-X approaches to testing conditional independence between a predictor and an outcome variable given a vector of covariates usually assume exact knowledge of the conditional distribution of the predictor given the covariates. Nevertheless, model-X methodologies are often deployed with this conditional distribution learned in sample. We investigate the consequences of this choice through the lens of the distilled conditional randomization test (dCRT). We find that Type-I error control is still possible, but only if the mean of the outcome variable given the covariates is estimated well enough. This demonstrates that the dCRT is doubly robust, and motivates a comparison to the generalized covariance measure (GCM) test, another doubly robust conditional independence test. We prove that these two tests are asymptotically equivalent, and show that the GCM test is optimal against (generalized) partially linear alternatives by leveraging semiparametric efficiency theory. In an extensive simulation study, we compare the dCRT to the GCM test. These two tests have broadly similar Type-I error and power, though dCRT can have somewhat better Type-I error control but somewhat worse power in small samples or when the response is discrete. We also find that post-lasso based test statistics (as compared to lasso based statistics) can dramatically improve Type-I error control for both methods.
	\end{abstract}

	\section{Introduction}
	
	\subsection{Conditional independence testing and the model-X assumption} \label{sec:ci-testing-and-mx}
	
	Given a predictor $\prx \in \R$, response $\pry \in \R$, and high-dimensional covariate vector $\prz \in \R^{p}$ drawn from a joint distribution $(\prx, \pry, \prz) \sim \law_n$ (potentially varying with $n$ to accommodate growing $p$), consider testing the hypothesis of conditional independence (CI)
	\begin{equation}
		H_{0n}: \pry\ \independent\ \prx\ |\ \prz
		\label{conditional-independence}
	\end{equation}
	at level $\alpha \in (0,1)$ using $n$ data points
	\begin{equation}
		\label{eq:xyz}
		(\srx, \sry, \srz) \equiv \{(\srx_i, \sry_i, \srz_i)\}_{i = 1, \dots, n} \iidsim \law_n. 
	\end{equation}
	In a high-dimensional regression setting, $H_{0n}$ is a model-agnostic way of formulating the null hypothesis that the predictor $\prx$ is unimportant in the regression of $\pry$ on $(\prx, \prz)$ \citep{CetL16}. In a causal inference setting with treatment $\prx$, outcome $\pry$, observed confounders $\prz$, and no unobserved confounders, $H_{0n}$ is the null hypothesis of no causal effect of $\prx$ on $\pry$ \citep{Pearl2009}.
	
	As \citet{Shah2018} showed, the CI null hypothesis is too large in the sense that any test controlling Type-I error on $H_{0n}$ must be powerless against all alternatives (unless $\prz$ is supported on a finite set). Therefore, additional assumptions must be placed on $\law_n$ to make progress. One such assumption is the \textit{model-X (MX) assumption} \citep{CetL16}, which states that $\law_n(\prx | \prz)$ is known exactly. Under the MX assumption, \citet{CetL16} propose the MX knockoffs and conditional randomization test (CRT) methodologies, which have elegant finite-sample Type-I error control guarantees. These MX methodologies have since exploded in popularity, undergoing active methodological development and deployment in a range of applications.
	
	One of the primary challenges in the practical application of MX methods is to obtain the required conditional distribution $\law_n(\prx|\prz)$. Outside the context of randomized controlled experiments \citep{Aufiero2022, Ham2022}, the MX assumption is an approximation \citep{Barber2018, Huang2019, Li2022}. In genome-wide association studies, a realistic parametric distribution can be postulated for this conditional law \citep{SetC17}, but the parameters of this distribution must still be learned from data. In practice, the conditional law is usually fit in sample on the same data that is used for testing, and then treated as if it were known \citep{CetL16, SetC17, SetS19, Bates2020, Liu2022a, Li2021b, Sesia2021, Katsevich2020c}. Such adaptations of MX methodologies are widely deployed, but their robustness and power properties have not been thoroughly investigated. 
	
	\subsection{Our contributions} 
	
	In this paper, we address this gap by investigating the properties of MX methods with $\law_n(\prx|\prz)$ learned in sample. This investigation leads us to establish close connections between these methods and double regression approaches to CI testing, and to explore the optimality of CI tests against semiparametric alternatives. We focus our analyses on the \textit{distilled conditional randomization test} (\textit{dCRT}), a fast and powerful instance of the CRT \citep{Liu2022a}, and the \textit{generalized covariance measure (GCM) test}, a prototypical double regression approach to CI testing \citep{Shah2018}. Both tests involve learning $\law_n(\prx \mid \prz)$ and $\law_n(\pry \mid \prz)$ in sample. Our main contributions are outlined next:
	
	\begin{enumerate}
		\item \textbf{The dCRT with $\law_n(\prx \mid \prz)$ learned in sample with can have poor Type-I error control if $\law_n(\pry \mid \prz)$ is learned poorly.} If $\law(\prx \mid \prz)$ is known exactly, then the dCRT has finite-sample Type-I error control regardless of $\law(\pry \mid \prz)$ or the quality of its estimate. This is no longer the case once $\law(\prx \mid \prz)$ is fit in sample, as we demonstrate in a numerical simulation and a theoretical counterexample (Section~\ref{sec:neg-results}).
		\item \textbf{The dCRT is doubly robust, in the sense that errors in $\law_n(\prx \mid \prz)$ can be compensated for by better approximations of $\law_n(\pry \mid \prz)$.} The MX assumption shifts the modeling burden entirely from $\law_n(\pry \mid \prz)$ to $\law_n(\prx \mid \prz)$. When the latter is fit in sample, shifting the modeling burden partially back towards $\law_n(\pry \mid \prz)$ helps recover asymptotic Type-I error control, as we demonstrate theoretically (Section~\ref{sec:double-robustness}).
		\item \textbf{The dCRT resampling distribution approaches normality, making this test asymptotically equivalent to the GCM test.} The dCRT is a resampling-based test, whereas the GCM test is asymptotic. In large samples, however, the resampling-based null distribution of the former converges to the $N(0,1)$ null distribution of the latter (Section~\ref{sec:conv-to-normal}). We show that these two tests are asymptotically equivalent against local alternatives (Section~\ref{sec:equivalence}).
		\item \textbf{The GCM test is asymptotically uniformly most powerful against local non-interacting alternatives.} Optimality results are widely prevalent in the semiparametric literature, but not in the CI testing literature. We leverage semiparametric optimality theory to prove that the GCM is the optimal CI test against  local (generalized) partially linear alternatives (Section~\ref{sec:optimality}), a broad class of alternatives in which $\prx$ and $\prz$ do not interact.
		\item \textbf{In finite samples, the dCRT and GCM test have broadly similar Type-I error and power, with some exceptions.} The asymptotic equivalence between the dCRT and GCM test largely carries over to finite samples, as we demonstrate in numerical simulations (Section~\ref{sec:simulations}). The two tests have broadly similar Type-I error and power, although there is some divergence in small samples or when $\pry$ is discrete: in these cases dCRT can have somewhat better Type-I error control but somewhat worse power. 
		\item \textbf{In finite samples, replacing the lasso with the post-lasso markedly improves Type-I error control for both dCRT and GCM test.} In MX applications, the lasso is perhaps the most common approach for learning both $\law_n(\prx \mid \prz)$ and $\law_n(\pry \mid \prz)$. However, we demonstrate in numerical simulations (Section~\ref{sec:simulations}) that the bias reduction offered by the post-lasso greatly improves Type-I error control in the context of both GCM test and dCRT, though at some cost in power. 
	\end{enumerate}
	
	On the way to making the aforementioned primary contributions, we make a few secondary contributions of independent interest:
	
	\begin{enumerate}
		\item[7.] \textbf{We reexamine numerical simulation setups from prior MX papers, finding that many have only low levels of marginal dependence between $\prx$ and $\pry$.} Prior works have used numerical simulations to establish that MX methods are fairly robust when fitting $\mathcal L_n(\prx|\prz)$ in sample. However, we note that the conditional independence testing problem~\eqref{conditional-independence} is difficult to the extent that $\prz$ induces spurious marginal dependence between $\prx$ and $\pry$ (a ``confounding'' effect). We find simulation setups in prior works have low levels of this marginal dependence (Section~\ref{sec:sim-revisiting}), potentially leading to optimistic conclusions.
		
		\item[8.] \textbf{We collate a number of conditional analogs of classical convergence theorems (some but not all novel).} The dCRT involves resampling conditionally on the observed data, so its asymptotic analysis requires reasoning about convergence after conditioning on a $\sigma$-algebra that changes with $n$. We state and prove conditional analogs of Slutsky's theorem, the law of large numbers, the central limit theorem, and other classical convergence theorems (Appendix~\ref{sec:conditional-convergence-results}). These results are not surprising, but at least some appear novel.
		
		\item[9.] \textbf{We prove a sharpened theorem on optimality in semiparametric testing.} In the literature on semiparametric \textit{estimation}, an estimator need only be regular \textit{in the vicinity of a point} for efficiency bounds to hold, whereas popular textbooks \citep{VDV1998, Kosorok2008} state semiparametric \textit{testing} optimality results \textit{globally}: a test must control Type-I error on the entire semiparametric null, rather than just in the vicinity of a point, for efficiency bounds to hold. We address this gap by proving a stronger local optimality result for semiparametric testing (Appendix~\ref{sec:semiparametric-preliminaries}).
	\end{enumerate}
	
	\subsection{Related work}
	
	The question of robustness of existing MX methods to misspecification of $\law_n(\prx \mid \prz)$ has been investigated before, though not specifically in the context of learning this distribution in sample. \citet{Berrett2019} proved that, \textit{in the worst case} over all possible test statistics and all possible distributions $\law_n(\pry \mid \prz)$, the excess Type-I error of the CRT based on an approximation to $\law_n(\prx \mid \prz)$ is bounded below by the total variation error in approximating $\prod_{i = 1}^n \law_n(\srx_i \mid \srz_i)$. This error is $O(1)$ when fitting $\law_n(\prx \mid \prz)$ in sample. We show (see contribution 1) that, even when specializing to the dCRT test statistic, Type-I error control can be poor when $\law_n(\pry \mid \prz)$ is estimated poorly. \cite{Berrett2019} provided a matching upper bound on the Type-I error of the CRT, while \citet{Barber2018} proved a similar upper bound for MX knockoffs. These worst-case bounds guarantee Type-I error control only when an additional unlabeled sample of size $N \gg n$ is available. Another kind of robustness to misspecification of the MX assumption was proposed by \citet{Katsevich2020a}; they showed that if only the first two moments of $\law_n(\prx \mid \prz)$ are known exactly, then the dCRT has asymptotic Type-I error control. Even this weaker assumption cannot be expected to hold when $\law_n(\prx \mid \prz)$ is fit in sample, however.
	
	Other MX methods have been designed specifically to have improved robustness to misspecifications of $\law_n(\prx \mid \prz)$. For example, if this law is known to belong to a parametric family with a low-dimensional sufficient statistic, MX inference can be carried out conditionally on this sufficient statistic without needing to accurately estimate the parameters themselves \citep{Huang2019,Barber2020}. The former methodology enjoys a double robustness property, related to but different from the one we state for the dCRT (see contribution 2). The \textit{conditional permutation test} \citep{Berrett2019} was proposed as a more robust variant of the CRT, though this additional robustness has yet to be formalized theoretically. Finally, the \textit{Maxway CRT} \citep{Li2022} has recently been proposed as a doubly robust analog of the dCRT. We argue that the dCRT itself is doubly robust. We conjecture that the improved empirical performance of the Maxway CRT over the (lasso-based) dCRT is primarily due to the post-lasso step in the former. Indeed, our inspiration to apply the dCRT with the post-lasso (see contribution 6) comes from the Maxway CRT; we find in simulations that this variant of the dCRT is actually more robust than the Maxway CRT.
	
	Asymptotic analysis of MX methodologies has also been undertaken before \citep{Weinstein2017, Liu2019, Weinstein2020, Katsevich2020a, Wang2020b}, although primarily for the purposes of power analyses and none in the context of fitting $\law_n(\prx \mid \prz)$ in sample. All but \citet{Katsevich2020a} assume that $\law_n(\prx \mid \prz)$ is known exactly (the full MX assumption), whereas the latter assumes that the first two moments of this distribution are known exactly. In some ways, the current work generalizes the results of \citet{Katsevich2020a}. For example, the convergence of the dCRT resampling distribution to normality (see contribution 3) was shown in a fixed-dimensional setting where $\law_n(\prx \mid \prz)$ is learned out of sample and the first two moments of $\law_n(\prx \mid \prz)$ are known. Here, we allow growing dimension, and learning both $\law_n(\prx \mid \prz)$ and $\law_n(\pry \mid \prz)$ in sample.
	
	\subsection{Notation, definitions, and preliminaries} \label{sec:notation-definitions-preliminaries}
	
	\paragraph*{Notation}
	
	We use boldface font to denote population quantities and regular font to denote sample quantities. We denote by 
	\begin{equation}
		\nulllaws_n \equiv \{\law_n: \law_n(\prx, \pry \mid \prz) = \law_n(\prx\mid\prz) \times \law_n(\pry\mid\prz)\}
	\end{equation}
	the set of laws satisfying conditional independence, and $\regclass_n$ a class of distributions satisfying some regularity assumptions. For example, the MX assumption is that 
	\begin{equation*}
	\law_n \in \regclass_n \equiv \{\law_n: \law_n(\prx | \prz) = \law_n^*(\prx | \prz)\}, 
	\end{equation*}
	where $\law_n^*(\prx | \prz)$ is a fixed, known distribution. For any regularity class $\regclass_n$, we consider testing the null hypothesis $\law_n \in \nulllaws_n \cap \regclass_n$. A sequence of tests $\phi_n: (\srx, \sry, \srz) \mapsto [0,1]$ of this null hypothesis has asymptotic Type-I error control if
	\begin{equation}
		\limsup_{n \rightarrow \infty}\sup_{\law_n \in \nulllaws_n \cap \regclass_n} \E_{\law_n}[\phi_n(\srx,\sry,\srz)] \leq \alpha.
		\label{eq:asymptotic-control}
	\end{equation}
	Let
	\begin{equation}
		\mu_{n,x}(\prz) \equiv \E_{\law_n}[\prx|\prz] \quad \text{and} \quad \mu_{n,y}(\prz) \equiv \E_{\law_n}[\pry|\prz].
	\end{equation}
	
	\paragraph*{The dCRT and $\dCRThat$}
	
	A simple approach to CI testing under the MX assumption is the \textit{conditional randomization test} (CRT, \cite{CetL16}), which controls Type-I error not just asymptotically~\eqref{eq:asymptotic-control} but in finite samples as well. The CRT is based on constructing a null distribution for any test statistic $T_n(\srx, \sry, \srz)$ by resampling $\srx$ conditionally on $\srz$ using the known conditional law $\law_n(\prx|\prz)$. While the CRT is in general computationally costly, using a test statistic of the form
	\begin{equation*}
		T_n^{\dCRT}(\srx, \sry, \srz) \equiv \frac{1}{\sqrt{n}}\sum_{i = 1}^n (\srx_i - \mu_{n,x}(\srz_i))(\sry_i - \widehat \mu_{n,y}(\srz_i))
	\end{equation*}
	gives a fast and powerful test called the \textit{distilled CRT} (dCRT, \cite{Liu2022a}). Here, $\mu_{n,x}$ is known under the MX assumption and $\widehat \mu_{n,y}$ is learned in sample. Variants of the dCRT have now been deployed in genetics \citep{Bates2020} and genomics \citep{Katsevich2020c} applications. As discussed in Section~\ref{sec:ci-testing-and-mx}, MX methodologies (including the dCRT) are usually deployed by learning $\law_n(\prx \mid \prz)$ in sample. For clarity, we give the dCRT with $\law_n(\prx \mid \prz)$ fit in sample a new name: $\dCRThat$. This procedure is based on the test statistic
	\begin{equation}
		T_n^{\dCRThat}(\srx, \sry, \srz) \equiv \frac{1}{\sqrt{n}}\sum_{i = 1}^n (\srx_i - \widehat \mu_{n,x}(\srz_i))(\sry_i - \widehat \mu_{n,y}(\srz_i)),
		\label{eq:dcrt-hat-stat}
	\end{equation}
	where $\widehat \mu_{n,x}(\srz_i) \equiv \E_{\lawhat_n}[\srx_i \mid \srz_i]$. The $\dCRThat$ procedure is outlined in Algorithm~\ref{alg:dcrt-hat}; one of the primary goals of this paper is to study this procedure.
	
	\begin{center}
		\begin{minipage}{\linewidth}
			\begin{algorithm}[H]
				\nonl  \textbf{Input:}  Data $(\srx,\sry,\srz)$, number of randomizations $M$. \\
				Learn $\lawhat_n(\prx|\prz)$ based on $(\srx, \srz)$ and $\widehat \mu_{n,y}(\prz)$ based on $(\sry, \srz)$\;
				Compute $T_n^{\dCRThat}(\srx, \sry, \srz)$\;
				\For{$m = 1, 2, \dots, M$}{
					Sample $\srxk^{(m)}|\srx, \sry, \srz \sim \prod_{i = 1}^n \lawhat_n(\srx_i|\srz_i)$ and compute 
					\begin{equation}
						T_n^{\dCRThat}(\srxk^{(m)}, \srx, \sry, \srz) \equiv \frac{1}{\sqrt{n}}\sum_{i = 1}^n (\srxk_i - \widehat \mu_{n,x}(\srz_i))(\sry_i - \widehat \mu_{n,y}(\srz_i)); \label{eq:resampled-dcrt-def}
					\end{equation}
				}
				\nonl \textbf{Output:} $\dCRThat$ $p$-value	 \small $\frac{1}{M+1} (1+ \sum_{m=1}^M\indicator\{T_n^{\dCRThat}(\srxk^{(m)}, \srx, \sry, \srz) \geq T_n^{\dCRThat}(\srx, \sry, \srz)\}).$
				\caption{\bf The $\dCRThat$.}
				\label{alg:dcrt-hat}
			\end{algorithm}
		\end{minipage}
	\end{center}
	The resampled test statistics $T_n^{\dCRThat}(\srxk^{(m)}, \srx, \sry, \srz)$~\eqref{eq:resampled-dcrt-def} have four arguments instead of three in order to emphasize that the conditional mean $\widehat \mu_{n,x}(\cdot)$ is not refit upon resampling.
	
	\paragraph*{The GCM test and double robustness}
	
	Another CI test is the GCM test \citep{Shah2018}, defined as
	\begin{equation}
		\phi_n^{\GCM}(\srx, \sry, \srz) \equiv \indicator(T_n^{\GCM}(\srx, \sry, \srz) > z_{1-\alpha}),
		\label{eq:gcm-test}
	\end{equation}
	where
	\begin{equation}
		T_n^{\GCM}(\srx, \sry, \srz) \equiv \frac{1}{\widehat S_{n}^{\GCM}}\frac{1}{\sqrt{n}}\sum_{i = 1}^n (\srx_i - \widehat \mu_{n,x}(\srz_i))(\sry_i - \widehat \mu_{n,y}(\srz_i)) \equiv \frac{1}{\widehat S_{n}^{\GCM}}T_n^{\dCRThat}(\srx, \sry, \srz)
	\end{equation}
	and $(\widehat S_{n}^{\GCM})^2$ is the empirical variance of the product-of-residual summands:
	\begin{equation}
		(\widehat S_{n}^{\GCM})^2 \equiv \widehat{\V}\{(\srx_i - \widehat \mu_{n,x}(\srz_i))(\sry_i - \widehat \mu_{n,y}(\srz_i))\}.
	\end{equation}
	It controls Type-I error if the following in-sample mean-squared error quantities are small \citep{Shah2018}:
	\small
	\begin{equation*}
		E_{n, x} \equiv \left(\frac{1}{n}\sum_{i = 1}^n (\widehat \mu_{n,x}(\srz_i) -  \mu_{n,x}(\srz_i))^2\right)^{1/2};\ E'_{n, x} \equiv \left(\frac{1}{n}\sum_{i = 1}^n (\widehat \mu_{n,x}(\srz_i) -  \mu_{n,x}(\srz_i))^2\textnormal{Var}_{\law_n}[\sry_i|\srz_i]\right)^{1/2};
	\end{equation*}	
	\begin{equation*}
		E_{n, y} \equiv \left(\frac{1}{n}\sum_{i = 1}^n (\widehat \mu_{n,y}(\srz_i) -  \mu_{n,y}(\srz_i))^2\right)^{1/2};\ E'_{n, y} \equiv \left(\frac{1}{n}\sum_{i = 1}^n (\widehat \mu_{n,y}(\srz_i) -  \mu_{n,y}(\srz_i))^2\textnormal{Var}_{\law_n}[\srx_i|\srz_i]\right)^{1/2}.
	\end{equation*}
	\normalsize
	In particular, \citet{Shah2018} require that
	\begin{equation}
		E_{n, x} E_{n, y}  = o_{\law_n}(n^{-1/2}),\  E'_{n, x} = o_{\law_n}(1),\  E'_{n, y} = o_{\law_n}(1)
		\label{eq:sp1}, \tag{SP1}
	\end{equation}
	and, for some constants $c_1, c_2, \delta > 0$,
	\begin{equation}
		\begin{split}
			&\inf_n\ \E_{\law_n}[(\prx-\mu_{n,x}(\prz))^2(\pry-\mu_{n,y}(\prz))^2] > c_1\\
			&\sup_n\ \E_{\law_n}[|(\prx-\mu_{n,x}(\prz))(\pry-\mu_{n,y}(\prz))|^{2+\delta}] < c_2.
			\label{eq:sp2}
		\end{split}
		\tag{SP2}
	\end{equation}
	The GCM test is therefore doubly robust in the sense that it controls Type-I error if the product of the estimation errors for $\E[\prx|\prz]$ and $\E[\pry|\prz]$ ($E_{n,x}E_{n,y}$) converges to zero at the $o_{\law_n}(n^{-1/2})$ rate. Note that this is a \textit{rate double robustness} property rather than a \textit{model double robustness} property; see \citet{Smucler2019} for a discussion of this distinction.
	
	\section{$\dCRThat$ resampling distribution converges to normal} \label{sec:conv-to-normal}
	
	To make it easier to analyze the asymptotic properties of the $\dCRThat$, in this section we prove that it is asymptotically equivalent to the resampling-free $\MXtwohat$ $F$-test, a variant of the $\MXtwo$ $F$-test \citep{Katsevich2020a} where the first two moments of $\law_n(\prx|\prz)$ are estimated in sample. This equivalence was already shown by these authors in the case when $\mu_{n,x}$ is known and $\widehat \mu_{n,y}$ is fit out of sample (see their Theorem 2). They conjectured that the equivalence continues to hold when $\widehat \mu_{n,y}$ is fit in sample. Here, we prove this conjecture, not just when $\widehat \mu_{n,y}$ is fit in sample, but also when the first two moments of $\mu_{n,x}$ are unknown and also fit in sample.
	
	Note that the variance of the resampling distribution of $T_n^{\dCRThat}$ is
	\begin{equation}
		(\widehat S_{n}^{\dCRThat})^2 \equiv \V_{\lawhat_n}[T_n^{\dCRThat}(\srxk, \srx, \sry, \srz) \mid \srx, \sry, \srz] = \frac{1}{n}\sum_{i = 1}^n \V_{\lawhat_n}[\srx_i|\srz_i](\sry_i - \widehat \mu_{n,y}(\srz_i))^2.
		\label{eq:conditional-variance-def}
	\end{equation}
	It will be convenient to reformulate $\dCRThat$ as 
	\begin{equation*}
		\begin{split}
			\phi^{\dCRThat}_n(\srx, \sry, \srz) &\equiv \indicator(T_n^{\dCRThat}(\srx, \sry, \srz) > \Q_{1-\alpha}[T_n^{\dCRThat}(\srxk, \srx, \sry, \srz) \mid \srx, \sry, \srz]) \\
			&= \indicator\left(\frac{1}{\widehat S_{n}^{\dCRThat}}T_n^{\dCRThat}(\srx, \sry, \srz) > \Q_{1-\alpha}\left[\frac{1}{\widehat S_{n}^{\dCRThat}}T_n^{\dCRThat}(\srxk, \srx, \sry, \srz) \mid \srx, \sry, \srz\right]\right) \\
			&\equiv \indicator\left(\frac{1}{\widehat S_{n}^{\dCRThat}}T_n^{\dCRThat}(\srx, \sry, \srz) > C^{\dCRThat}_n(\srx, \sry, \srz)\right).
		\end{split}
	\end{equation*}
	Note that this test is obtained from that in Algorithm~\ref{alg:dcrt-hat} by sending $M \rightarrow \infty$; we focus our theoretical analysis here and throughout on this infinite-resamples limit of the $\dCRThat$. Here, the $\alpha$ conditional quantile $\Q_{\alpha}[W \mid \mathcal F]$ of a random variable $W$ given a $\sigma$-algebra $\mathcal F$ is defined via
	\begin{equation}
		\mathbb{Q}_{\alpha}[W \mid \mathcal F] \equiv \inf\{t:\P[W \leq t \mid \mathcal F] \geq \alpha\}.
	\end{equation}
	One would expect, based on the central limit theorem, that the conditional distribution of the ratio $T_n^{\dCRThat}(\srxk, \srx, \sry, \srz)/\widehat S_{n}^{\dCRThat}$ tends to $N(0,1)$. This statement is complicated by the conditioning event, which requires us to be careful to define conditional convergence in distribution:
	
	\begin{definition} \label{def:conditional-convergence-distribution}
		For each $n$, let $W_n$ be a random variable and let $\mathcal F_n$ be a $\sigma$-algebra. Then, we say $W_n$ converges in distribution to a random variable $W$ conditionally on $\mathcal F_n$ if
		\begin{equation}
			\P[W_n \leq t \mid \mathcal F_n] \convp \P[W \leq t] \ \text{for each } t \in \R \text{ at which } t \mapsto \P[W \leq t] \text{ is continuous.}
		\end{equation}
		We denote this relation via $W_n \mid \mathcal F_n \convdp W$.
	\end{definition}
	
	Based on an extension of the Lyapunov central limit theorem to conditional convergence in distribution (Theorem~\ref{thm:conditional-clt}), we get the following result:
	\begin{theorem} \label{thm:normal-limit}
		Suppose the sequences of true and learned laws $\law_n$ and $\lawhat_n$ satisfy the following two nondegeneracy properties:
		\begin{gather}
			\P_{\law_n}[(\widehat S_{n}^{\dCRThat})^2 \geq \epsilon] \rightarrow 1 \ \text{for some } \epsilon > 0
			\label{eq:var-bounded-below}; \tag{NDG1} \\
			0 < \V_{\lawhat_n}[\srx_i|\srz_i], (\sry_i - \widehat \mu_{n, y}(\srz_i))^2, (\sry_i - \mu_{n, y}(\srz_i))^2 < \infty \ \text{almost surely}. \label{eq:variance-bounds} \tag{NDG2}
		\end{gather}
		If the conditional Lyapunov condition
		\begin{equation}
			\frac{1}{n^{1+\delta/2}} \sum_{i=1}^n |\sry_i-\widehat\mu_{n,y}(\srz_i)|^{2+\delta}\E_{\lawhat_n}\left[|\srxk_i-\widehat\mu_{n,x}(\srz_i)|^{2+\delta}\mid \srx,\srz\right] \convp 0
			\label{eq:lyapunov-condition} \tag{Lyap-1}
		\end{equation}
		is satisfied for some $\delta > 0$, then
		\begin{equation}
			\frac{1}{\widehat S_{n}^{\dCRThat}}T_n^{\dCRThat}(\srxk, \srx, \sry, \srz) \mid \srx, \sry, \srz \convdp N(0,1)
			\label{eq:conditional-convergence}
		\end{equation}
		and therefore
		\begin{equation}
			C^{\dCRThat}_n(\srx, \sry, \srz) \equiv \Q_{1-\alpha}\left[\frac{1}{\widehat S_{n}^{\dCRThat}}T_n^{\dCRThat}(\srxk, \srx, \sry, \srz) \mid \srx, \sry, \srz\right] \convp z_{1-\alpha}.
			\label{eq:critical-value-convergence}
		\end{equation}
	\end{theorem}
	
	This suggests that the $\dCRThat$ is asymptotically equivalent to the $\MXtwohat$ $F$-test, defined
	\begin{equation}
		\phi_n^{\MXtwohat}(\srx,\sry,\srz) \equiv \indicator\left(\frac{1}{\widehat S_{n}^{\dCRThat}}T_n^{\dCRThat}(\srx, \sry, \srz) > z_{1-\alpha}\right).
		\label{eq:mx-2-f-test}
	\end{equation}
	Indeed, we have the following corollary.
	
	\begin{corollary} \label{cor:asymptotic-equivalence} 
		Consider a sequence of laws $\law_n$ satisfying the assumptions~\eqref{eq:var-bounded-below}, \eqref{eq:variance-bounds}, and~\eqref{eq:lyapunov-condition} of Theorem~\ref{thm:normal-limit}, and assume that the test statistic does not accumulate near $z_{1-\alpha}$, i.e.
		\begin{equation}
			\lim_{\delta \rightarrow 0}\limsup_{n \rightarrow \infty}\ \P_{\law_n}[|T_n^{\dCRThat}(\srx, \sry, \srz)-z_{1-\alpha}| \leq \delta] = 0.
			\label{eq:non-accumulation}
		\end{equation}
		Then, the $\dCRThat$ is asymptotically equivalent to the $\MXtwohat$ $F$-test:
		\begin{equation}
			\lim_{n \rightarrow \infty}\P_{\law_n}[\phi_n^{\dCRThat}(\srx, \sry, \srz) = \phi_n^{\MXtwohat}(\srx, \sry, \srz)] = 1.
		\end{equation}
		
	\end{corollary}
	
	This result extends \citet[Theorem 2]{Katsevich2020a} by allowing $\widehat \mu_{n,x}$ and $\widehat \mu_{n,y}$ to be fit in sample, rather than assuming $\mu_{n,x}$ is known and $\widehat \mu_{n,y}$ is fit out of sample. It is a first indication that the $\dCRThat$ approximates a test based on asymptotic normality.
	
	\section{$\dCRThat$ is not robust for general $\widehat \mu_{n,y}$} \label{sec:neg-results}
		
	One of the hallmarks of MX inference is that it requires ``no restriction on the dimensionality of the data or the conditional distribution of [$\law_n(\pry|\prz)$]'' \citep{CetL16}. For the CRT, this means that Type-I error is controlled in finite samples, regardless of the test statistic used or the distribution of the response variable. If $\law_n(\prx|\prz)$ is described by a parametric model with $k$ unknown parameters and we have $N \gg n \cdot k$ unlabeled samples to learn this model, then at least asymptotic Type-I error control is still possible without assumptions on $\law_n(\pry|\prz)$ \citep{Berrett2019}. By contrast, in this section we show that when $\law_n(\prx|\prz)$ is approximated in sample, we cannot expect Type-I error control without assumptions on the response variable.
	
	Let us consider a simple null model $\law_n$ with
	\begin{equation}
		\law_n(\prz) = N(0, I_p), \quad \law_n(\prx|\prz) = N(\prz^T \beta, 1), \quad \text{and} \quad \law_n(\pry|\prz) = N(\prz^T \beta, 1).
		\label{eq:simple-null-model}
	\end{equation}
	Suppose we fit $\law_n(\prx|\prz)$ via a ridge regression while using the trivial estimate $\widehat \mu_{n,y}(\prz) \equiv 0$ for $\E[\pry|\prz]$. To build intuition while avoiding technical difficulties, we loosely approximate the ridge regression estimator as $\widehat \beta_n \equiv (1-\frac{c}{\sqrt{n}}) \beta$, where the $1/\sqrt{n}$ error term reflects that we are fitting $\widehat \beta_n$ in sample (and is optimistic in the sense that it ignores possible growth in $p$). Then, consider the $\dCRThat$ based on $\lawhat_n(\prx|\prz) = N(\prz^T \widehat \beta_n, 1)$ and  $\widehat \mu_{n,y}(\prz) \equiv 0$. In this case, the normality of $\lawhat_n(\prx|\prz)$ leads to normality of the resampling distribution holding not just asymptotically~\eqref{eq:conditional-convergence} but in finite samples as well. Therefore, the $\dCRThat$ is equal to the $\MXtwohat$ $F$-test:
	\begin{equation}
		\phi^{\dCRThat}_{n}(\srx, \sry, \srz) = \indicator\left(\frac{1}{\sqrt{\frac{1}{n}\sum_{i = 1}^n \sry_i^2}}\frac{1}{\sqrt n}\sum_{i = 1}^n (\srx_i - \srz_i^T \widehat \beta_n)\sry_i > z_{1-\alpha}\right).
	\end{equation}
	On the other hand, it is easy to derive that
	\begin{equation}
		\frac{1}{\sqrt{\frac{1}{n}\sum_{i = 1}^n \sry_i^2}}\frac{1}{\sqrt n}\sum_{i = 1}^n (\srx_i - \srz_i^T \widehat \beta_n)\sry_i \convd N\left(\frac{c\|\beta\|^2}{\sqrt{\|\beta\|^2+1}}, 1\right).
	\end{equation}
	Therefore, the limiting Type-I error of the $\dCRThat$ in this case is
	\begin{equation}
		\lim_{n \rightarrow \infty}\E_{\law_n}[\phi^{\dCRThat}_{n}(\srx, \sry, \srz)] = 1-\Phi\left(z_{1-\alpha} -\frac{c\|\beta\|^2}{\sqrt{\|\beta\|^2+1}}\right),
	\end{equation}
	which can be made arbitrarily close to one as $c \rightarrow \infty$. This issue is caused by a combination of the $O(1/\sqrt{n})$ shrinkage bias in the estimator for $\mu_{n,x}$ and the failure to estimate $\mu_{n,y}$. This leaves an $O(1/\sqrt{n})$ correlation between $\prx-\widehat \mu_{n,x}(\prz)$ and $\pry$ induced by $\prz$, which shifts the mean of the null distribution of the $\dCRThat$ test statistic away from zero by a nontrivial amount.
	
	Numerical simulations (although with lasso instead of ridge regression) confirm this phenomenon. We constructed a numerical simulation based on the null model~\eqref{eq:simple-null-model} with $n = 1600$,  $p = 400$, and $\beta$ having only $s = 5$ nonzero entries (see Section~\ref{sec:sim-design} below for more on our data-generating model). In this setting, we applied the $\dCRThat$ using the cross-validated lasso and intercept-only models to estimate $\mu_{n,x}$ and $\mu_{n,y}$, respectively. As we increased the magnitude of the coefficient vector $\beta$, this test exhibited significant loss of Type-I error control (Figure~\ref{fig:negative_Result}). By contrast, using the lasso instead of the intercept-only model to estimate $\mu_{n,y}$ reduced the Type-I error to nearly the nominal level.
	\begin{figure}[!ht]
		\centering
		\includegraphics[scale = 0.85]{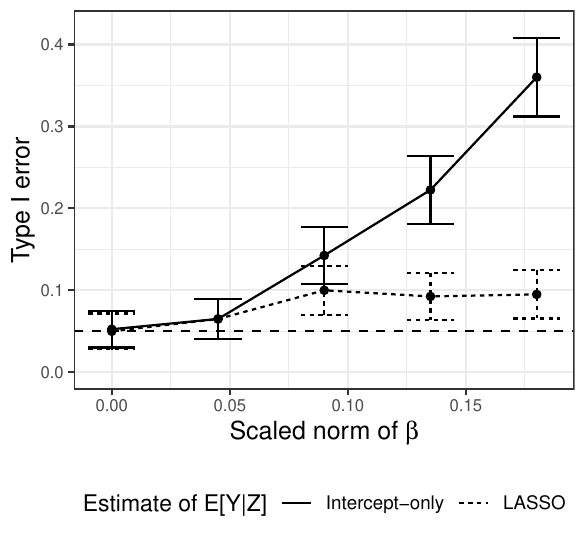}
		\caption{The Type-I error of two instances of the $\dCRThat$ under the data-generating model~\eqref{eq:simple-null-model}, depending on which method is used to estimate $\mu_{n,y}$, when the lasso is used to estimate $\mu_{n,x}$. Improved estimation of $\mu_{n,y}$ leads to markedly reduced Type-I error.}
		\label{fig:negative_Result}
	\end{figure}
	
	So even when $\law_n(\prx|\prz)$ is estimated at a parametric rate (albeit with regularization),  the $\dCRThat$ can have inflated Type-I error rate for certain test statistics. A similar observation was made by \citet{Li2022} (see the discussion after Theorem 3). Similar phenomena have been noted in the contexts of causal inference \citep{Dukes2020a} and doubly robust estimation \citep{Chernozhukov2018, Chernozhukov2022}; in the latter literature this issue is called ``regularization bias.'' We note that poor estimation of $\E[\pry|\prz]$, in conjunction with the plug-in resampling scheme of the $\dCRThat$ can also lead to conservative inference rather than liberal inference. This happens in cases when $\widehat \beta_n$ is an efficient estimator of $\beta$, e.g. that derived from ordinary least squares. In the causal inference context, this conservatism is a consequence of the fact that using estimated propensity scores can lead to more efficient estimates than using known propensity scores \citep{Robins1992, Henmi2004}. If the propensity score is estimated but the standard error is constructed as though it were known, then conservative inference would result. 
		
	As already alluded to, the Type-I error inflation in the above example stems from the fact that
	\begin{equation*}
	\E_{\law_n}[(\widehat \mu_{n,x}(\prz) - \mu_{n,x}(\prz))(\widehat \mu_{n,y}(\prz) - \mu_{n,y}(\prz))] = O(1/\sqrt{n}), 
	\end{equation*}
	a rate insufficient for Type-I error control. If we had at least consistency of $\widehat \mu_{n,y}(\prz)$, then this rate would improve to $o(1/\sqrt{n})$ and Type-I error control would be restored. This intuition is supported by the simulation results in Figure~\ref{fig:negative_Result}, where estimating $\E[\pry|\prz]$ via lasso brought the Type-I error down to nearly the nominal level. This discussion suggests that, if $\law_n(\prx|\prz)$ is learned in sample (or on an external sample of similar size), then assumptions must be placed not only on $\law_n(\prx|\prz)$ but also on $\law_n(\pry|\prz)$ for Type-I error control. This motivates us to investigate the double robustness of the $\dCRThat$ and compare it to the GCM test.
	
	\section{$\dCRThat$ is doubly robust and equivalent to GCM test} \label{sec:dr-and-equivalence}
		
	Of course, in practice $\widehat \mu_{n,y}$ is not fit as naively as in the counterexample from Section~\ref{sec:neg-results}. The conditional mean $\E[\pry|\prz]$ is usually approximated via a machine learning algorithm, as improved approximation of this quantity improves the power of the dCRT \citep{Katsevich2020a}. In the context where $\law_n(\prx|\prz)$ must be approximated, we claim that more accurate estimation of $\E[\pry|\prz]$ can improve not just the power but also the Type-I error control of the $\dCRThat$. We formalize this by showing that the $\dCRThat$ is \textit{doubly robust} (recall Section~\ref{sec:notation-definitions-preliminaries}). This property is a consequence of the fact that, under the null, the $\dCRThat$ is asymptotically equivalent to the GCM test, which itself is doubly robust. This equivalence also implies that the $\dCRThat$ and GCM test have the same asymptotic power against contiguous alternatives.
	
	\subsection{Equivalence between GCM test and $\dCRThat$} \label{sec:equivalence}
	
	When comparing the GCM test~\eqref{eq:gcm-test} to the $\MXtwohat$ $F$-test~\eqref{eq:mx-2-f-test}, which is asymptotically equivalent to the $\dCRThat$ (Corollary~\ref{cor:asymptotic-equivalence}), the only difference is the normalization term. Under the null hypothesis, this difference vanishes asymptotically as long as the estimated variance $\V_{\lawhat_n}[\prx|\prz]$ is consistent in the following sense:
	\begin{equation}
		\frac{1}{n} \sum_{i=1}^n (\V_{\lawhat_n}[\srx_i\mid \srz_i]-\V_{\law_n}[\srx_i\mid \srz_i])\V_{\law_n}[Y_i \mid Z_i] \convp 0.
		\label{eq:variance-consistency}
	\end{equation}
	In preparation to state our equivalence result, we augment the assumption~\eqref{eq:sp1} as follows:
	\begin{equation}
		E_{n, x} E_{n, y}  = o_{\law_n}(n^{-1/2}),\  E'_{n, x} = o_{\law_n}(1),\  E'_{n, y} = o_{\law_n}(1),\ \widehat E'_{n,y} = o_{\law_n}(1),
		\label{eq:sp1prime} \tag{SP1'}
	\end{equation}
	where
	\begin{equation}
		\widehat E'_{n, y} \equiv \left(\frac{1}{n}\sum_{i = 1}^n (\widehat \mu_{n,y}(\srz_i) -  \mu_{n,y}(\srz_i))^2\textnormal{Var}_{\lawhat_n}[\srx_i|\srz_i]\right)^{1/2}.
	\end{equation}
	
	\begin{theorem} \label{thm:equivalence}
		Suppose $\law_n \in \nulllaws_n$ is a sequence of laws satisfying the assumptions~\eqref{eq:sp1prime} and~\eqref{eq:sp2}, the nondegeneracy condition~\eqref{eq:variance-bounds}, the variance consistency property~\eqref{eq:variance-consistency} and the Lyapunov condition 
		\begin{equation}
			\frac{1}{n^{1+\delta/2}} \sum_{i=1}^n \E_{\law_n}\left[|\sry_i-\mu_{n,y}(\srz_i)|^{2+\delta}\mid \srz_i\right]\E_{\lawhat_n}[|\srxk_i-\widehat\mu_{n,x}(\srz_i)|^{2+\delta}\mid \srx,\srz] \convp 0.
			\label{eq:lyapunov-condition-2} \tag{Lyap-2}
		\end{equation}
		Then, the $\dCRThat$ and GCM variance estimates are asymptotically equivalent:
		\begin{equation}
			\frac{(\widehat S_n^{\dCRThat})^2}{(\widehat S_{n}^{\GCM})^2} \convp 1,
			\label{eq:asymptotic-variance-equivalence}
		\end{equation}
		as are the $\dCRThat$ and GCM tests themselves:
		\begin{equation}
			\lim_{n \rightarrow \infty} \P_{\law_n}[\phi^{\dCRThat}_{n}(\srx, \sry, \srz) = \phi^{\GCM}_{n}(\srx, \sry, \srz)] = 1.
			\label{eq:asymptotic-test-equivalence}
		\end{equation}
	\end{theorem}
	
	The variance consistency property~\eqref{eq:variance-consistency} is relatively easy to achieve, given the other assumptions of Theorem~\ref{thm:equivalence}. The following proposition states two sufficient conditions for this property.
	
	\begin{proposition} \label{prop:sufficient-for-variance-consistency} 
		If the assumptions of Theorem~\ref{thm:equivalence} other than variance consistency~\eqref{eq:variance-consistency} hold, then the latter property holds in the following two cases:
		\begin{enumerate}
			\item $\V_{\lawhat_n}[\srx_i|\srz_i] \equiv (\srx_i - \widehat \mu_{n,x}(\srz_i))^2$;
			\item $\V_{\lawhat_n}[\prx|\prz] \equiv f(\widehat \mu_{n,x}(\prz))$, if 
			\begin{itemize}
				\item $\V_{\law_n}[\prx|\prz] = f(\mu_{n,x}(\prz))$ for $f$ Lipschitz on domain $\cup_{n=1}^{\infty}\mathrm{Conv}(\mathrm{supp}(\law_{n}(\prx)))$ and $\mathrm{supp}(\widehat \mu_{n,x}(\prz))\subseteq\mathrm{Conv}(\mathrm{supp}(\law_{n}(\prx)))$ almost surely for every $n$;
				\item $\sup_n \E_{\law_n}[|\pry-\mu_{n,y}(\prz)|^{2+\delta}] < \infty$ for some $\delta > 0$.
			\end{itemize}
		\end{enumerate} 
	\end{proposition}
	
	The first variance estimate given in the proposition can always be applied; the second applies to cases when the mean-variance relationship for $\law_n(\prx|\prz)$ is known and Lipschitz on the convex hull of the support of $\prx$, denoted $\mathrm{Conv}(\law_{n}(\prx))$. This is the case, for example, if $\prx$ is binary and we define $f(t) \equiv t(1-t)$.
	
	One consequence of Theorem~\ref{thm:equivalence} is that the $\dCRThat$ and GCM test are also asymptotically equivalent against local alternatives, so in particular have the same power.
	
	\begin{corollary} \label{cor:asymptotic-equivalence-alternative}
		If $\law_n'$ is a sequence of alternative distributions that is contiguous to a sequence $\law_n \in \nulllaws_n$ satisfying the assumptions of Theorem~\ref{thm:equivalence}, then the $\dCRThat$ and GCM tests are asymptotically equivalent against $\law'_n$:
		\begin{equation}
			\lim_{n \rightarrow \infty} \P_{\law_n'}[\phi^{\dCRThat}_{n}(\srx, \sry, \srz) = \phi^{\GCM}_{n}(\srx, \sry, \srz)] = 1
			\label{eq:equivalent-tests-alternative}
		\end{equation}
		and therefore have the same asymptotic power:
		\begin{equation}
			\lim_{n \rightarrow \infty} \E_{\law'_n}[\phi_n^{\dCRThat}(\srx, \sry, \srz)] - \E_{\law'_n}[\phi_n^{\GCM}(\srx, \sry, \srz)] = 0.
		\end{equation}
	\end{corollary}
	
	By constructing a null distribution via resampling, the CRT allows for arbitrarily complicated test statistics whose asymptotic distributions are not known. For the $\dCRThat$, however, the resampling-based null distribution simply recapitulates the asymptotic normal distribution used by the GCM test (Theorems~\ref{thm:normal-limit} and~\ref{thm:equivalence}). Therefore, at least in large samples, the extra computational burden of resampling is unnecessary as the equivalent GCM can be applied instead.
	
	\subsection{Double robustness of $\dCRThat$} \label{sec:double-robustness}
	
	Another consequence of Theorem~\ref{thm:equivalence} is that the $\dCRThat$ is doubly robust under the variance consistency condition~\eqref{eq:variance-consistency}, since it is equivalent under the null hypothesis to the doubly robust GCM test. 
	
	\begin{corollary} \label{cor:dcrt-double-robustness}
		Let $\regclass_n$ be a sequence of regularity conditions such that for any sequence $\law_n \in \regclass_n$, we have the nondegeneracy condition~\eqref{eq:variance-bounds}, the Lyapunov condition~\eqref{eq:lyapunov-condition-2}, the conditions~\eqref{eq:sp1prime} and~\eqref{eq:sp2}, and consistent variance estimates~\eqref{eq:variance-consistency}. Then, the $\dCRThat$ has asymptotic Type-I error control over $\nulllaws_n \cap \regclass_n$ in the sense of the definition~\eqref{eq:asymptotic-control}.
	\end{corollary}
	
	Therefore, Type-I error control requires accuracy of only the first two moments of $\lawhat_n$, in parallel to Theorem 2 of \citet{Katsevich2020a}. The condition on the second moment of $\lawhat_n(\prx|\prz)$ is needed because the variance of the resampling distribution must not be smaller (asymptotically) than the true variance of the test statistic. This condition does not require much more than accurate estimation of the first moments (Proposition~\ref{prop:sufficient-for-variance-consistency}). It can be dropped altogether if we build normalization directly into the $\dCRThat$ test statistic. We explore this possibility in Appendix~\ref{sec:ndcrt}.
	
	Our conclusion that $\dCRThat$ is doubly robust initially appears at odds with the statement that ``the model-X CRT...does not pursue such double robustness through learning and adjusting for both $X|Z$ and $Y|Z$...'' \citep{Li2022}. This statement is in reference to the worst-case performance of the CRT across all possible test statistics \citep{Berrett2019}. We agree that this worst-case performance can be poor when learning $\law_n(\prx|\prz)$ in sample (Section~\ref{sec:neg-results}). However, the test statistics applied in conjunction with the CRT (such as the dCRT statistic) do usually involve learning and adjusting for $\law_n(\pry|\prz)$. In this sense, practical applications of the (d)CRT do learn and adjust for both $\law_n(\prx|\prz)$ and $\law_n(\pry|\prz)$; the former is learned when approximating the ``model for X'' and the latter when computing the test statistic. If the quality of these estimates is sufficiently good, then the $\dCRThat$ will control Type-I error (Corollary~\ref{cor:dcrt-double-robustness}).

	\section{GCM test is optimal against certain alternatives} \label{sec:optimality}
	
	We have shown that, in large samples, the $\dCRThat$ has the same power against local alternatives as the resampling-free GCM test. Of course, other instances of the much more general CRT paradigm have better power than the GCM test against certain alternatives. We show in this section, however, that this is not the case for generalized partially linear models (GPLMs), a broad class of alternatives. In fact, the GCM test is asymptotically most powerful against GPLM alternatives. We leverage classical semiparametric efficiency theory \citep{Choi1996, VDV1998, Kosorok2008} to prove this result. We state our optimality result in Section~\ref{sec:optimality-result}, give an example of its application in Section~\ref{sec:kernel-ridge-regression}, and then compare it to existing semiparametric optimality results in Section~\ref{sec:semiparametric-discussion}.
	
	\subsection{Optimality result} \label{sec:optimality-result}
	
	To facilitate the link with semiparametric theory, in this section of the paper we operate in a fixed-dimensional setting. Accordingly, we drop the subscript $n$ from $\nulllaws_n$ and $\regclass_n$. For each value of $n$, we have $(\prx, \pry, \prz) \in \R^{1 + 1 + p}$ for fixed $p$. We will seek power against semiparametric GPLM alternatives of the form
	\begin{equation}
		\law_{\theta}(\prx, \pry, \prz) \equiv \law_{\beta, \eta}(\prx, \pry, \prz) \equiv \law_{x,z}(\prx, \prz) \times f_{\bm \eta}(\pry|\prx, \prz),\quad \bm \eta = \prx \beta + g(\prz).
		\label{eq:alternatives-1}
	\end{equation}
	Here, $\law_{x,z}$ is a fixed law, $f_\eta$ is a one-parameter exponential family with natural parameter $\eta \in \R$ and log-partition function $\psi$, $\beta \in \R$ and 
	\begin{equation}
		g \in \H_g \subseteq L^2(\law_{x,z}(\bm Z)), 
	\end{equation}
	where $\H_g$ is a linear subspace of the $L^2$ space of functions on $\R^p$ with the measure $\law_{x,z}(\bm Z)$. The alternatives~\eqref{eq:alternatives-1} are those where $\pry|\prx,\prz$ follows an exponential family distribution with natural parameter linear in $\prx$ and potentially nonlinear in $\prz$. Note that GPLMs include linear and generalized linear models as special cases, and therefore cover a broad range of alternative distributions.
	
	We focus on power against local alternatives $\law_{\theta_n(h)}$ near $\theta_0 \equiv (0, g_0)$, defined by 
	\begin{equation}
		\theta_n(h) \equiv \theta_n(h_\beta, h_g) \equiv (h_\beta/\sqrt n, g_0 + h_g/\sqrt n), \quad \text{for} \quad h \equiv (h_\beta, h_g) \in (0, \infty) \times \H_g.
		\label{eq:local-alternatives}
	\end{equation}
	We leave the dependence of $\theta_n(h)$ on $g_0$ implicit. Next, we define asymptotic optimality against such local alternatives following \citet{Choi1996}:
	\begin{definition}
		For $h \in (0, \infty) \times \H_g$, we say a test $\phi^*_n$ is the locally asymptotically most powerful level $\alpha$ test of 
		\begin{equation}
			H_{0}: \law \in  \regclass\subseteq \nulllaws \quad \text{versus}  \quad H_{1n}: \law = \law_{\theta_n(h)}
			\label{eq:testing-problem}
		\end{equation}
		if $\phi^*_n$ has asymptotic Type-I error control over $ \regclass$ at level $\alpha$ and for any other test $\phi_n$ satisfying the same property we have
		\begin{equation}
			\limsup_{n \rightarrow \infty}\  \E_{\law_{\theta_n(h)}}[\phi_n(\srx, \sry, \srz)] \leq \liminf_{n \rightarrow \infty}\ \E_{\law_{\theta_n(h)}}[ \phi^{*}_n(\srx, \sry, \srz)].
			\label{eq:power-bound}
		\end{equation}
		If this is true for every $h \in (0, \infty) \times \H_g$, such a test is locally asymptotically uniformly most powerful at $g_0$, or LAUMP($g_0$). A test is LAUMP($\mathcal{S}$) against $\law_{\theta_n(h)}$ for $h \in (0, \infty) \times \H_g$ if it is LAUMP($g_0$) for each $g_0 \in\mathcal{S}\subseteq \H_g$.
		
	\end{definition}
	
	Finally, define 
	\begin{equation}
		s^2(\theta_0) \equiv \E_{\law_{\theta_0}}[\V_{\law_{\theta_0}}[\prx|\prz]\V_{\law_{\theta_0}}[\pry|\prz]].
	\end{equation}
	We are now ready to state our main optimality result.
	
	\begin{theorem} \label{thm:optimality} 
		Consider the conditional independence testing problem~\eqref{eq:testing-problem}, with a collection of null distributions $\regclass\subseteq\nulllaws$ satisfying some regularity conditions, a linear subspace $\H_g \subseteq L^2(\law_{x,z}(\bm Z))$ specifying possible values for the nonparametric component $g$ in the GPLM alternative model~\eqref{eq:alternatives-1}, and some subset $\mathcal S \subseteq \H_g$. If the following four assumptions hold:
		\begin{align}
			\text{assumptions \eqref{eq:sp1} and \eqref{eq:sp2} hold for all } \law \in \regclass, \label{eq:sp-assumptions}\\
			\ddot{\psi} = K > 0 \text{ and } \E_{\law_{x,z}}[\prx^2] < \infty \text{ OR } \mathrm{supp}(\prx, \prz) \text{ is compact and } \H_g \subseteq C(\R^p), \label{eq:moment-assumptions}\\
			\E_{\law_{x,z}}[\prx|\ \cdot \ ] \in \H_g, \label{eq:conditional-expectation} \\
			\forall\ g_0\in\mathcal{S}, h_g \in \H_g, \ \law_{\theta_n(0, h_g)} \in \regclass \text{ for large enough } n, \label{eq:interior-point}
		\end{align}
		then $\phi_n^{\GCM}$ is LAUMP($\mathcal{S}$) against $\law_{\theta_n(h)}$ for $h \in (0, \infty) \times \H_g$, with 
		\begin{equation}
			\lim_{n \rightarrow \infty}\E_{\law_{\theta_n(h)}}[\phi_n^{\GCM}(\srx, \sry, \srz)] = 1 - \Phi(z_{1-\alpha} - h_\beta \cdot s(\theta_0)).
			\label{eq:power-of-gcm-test}
		\end{equation}
	\end{theorem}
	
	\noindent Let us discuss each of the four assumptions of Theorem~\ref{thm:optimality}:
	\begin{itemize}
		\item The assumption~\eqref{eq:sp-assumptions} is a set of regularity conditions on the null distributions $\regclass$. It is the same set of assumptions made by \citet{Shah2018} to ensure Type-I error control of the GCM test over $\regclass$, including the assumption that the conditional means $\mu_{n,x}$ and $\mu_{n,y}$ are fit accurately enough~\eqref{eq:sp1} and fairly mild moment assumptions~\eqref{eq:sp2}.
		\item The assumption~\eqref{eq:moment-assumptions} is a set of regularity conditions on the alternative distribution~\eqref{eq:alternatives-1}. These conditions are required for the semiparametric optimality theory to apply. These assumptions allow for GPLMs based on the normal distribution (assuming $\prx$ has second moment) or any other exponential family (assuming $(\prx, \prz)$ is compactly supported and the functions $g$ are continuous).
		\item The assumption~\eqref{eq:conditional-expectation} states that the conditional expectation $\prz \mapsto \E_{\law_{x,z}}[\prx|\prz]$ must belong to the subspace $\H_g$. It guarantees that the ``least favorable'' value of the nonparametric component $g$ is in the space $\H_g$, yielding the optimality of the GCM statistic.
		\item The assumption~\eqref{eq:interior-point} connects the semiparametric alternative hypothesis to the conditional independence null hypothesis. In some sense it requires $\law_{\theta_0} \equiv \law_{(0, g_0)}$ (derived from the semiparametric alternative distribution~\eqref{eq:alternatives-1}) to be an interior point of $\regclass$ (the conditional independence null) for each $g_0 \in \mathcal S$.
	\end{itemize}
	We give an example of when these assumptions hold in the next section.
	
	\subsection{Example: Kernel ridge regression} \label{sec:kernel-ridge-regression}
	We illustrate Theorem \ref{thm:optimality} with a kernel ridge regression example, borrowed from \citet[Section 4]{Shah2018}. Suppose the conditional expectations $\mu_x(\prz) \equiv \E_{\law}[\prx|\prz]$ and $\mu_y(\prz) \equiv \E_{\law}[\pry|\prz]$ satisfy $\mu_x, \mu_y \in \H_k$ for some reproducing kernel Hilbert space $(\H_k,\|\cdot\|_{\H_k})$ with reproducing kernel $k:\mathbb{R}\times\mathbb{R}\rightarrow\mathbb{R}$. In particular, we consider $\H_k \equiv W^{1,2}([0,1]) \subset L^2([0,1])$, i.e. the Sobolev space defined 
	\begin{align*}
		W^{1,2}([0,1]) \equiv \big\{f:[0,1]\rightarrow \mathbb{R} \mid f(0)=0,\ f \text{ is absolutely continuous with } \dot{f}\in L^2([0,1])\big\},
	\end{align*} 
	equipped with the inner product
	\begin{equation*}
		\langle f,g\rangle_{W^{1,2}([0,1])} \equiv \int_0^1\dot{f}(z)\dot{g}(z)\mathrm{d}z.
	\end{equation*}
	$W^{1,2}([0,1])$ is an RKHS with kernel $k(x,y)=\min\{x,y\}$ \citep[Example 12.16]{Wainwright2019}. Consider the kernel ridge estimators
	\begin{equation}
		\begin{split}
			\widehat \mu_x &\equiv \argmin{\mu_x \in W^{1,2}([0,1])}\bigg\{\frac{1}{n}\sum_{i=1}^n|X_i-\mu_x(Z_i)|^2+\lambda\|\mu_x\|_{W^{1,2}([0,1])}^2\bigg\};\\ 
			\widehat \mu_y &\equiv \argmin{\mu_y \in W^{1,2}([0,1])}\bigg\{\frac{1}{n}\sum_{i=1}^n|Y_i-\mu_y(Z_i)|^2+\lambda\|\mu_y\|_{W^{1,2}([0,1])}^2\bigg\},
			\label{eq:kernel-ridge-estimators}
		\end{split}
	\end{equation}
	with $\lambda$ tuned as described in \citet[Section 4]{Shah2018}. Using \citet[Theorem 11]{Shah2018}, the following result can be derived as a consequence of Theorem~\ref{thm:optimality}.
	
	\begin{corollary}\label{cor:RKHS_example}
		Fix $C>0$, and consider the following regularity class $\regclass \subseteq \nulllaws$:
		\begin{equation}
			\begin{split}
				\regclass \equiv \{&\law(\prx, \pry, \prz)=\law(\prz)\times \law(\prx|\prz)\times\law(\pry|\prx,\prz):\\
				&\law(\prz) = \textnormal{Unif}([0,1]),\ \law(\prx|\prz) = N(\mu_x(\prz), 1),\ \law(\pry|\prx, \prz) = N(\mu_{y}(\prz), 1), \\
				&\mu_x, \mu_y \in B_{W^{1,2}}(0, C)\},
			\end{split}
		\end{equation}
		where we define the $W^{1,2}([0,1])$ ball
		\begin{equation}
			B_{W^{1,2}}(0, C) \equiv \{f \in W^{1,2}([0,1]): \norm{f}_{W^{1,2}([0,1])} < C\}.
		\end{equation}
		Now, fix $\mu_{0x}, \mu_{0y} \in B_{W^{1,2}}(0, C)$ and for each $h=(h_\beta,h_g)\in(0,\infty)\times W^{1,2}([0,1])$ consider the set of local alternatives $\law_{\theta_n(h)}(\prx, \pry, \prz)$ given by 
		\begin{equation}
			\begin{split}
				&\law_{\theta_n(h)}(\prz) \equiv \textnormal{Unif}([0,1]); \\
				&\law_{\theta_n(h)}(\prx|\prz) \equiv N(\mu_{0x}(\prz), 1); \\
				&\law_{\theta_n(h)}(\pry|\prx, \prz) \equiv N(\prx h_\beta/\sqrt{n} + \mu_{0y}(\prz) + h_g(\prz)/\sqrt{n}, 1).
			\end{split}
		\end{equation}
		Then, the GCM test based on the kernel ridge estimators~\eqref{eq:kernel-ridge-estimators} is LAUMP($B_{W^{1,2}}(0, C)$) against alternatives $\law_{\theta_n(h)}$.
	\end{corollary}
	
	Hence, the GCM test based on kernel ridge regression does not just control Type-I error \citep[Theorem 11]{Shah2018}; it is also optimal against local alternatives.
	
	\subsection{Discussion of Theorem~\ref{thm:optimality}} \label{sec:semiparametric-discussion}
	
	Theorem~\ref{thm:optimality} states that the GCM test of \citet{Shah2018} is the optimal test of conditional independence against a broad class of semiparametric GPLM alternatives, including linear and generalized linear models. To our knowledge, it is the first result at the intersection of conditional independence testing and semiparametric optimality, although \citet{Shah2018} have already noted the connection between the GCM test and nonparametric estimation of the expected conditional covariance between $\prx$ and $\pry$ given $\prz$. Our result complements another line of work on minimax optimality for conditional independence testing \citep{Canonne2018, Neykov2021, Kim2021}. In the related model-X context, few optimality results are available. Two existing works show optimality statements based on likelihood ratio statistics; one in the context of the CRT \citep{Katsevich2020a} and the other in the context of model-X knockoffs \citep{Spector2022a}.
	
	Theorem~\ref{thm:optimality} closely parallels results on estimation in semiparametric regression \citep{Robinson1988, Bickel1993,Donald1994, Hardle2000, Robins2001, VanDeGeer2014, Ning2017, Jankova2018a, Chernozhukov2018}. It follows from \citet{Bickel1993, Robins2001} that the GCM statistic with the true conditional means $\mu_x$ and $\mu_y$ is the efficient score under the null hypothesis $\beta = 0$ in the context of GPLMs based on one-parameter exponential families with canonical link. Existing results on semiparametric optimality for hypothesis testing state that tests based on optimal estimators are themselves optimal \citep{Choi1996, VDV1998, Kosorok2008}.
	
	Despite the similarity between Theorem~\ref{thm:optimality} and existing semiparametric optimality results, we emphasize that this theorem is a statement about optimality for conditional independence testing rather than for semiparametric testing. The semiparametric model~\eqref{eq:alternatives-1} plays the role of the alternative distribution with respect to which power is evaluated, and need not hold under the null hypothesis. To bridge this gap, it suffices to find an open ball within the conditional independence null hypothesis containing the semiparametric null hypothesis~\eqref{eq:interior-point}. This allows us to reduce the conditional independence testing problem to a semiparametric testing problem, and therefore to leverage existing semiparametric optimality results (Appendix~\ref{sec:optimality-proofs}).
	
	Note that Theorem~\ref{thm:optimality} gives the power against local alternatives of the GCM test with $\mu_x$ and $\mu_y$ estimated in sample. This complements \citet[Theorem 8]{Shah2018}, where these authors compute the power of the GCM test against non-local alternatives by resorting to sample splitting, which is not required to show Type-I error control for the GCM test. This sample splitting is necessary under non-local alternatives to avoid Donsker conditions; using either sample splitting or Donsker conditions is also standard practice in the semiparametric literature. By contrast, we avoid sample splitting by exploiting the special structure of the conditional independence null and contiguity arguments to compute limiting power under local alternatives.
	
	While the Type-I error control results in Section~\ref{sec:dr-and-equivalence} are stated in the high-dimensional setting, Theorem~\ref{thm:optimality} is stated only for fixed-dimensional covariate vectors $\prz$. Indeed, semiparametric optimality theory is predominantly low-dimensional. A notable exception is the work of \citet{Jankova2018a}, which provides a semiparametric theory of estimation in high dimensions. Extending this theory to hypothesis testing is nontrivial, and beyond the scope of the current work. Nevertheless, proving optimality statements for conditional independence testing in high dimensions is an interesting direction for future work. We note in passing that high-dimensional results for lasso-based estimators often assume exact sparsity of the coefficient vector, which poses a problem for condition~\eqref{eq:interior-point} requiring the regularity class $\regclass$ to have interior points.
	
	Finally, we note that Theorem~\ref{thm:optimality} gives the optimality of the GCM statistic against alternative models for $\pry$ in which $\prx$ and $\prz$ do not interact. For alternatives where the conditional association between $\pry$ and $\prx$ is modified by $\prz$,  the GCM test will no longer be optimal. Variants of the CRT \citep{Zhong2021, Sesia2022}, model-X knockoffs \citep{Li2021b}, and the GCM test \citep{Lundborg2022} are designed to improve power in the presence of effect modification are available, although their optimality properties are not described. Optimal tests developed specifically for detecting interaction effects between $\bm X$ and $\bm Z$ (rather than main effects) may be constructed based on \citet{Vansteelandt2008}.
	
	\section{Finite-sample performance assessment} \label{sec:simulations}
		
	The results in the preceding sections are all asymptotic. In this section, we complement these results with a comprehensive simulation-based assessment of Type-I error and power in finite samples. Previous simulation-based assessments of the Type-I error of MX methods have come to differing conclusions: \citet{SetC17, Romano2019a, SetS19, Liu2022a} found broad robustness to misspecification of $\law_n(\prx|\prz)$ while \citet{Li2022} found such misspecifications to cause marked Type-I error inflation. We show that differences in the level of marginal association between $\prx$ and $\pry$ implied by the simulation design explain these discrepancies, and then use this insight to inform our own simulation design in Section~\ref{sec:sim-design}. Then, we present the results of our numerical simulations in Section~\ref{sec:sim-results}. Numerical simulation results and instructions to reproduce them are available at \url{https://github.com/Katsevich-Lab/symcrt-manuscript-v1}.	
	
	\subsection{Revisiting prior simulations of robustness}  \label{sec:sim-revisiting}
	
	The question of robustness of MX methods to the misspecification of $\law_n(\prx|\prz)$ has been investigated starting from the paper in which the model-X framework was originally proposed \citep{CetL16}. In this paper, the joint distribution $\law_n(\prx,\prz)$ was estimated in sample via the graphical lasso, which is similar to estimating the conditional distribution $\law_n(\prx|\prz)$ via the ordinary lasso. These authors found that
	\begin{quote}
		``Although the graphical Lasso is well suited for this problem since the covariates have a sparse precision matrix, its covariance estimate is still off by nearly 50\%, and yet surprisingly the resulting power and FDR are nearly indistinguishable from when the exact covariance is used...the nominal level of 10\% FDR is never violated, even for covariance estimates very far from the truth.''
	\end{quote}
	Similar conclusions have been drawn from numerical simulations in subsequent papers as well \citep{SetC17, Romano2019a, SetS19, Liu2022a}, the latter studying the dCRT specifically. On the other hand, the numerical simulations of \citet{Li2022} show that the dCRT can suffer significant Type-I error inflation when $\law_n(\prx|\prz)$ is inaccurately fit. These authors state that ``for model-X inference, the dependence of $\prx$ on $\prz$ is not adequately characterized and adjusted [for] due to the shrinkage bias of lasso.''
	
	To resolve this apparent contradiction, we consider a common data-generating model used in MX literature:
	\begin{equation}
		\law_n(\prx, \prz) = N(0, \Sigma), \quad \law_n(\pry|\prx,\prz) = N(\prx \theta + \prz^T \beta, \sigma^2_y).
		\label{eq:ar1-dgp}
	\end{equation}
	Often, $(\prx, \prz)$ are assumed to have a spatial structure (motivated by the GWAS application), with $\Sigma = \Sigma(\rho) \in \R^{(1+p) \times (1+p)}$ taken to be the AR(1) covariance matrix with autocorrelation parameter $\rho \in (-1,1)$. This covariance matrix roughly approximates linkage disequilibrium structure among genotypes, where correlations among variables are local with respect to the spatial structure. Conditional independence under this model~\eqref{eq:ar1-dgp} reduces to $H_0: \theta = 0$. Furthermore, the conditional distribution $\law_n(\prx|\prz)$ implied by the normal joint distribution is that of a linear model:
	\begin{equation}
		\text{Under } H_0, \quad  \law_n(\prx | \prz) = N(\prz^T \gamma, \sigma^2_x), \quad \law_n(\pry|\prz) = N(\prz^T \beta, \sigma^2_y).
	\end{equation}
	In the context of this model, the conditional independence testing problem is nontrivial to the extent that $\prz$ induces marginal association between $\prx$ and $\pry$ even in the absence of conditional association. In a causal inference context, this spurious marginal association would be called a confounding effect of $\prz$. This marginal association can be small or large, depending on the correlation structure of $\prz$ and the extent to which the supports of $\beta$ and $\gamma$ overlap. Properly adjusting for $\prz$ is important to the extent that $\prz$ induces marginal association between $\prx$ and $\pry$.
	
	We claim that the simulation studies in much of the original MX literature had relatively low levels of marginal association between $\prx$ and $\pry$, whereas the simulation studies in \citet{Li2022} were done in a regime with much more marginal association. To illustrate this point, we quantify the level of marginal association in a given problem setup as the Type-I error of the GCM test with intercept-only models for $\law_n(\prx|\prz)$ and $\law_n(\pry|\prz)$. This test is essentially a Pearson test of (marginal) independence between $\prx$ and $\pry$, and ignores the variables $\prz$ altogether. We compute this Type-I error for the data-generating models used to assess robustness by \citet{CetL16, Liu2022a, Li2022} (Appendix~\ref{sec:sim_liter}). The former two papers are framed in the variable selection context, where several explanatory variables $\bm W_j$ are considered, and the hypothesis $H_0: \pry \independent \bm W_j \mid \bm W_{\text{-}j}$ is tested for each $j$. Therefore, $\prx \equiv \bm W_j$ for each $j$. On the other hand, \citet{Li2022} considered a conditional independence testing framework, where $\prx$ was a single variable of interest.
	
	For the data-generating models used by \citet{CetL16, Liu2022a}, we evaluate the Type-I error of the marginal GCM test for each hypothesis $H_0: \pry \independent \bm W_j \mid \bm W_{\text{-}j}$, plotting these as a function of $j$ (Figure~\ref{fig:evaluation_typeI_err}, top row). We superimpose onto these plots a blue horizontal line indicating the Type-I error of the marginal GCM test for the data-generating model used by \citet{Li2022} (equal to 0.99, suggesting strong marginal association), and a red dashed horizontal line indicating the nominal level of this marginal test (equal to 0.05). The green ticks indicate the locations of the non-null variables. As expected for a setting where variable correlation is local, we see that Type-I error is inflated for null variables near the signal variables. The extent of this inflation depends on the autocorrelation parameter (set at 0.3 by \cite{CetL16} and 0.5 by \cite{Liu2022a}) and the locations of the signal variables. Most null variables, however, are not near signal variables, and therefore the marginal GCM test shows no inflation. This is reflected by the histograms of the Type-I error inflations (Figure~\ref{fig:evaluation_typeI_err}, bottom row). The median Type-I error of the marginal GCM test is near the nominal level of 0.05 in all three of the simulation setups from \citet{CetL16, Liu2022a}.
	
	\begin{figure}[!ht]
		\centering
		\includegraphics[width = \textwidth]{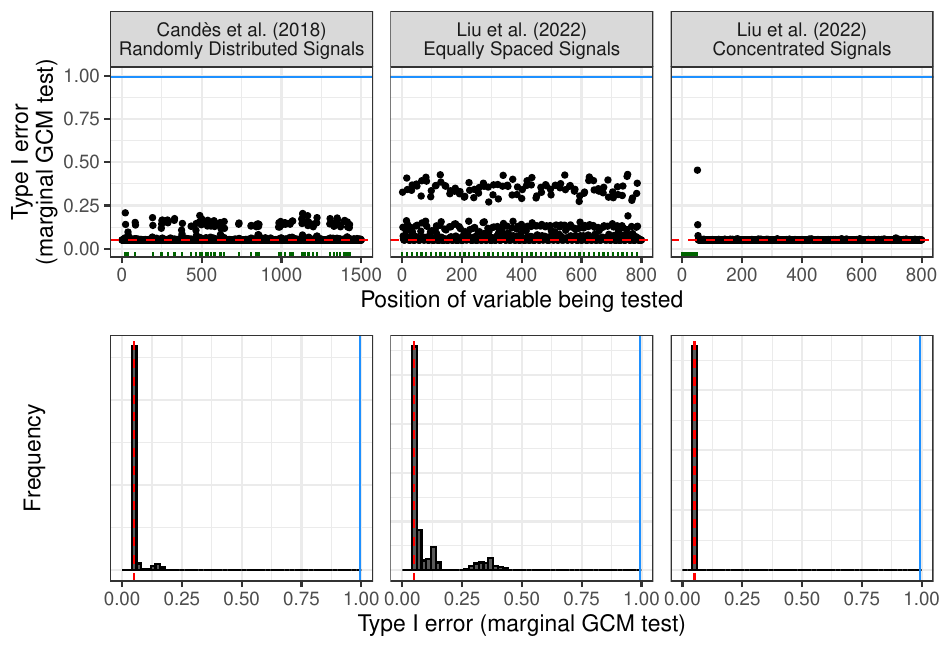}
		\caption{Comparing the marginal associations between $\prx$ and $\pry$ in the robustness simulations of \citet{CetL16, Liu2022a, Li2022} (Appendix~\ref{sec:sim_liter}). Top: Type-I error of the marginal GCM test as a function of the position of null variables with respect to the non-null variables (represented as green ticks). Bottom: Histograms of the Type-I error across null variables. The solid blue line indicates the Type-I error of the marginal GCM test for the robustness simulation of \citet{Li2022}, and the dashed red line the nominal Type-I error level of the marginal GCM test (0.05).}
		\label{fig:evaluation_typeI_err}
	\end{figure}

	\subsection{Simulation design} \label{sec:sim-design}
	
	\paragraph*{Data-generating model} As discussed in the previous section, appropriately setting the marginal correlation between $\prx$ and $\pry$ in a given data-generating model is crucial to properly evaluate the impact of inaccurate estimation of $\law_n(\prx|\prz)$ on the Type-I error control of a model-X method. Keeping this in mind, we propose the following data-generating model:
	\begin{equation}
		\law_n(\prz) = N(0, \Sigma(\rho)),\ \law_n(\prx | \prz) = N(\prz^T \beta, 1), \ \law_n(\pry|\prx,\prz) = N(\prx \theta + \prz^T \beta, 1).
		\label{eq:normal-dgm}
	\end{equation}
	We set the first $s$ coefficients of $\beta$ to be equal to $\nu$ and the rest to zero. Therefore, the entire data-generating process is parameterized by the six parameters $(n, p, s, \rho, \theta, \nu)$ (Table~\ref{tab:parameter-values}). For both null and alternative simulations, we vary each of the first four across five values each, setting the remaining three to the default value indicated in bold. The fifth parameter $\theta$ controls the signal strength and the sixth parameter $\nu$ controls the extent of marginal association between $\prx$ and $\pry$. For the null simulation, we set $\theta \equiv 0$, and for each setting of $(n, p, s, \rho)$, we choose five values of $\nu$ equally spaced between 0 (no marginal association) and $\nu_{\max}$ (computed so that the marginal GCM method has Type-I error 0.99). Note that $\nu_{\max}$ depends on the parameters $(n, p, s, \rho)$, so not exactly the same values of $\nu$ were used across settings of these four parameters. For the alternative simulation, we kept $\nu$ fixed at $\nu_{\max}/2$ while for each setting of $(n, p, s, \rho)$, we choose five values of $\theta$ equally spaced between 0 (no signal) and $\theta_{\max}$ (computed so that the GCM method with oracle settings of $\widehat \mu_{n,x}$ and $\widehat \mu_{n,y}$ has power 0.99). Finally, we complement the linear regression data-generating model~\eqref{eq:normal-dgm} with an analogous one based on logistic regression.
	
	\begin{table}
		\centering
		\begin{tabular}{ llll }
			$n$ & $p$ & $s$ & $\rho$ \\ 
			\hline
			100 & 100 & \textbf{5} & 0 \\
			\textbf{200} & 200 & 10 & 0.2 \\
			400 & \textbf{400} & 20 & \textbf{0.4} \\
			800 & 800 & 40 & 0.6 \\
			1600 & 1600 & 80 & 0.8
		\end{tabular}
		\hspace{0.5cm}
		\begin{tabular}{ ll }
			$\theta$ (null) & $\nu$ (null) \\ 
			\hline
			0 & 0  \\
			0 & $\nu_{\max}/4$ \\
			0 & $\nu_{\max}/2$ \\
			0 & $3\nu_{\max}/4$ \\
			0 & $\nu_{\max}$ 
		\end{tabular}
		\hspace{0.5cm}
		\begin{tabular}{ ll }
			$\theta$ (alt) & $\nu$ (alt) \\ 
			\hline
			0 & $\nu_{\max}/2$  \\
			$\theta_{\max}/4$ & $\nu_{\max}/2$ \\
			$\theta_{\max}/2$ & $\nu_{\max}/2$ \\
			$3\theta_{\max}/4$& $\nu_{\max}/2$ \\
			$\theta_{\max}$ & $\nu_{\max}/2$
		\end{tabular}
		\caption{The values of the sample size $n$, covariate dimension $p$, sparsity $s$, autocorrelation of covariates $\rho$, signal strength $\theta$, and marginal association strength $\nu$ used for the simulation study. Each of the parameters $n, p, s, \rho$ was varied among the values in the first table while keeping the other three at their default values, indicated in bold. For example, $p = 400, s = 5, \rho = 0.4$ were kept fixed while varying $n \in \{100, 200, 400, 800, 1600\}$. The second and third tables denote the values of $(\theta, \nu)$ used for the null and alternative simulations. Each combination of $(n, p, s, \rho)$ was paired with each of the five values of $(\theta, \nu)$ displayed for null and alternative simulations.}
		\label{tab:parameter-values}
	\end{table}
	
	\paragraph*{Methodologies compared}
	
	In Section~\ref{sec:dr-and-equivalence}, we found that the GCM test and the $\dCRThat$ are equivalent when applied with the same estimation methods for $\mu_{n,x}$ and $\mu_{n,y}$. Using this equivalence, we also showed that the $\dCRThat$ is robust to errors in $\widehat \mu_{n,x}$ if they are compensated for by accurate estimates $\widehat \mu_{n,y}$. In our simulation to assess Type-I error, we wish to probe the finite-sample Type-I error control of the GCM and the $\dCRThat$. We apply both of these methods with the lasso to estimate $\mu_{n,x}$ and $\mu_{n,y}$, as this is the most common choice in the MX literature.
	
	In addition to the GCM test and the $\dCRThat$, we apply the Maxway CRT \citep{Li2022}, designed specifically to improve the Type-I error control of the dCRT in the context when $\mu_{n,x}$ must be estimated. The Maxway CRT is inherently a semi-supervised method, assuming the existence of an auxiliary unlabeled dataset containing observations of $\prx$ and $\prz$ but not of $\pry$. The methodology (specifically, ``Maxway$_{\text{in}}$ example 1'') proceeds---roughly---by fitting $\lawhat_n(\prx|\prz)$ on the unlabeled data via the post-lasso (i.e. selecting active variables via the lasso and then refitting via ordinary least squares, \cite{Belloni2013}), fitting $\widehat \mu_{ny}(\prz)$ on the labeled data via post-lasso, and then applying dCRT on the labeled data based on these two models.
	
	Since the primary focus of this paper is the setting when no auxiliary unlabeled data are available, we implement the Maxway CRT by randomly splitting the data into two equal pieces, using the first as the unlabeled data (in particular, ignoring the response data) and the second as the labeled data. This strategy is consistent with the real data analysis in \citet[Section 6]{Li2022}. We also consider a bona-fide semi-supervised setup, in order to compare the GCM test and $\dCRThat$ to the Maxway CRT in the setting originally considered by \citet{Li2022}. However, in the semi-supervised setting we use all of the available data on $(\prx, \prz)$ (i.e. both unlabeled and labeled data) to fit $\law_n(\prx|\prz)$. By contrast, \citet{Li2022} used only the unlabeled data to learn $\law_n(\prx|\prz)$ in their implementation of the $\dCRThat$ for semi-supervised data.
	
	Finally, we noted in Section~\ref{sec:dr-and-equivalence} that the $\dCRThat$ already has a built-in doubly robust property. Therefore, we conjectured that the Type-I error inflation observed in the simulations of \citet{Li2022} is attributable to poor estimation of $\mu_n(\prx|\prz)$ and/or $\mu_n(\pry|\prz)$ and that the $\dCRThat$ can achieve Type-I error control if used in conjunction with better estimators of these conditional means. Taking inspiration from \citet{Li2022}, we also considered versions of the $\dCRThat$ and the GCM test based on the post-lasso in addition to those based on the usual lasso. In summary, we compared five methods: lasso and post-lasso based GCM, lasso and post-lasso based $\dCRThat$, and Maxway CRT (Table~\ref{tab:methods-compared}). As a point of reference for the null simulation, we also included the GCM test with intercept-only models for $\mu_{n,x}$ and $\mu_{n,y}$; the Type-I error of this test quantifies the degree of marginal association in the data-generating model (Section~\ref{sec:sim-revisiting}). As a point of reference for the alternative simulation, we also included the GCM test with $\mu_{n,x}$ and $\mu_{n,y}$ set to their ground truth values; the power of this test is the maximum power achievable by any test and therefore quantifies the signal strength in the data-generating model.
	
	\begin{table}
		\centering
		\begin{tabular}{ lllll }
			Method name & Estimating $\mu_{n,x}$ & Data for $\widehat \mu_{n,x}$ & Estimating $\mu_{n,y}$ & Data for $\widehat \mu_{n,y}$\\ 
			\hline
			GCM (LASSO) & lasso & all & lasso & all/labeled \\
			$\dCRThat$ (LASSO) & lasso & all & lasso & all/labeled \\
			GCM (PLASSO) & post-lasso & all & post-lasso & all/labeled \\
			$\dCRThat$ (PLASSO) & post-lasso & all & post-lasso & all/labeled \\
			Maxway CRT & post-lasso & unlabeled & post-lasso & labeled \\
			\hline
			GCM (marginal) & intercept-only & all & intercept-only & all/labeled \\
			GCM (oracle) & ground truth & -- & ground truth & --
		\end{tabular}
		\caption{The five methodologies compared, how they estimate $\mu_{n,x}$ and $\mu_{n,y}$, and what data they use for each in the context of semi-supervised or fully supervised data. Note that in the fully supervised case, data is split in half to form ``unlabeled'' and labeled sets for Maxway CRT. In this case, the $\dCRThat$ and GCM tests still use all of the data available for estimating $\mu_{n,x}$ and $\mu_{n,y}$. Two additional tests were used for reference purposes: the GCM test with intercept-only models for $\mu_{n,x}$ and $\mu_{n,y}$ and the GCM test with $\mu_{n,x}$ and $\mu_{n,y}$ set to their ground truth values.}
		\label{tab:methods-compared}
	\end{table}
	
	\paragraph*{Evaluation of power in the presence of Type-I error inflation}
	
	The methodologies compared control Type-I error to differing extents across the variety of simulation parameters in Table~\ref{tab:parameter-values}. This makes it challenging to compare power across methods, since some control Type-I error while others do not. To address this challenge, we chose to compare the power of the \textit{test statistics} underlying the methods, each under oracle calibration to ensure Type-I error control. Given the composite null, exact oracle calibration is computationally intractable. Therefore, we instead calibrated each test with respect to the point null given by 
	\begin{equation*}
		\law_n(\prz) = N(0, \Sigma(\rho)),\ \law_n(\prx | \prz) = N(\prz^T \beta, 1), \ \law_n(\pry|\prx,\prz) = N(\E[\prx|\prz]^T \theta + \prz^T \beta, 1).
	\end{equation*}
	This is the ``closest'' point in the null to the alternative~\eqref{eq:normal-dgm} under consideration; therefore ensuring Type-I error control at this point null should be a decent proxy for ensuring Type-I error control over the whole null. To calibrate two-sided tests with respect to this point null, we generate samples of a test statistic from the null and then define lower and upper critical values as the 2.5\% and 97.5\% quantiles of this distribution. Using potentially asymmetric lower and upper critical values is necessary, as the null distribution may not be symmetric and centered at zero \citep{Liu2022a}.
	
	\subsection{Simulation results} \label{sec:sim-results}
	
	We conducted simulations for Gaussian and binary models for the response $\pry$, each within the supervised and semi-supervised settings. We present the Type-I error and power for Gaussian responses in the supervised setting in Figures~\ref{fig:gaussian_supervised_partial_null} and~\ref{fig:gaussian_supervised_partial_alternative}, respectively, while deferring the other cases to Appendix~\ref{sec:additional-simulation-results}. Note also that for the sake of brevity Figures~\ref{fig:gaussian_supervised_partial_null} and~\ref{fig:gaussian_supervised_partial_alternative} only present three out of the five values for the four parameters $n, p, s, \rho$; the complete results are presented in Appendix~\ref{sec:additional-simulation-results}.
	
	Next we list the main conclusions regarding Type-I error based on the results in Figures~\ref{fig:gaussian_supervised_partial_null} (Gaussian supervised), \ref{fig:gaussian_semi-supervised_null} (Gaussian semi-supervised), \ref{fig:binomial_supervised_null} (binary supervised), and~\ref{fig:binomial_semi-supervised_null} (binary semi-supervised): 
	\begin{itemize}
		\item As one would expect, across all simulation settings, all methods have poorer Type-I error control as sample size $n$ decreases, dimension $p$ increases, number of nonzero coefficients $s$ increases, autocorrelation $\rho$ increases, or marginal association strength $\nu$ increases.
		\item For Gaussian responses, the $\dCRThat$ and GCM methods based on the same test statistics have very similar Type-I error control, echoing the asymptotic equivalence of the two methods (Theorem~\ref{thm:equivalence}). For binary responses, the lasso-based $\dCRThat$ has somewhat lower Type-I error than the lasso-based GCM test (Figure~\ref{fig:binomial_supervised_null}). The discreteness of binary responses likely slows down the convergence to normality of the GCM statistic, rendering the resampling-based null distribution of the $\dCRThat$ a better approximation to the null distribution.
		\item Across all simulation settings, the $\dCRThat$ and GCM methods based on the post-lasso have dramatically better Type-I error control than their lasso-based counterparts. This is because the post-lasso tends to more fully regress the confounders $\prz$ out of the response $\pry$; see also Appendix~\ref{sec:lasso-vs-plasso}.
		\item Across all simulation settings, Maxway CRT has better Type-I error control than the lasso-based $\dCRThat$ (in line with the results of \cite{Li2022}), but worse Type-I error control than the post-lasso-based $\dCRThat$. The latter is likely due to the fact that Maxway CRT uses only half of the available data on $(\prx, \prz)$ to fit $\law_n(\prx|\prz)$, and therefore does not adjust for $\prz$ as accurately.
	\end{itemize}
	
	\begin{figure}[h]
		\centering
		\includegraphics[width = \textwidth]{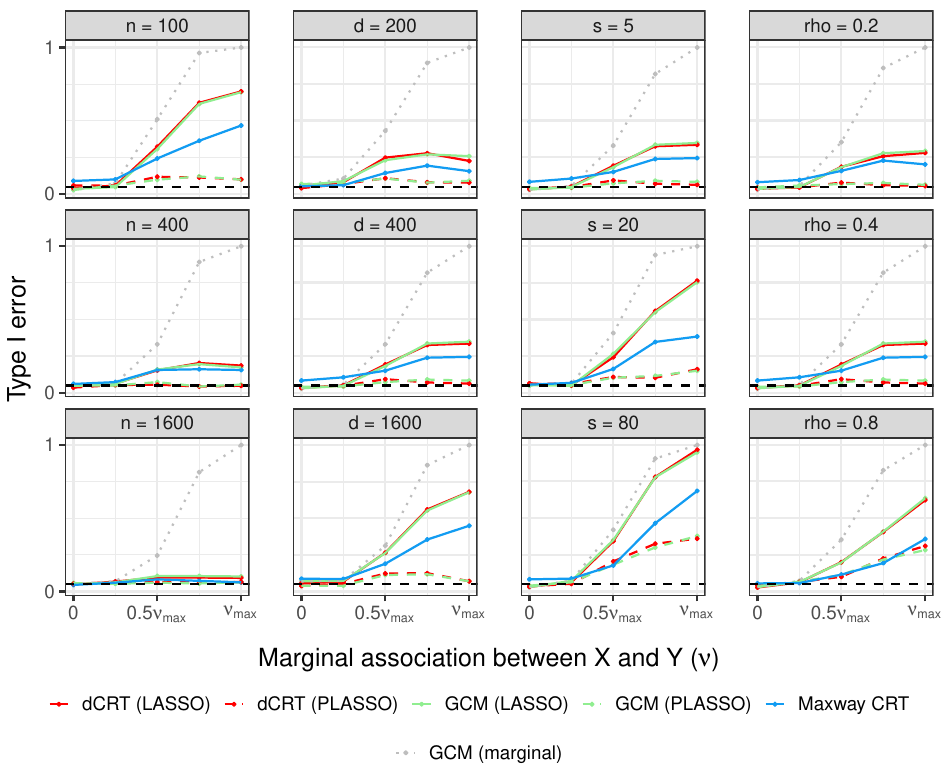}
		\caption{Type I error control for Gaussian supervised setting: we vary only one parameter in each column and there are five values of the marginal association strength $\nu$ in each subplot. Each point is the average of 400 Monte Carlo replicates.}
		\label{fig:gaussian_supervised_partial_null}
	\end{figure}
	
	Next, we list the main conclusions regarding power based on the results in Figures~\ref{fig:gaussian_supervised_partial_alternative} (Gaussian supervised), \ref{fig:gaussian_semi-supervised_alternative} (Gaussian semi-supervised), \ref{fig:binomial_supervised_alternative} (binary supervised), and~\ref{fig:binomial_semi-supervised_alternative} (binary semi-supervised): 
	\begin{itemize}
		\item Across all simulation settings, GCM-based methods have somewhat higher power than their $\dCRThat$-based methods. This may have to do with the stabilizing effect of the GCM normalization, compared to the unnormalized $\dCRThat$ statistic. The difference between the two tends to vanish as sample size grows, reflecting the asymptotic equivalence of the two methods (Corollary~\ref{cor:asymptotic-equivalence-alternative}).
		\item Across all simulation settings, the $\dCRThat$ and GCM methods based on the lasso have lower power than their post-lasso-based counterparts. This is because the post-lasso introduces more variance into the estimation of $\mu_{n,y}$; see also Appendix~\ref{sec:lasso-vs-plasso}.
		\item Across Gaussian and binary supervised simulation settings (Figures~\ref{fig:gaussian_supervised_alternative} and~\ref{fig:binomial_supervised_alternative}), Maxway CRT has the lowest power among all methods compared. The reason for this is that Maxway CRT relies on data splitting and therefore has half the effective sample size of the other methods. On the other hand, for semi-supervised settings (Figures~\ref{fig:gaussian_semi-supervised_alternative} and~\ref{fig:binomial_semi-supervised_alternative}), Maxway CRT has power comparable to or better than those of the post-lasso-based methods, but still worse than the lasso-based methods. This is due to the additional variance introduced by the refitting step in the post-lasso.
	\end{itemize}
	
	\begin{figure}[h]
		\centering
		\includegraphics[width = \textwidth]{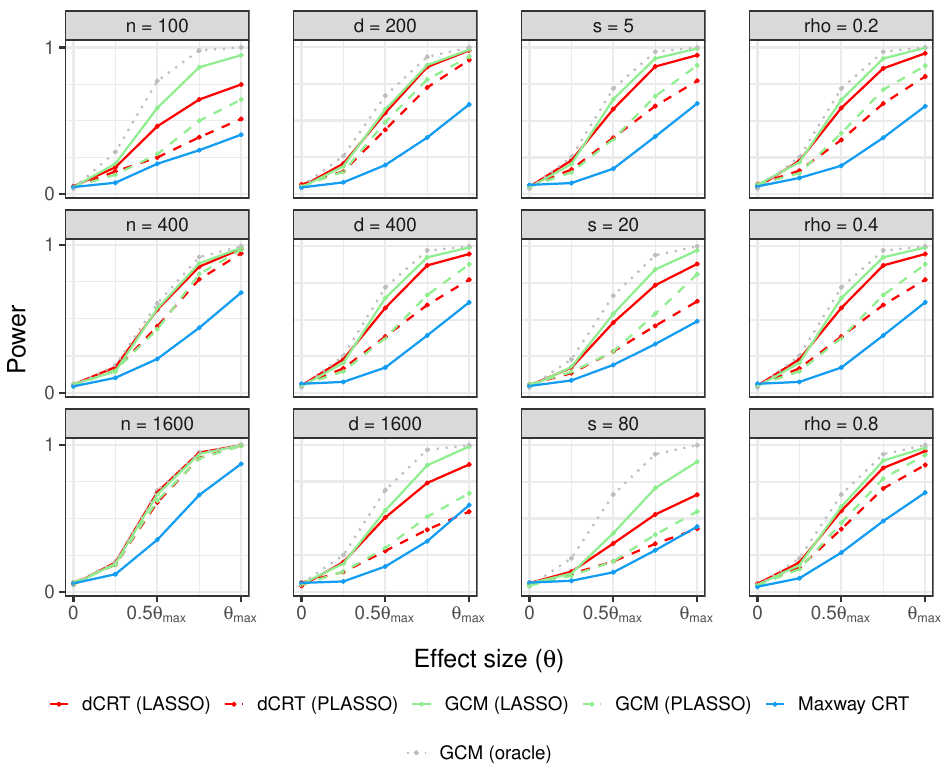}
		\caption{Power for Gaussian supervised setting: we vary only one parameter in each column and there are five values of the signal strength $\theta$ in each subplot. Each point is the average of 400 Monte Carlo replicates.}
		\label{fig:gaussian_supervised_partial_alternative}
	\end{figure}
	
	In summary, the methods with the best Type-I error control across all simulation settings are the $\dCRThat$ and the GCM test based on the post-lasso, although this improved robustness does come with a cost in terms of power when compared to the lasso-based methods. We investigate the associated trade-off in Appendix~\ref{sec:lasso-vs-plasso}.
	
	\section{Conclusion} \label{sec:conclusion}
	
	We conclude by summarizing our main findings and highlighting directions for future work.
	
	\paragraph*{Model-X inference with $\law(\prx|\prz)$ fit in sample can be doubly robust}
	
	Model-X inference \citep{CetL16} is presented as a mode of inference where the assumptions are transferred entirely from $\law(\pry|\prz)$ to $\law(\prx|\prz)$; no restrictions are made on the former law (or the test statistic used, at least in the context of the CRT), while the latter law is assumed exactly known. In practice, however, the law $\law(\prx|\prz)$ is often fit in sample. In the context of the dCRT, we show that Type-I error control cannot be guaranteed without restrictions on $\law(\pry|\prz)$ or the test statistic used (Section~\ref{sec:neg-results}). On the other hand, test statistics based on decent estimates of $\E[\pry|\prz]$ can compensate for errors in the estimation of $\law(\prx|\prz)$ and restore Type-I error control (Corollary~\ref{cor:dcrt-double-robustness}), a double robustness phenomenon. This result brings model-X inference more in line with double regression inferential methodologies: The conditional mean $\E[\prx|\prz]$ is estimated in the context of in-sample approximation to the ``model for X,'' and the conditional mean $\E[\pry|\prz]$ is estimated when computing the model-X test statistic. Relatedly, a double robustness property was noted for conditional model-X knockoffs \citep{Huang2019}. A doubly robust version of the dCRT has also been recently proposed (the Maxway CRT; \cite{Li2022}), although we argue that the original dCRT is itself doubly robust.
	
	\paragraph*{The GCM test has broadly similar Type-I error and power as the dCRT, but requires no resampling}
	
	When fitting $\law(\prx|\prz)$ in sample, the dCRT is essentially a double regression methodology. This prompts a comparison to the GCM test \citep{Shah2018}, another conditional independence test based on double regression. We established that the two tests are asymptotically equivalent under the null (Theorem~\ref{thm:equivalence}) and under arbitrary local alternatives (Corollary~\ref{cor:asymptotic-equivalence-alternative}). This suggests that the dCRT and the GCM test---when applied with the same estimators for $\E[\prx|\prz]$ and $\E[\pry|\prz]$---should have similar Type-I error control and power. Our numerical simulations (Section~\ref{sec:simulations}) largely confirm this behavior in finite samples. A possible exception to this conclusion is the case when small samples or discreteness in the data slows down the convergence of the $\dCRThat$ resampling distribution to normality (Theorem~\ref{thm:normal-limit}). In such cases, we observed that the $\dCRThat$ can in fact have better Type-I error control than the GCM based on the same estimators (Figure~\ref{fig:binomial_supervised_null}), presumably thanks to a better approximation to the null distribution in finite samples. Nevertheless, the broad similarity between the performances of the GCM test and the dCRT and the fact that the former test requires no resampling suggest that the GCM test may be preferable to the dCRT in practical problems with relatively large sample sizes.
	
	\paragraph*{The post-lasso yields much better Type-I error control than the lasso}
	
	Double robustness results for the GCM test and the dCRT apply only insofar as the estimation methods used in conjunction with these tests are accurate enough~\eqref{eq:sp1}. The default estimation method for $\E[\prx|\prz]$ and $\E[\pry|\prz]$ in many model-X applications is the lasso. As was demonstrated by \cite{Li2022}, the shrinkage bias of the lasso leads to inadequate adjustment of $\prx$ and $\pry$ for $\prz$, which in turn leads to inflated Type-I error. The same authors proposed the Maxway CRT, an extension of the dCRT involving the identification of coordinates of $\prz$ impacting $\prx$ and $\pry$ via the lasso followed by least squares refitting. Inspired by this work, we applied the original dCRT with post-lasso estimates for $\E[\prx|\prz]$ and $\E[\pry|\prz]$. We found vastly improved Type-I error control (Figure~\ref{fig:gaussian_supervised_null}), compared not just to the lasso-based dCRT but also to the Maxway CRT itself. The decreased bias of the post-lasso helps adjust for $\prz$ more fully, although we found that the extra variance incurred by refitting does come at a cost in power. Nevertheless, our results suggest that applying the post-lasso in conjunction with model-X methodologies can lead to significant improvements in robustness.
	
	\paragraph*{The GCM test is the optimal conditional independence test against alternatives without interactions between $\prx$ and $\prz$} 
	
	It is widely known in the semiparametric literature that the GCM test is the efficient score test for (generalized) partially linear models. The connection between the GCM test and semiparametric theory was noted briefly by \citet{Shah2018}, though not explored in depth; presumably because the GCM test is a conditional independence test rather than a test of a parameter in a semiparametric model. Nevertheless, we find that if the semiparametric \textit{null} hypothesis can be embedded within the conditional independence null hypothesis~\eqref{eq:interior-point}, semiparametric optimality theory can be carried over fairly directly to conditional independence testing to establish optimality against semiparametric \textit{alternative} distributions (Theorem~\ref{thm:optimality}). Thanks to this connection, we find that the GCM test has optimal asymptotic power among conditional independence tests against local generalized partially linear model alternatives~\eqref{eq:alternatives-1}. On the other hand, we leave open the question of optimality against alternatives where $\prx$ and $\prz$ are allowed to interact. We also leave open whether our optimality result can be extended to the high-dimensional regime.
	
	\paragraph*{Future work: The proportional regime and the variable selection problem} 
	
	Our results about the equivalence between the GCM test and the dCRT, and the double robustness of the latter, require estimates of $\E[\prx|\prz]$ and $\E[\pry|\prz]$ that are individually consistent and whose rates of convergence are sufficiently fast~\eqref{eq:sp1}. In the case of sparse linear models, we can get such rates if $\E[\prx|\prz]$ and $\E[\pry|\prz]$ depend on at most $s = o(\sqrt{n}/\log(p))$ of the coordinates of $\prz$. Such assumptions are common in other lines of work on high-dimensional / semiparametric / doubly-robust inference, including the debiased lasso \citep{VanDeGeer2014, ZZ14, Javanmard2014, Ning2017, Jankova2018a} and doubly-robust causal inference \citep{BetH14, Chernozhukov2018}. On the other hand, consistent estimates are typically not available in the regime when $n$, $p$, and $s$ grow proportionally \citep{Bayati2011}, causing a failure in traditional debiased estimates \citep{Celentano2021}. An additional limitation of the current work is that we do not directly consider the variable selection problem. For example, application of the GCM test to each variable is much more computationally costly than applying model-X knockoffs. Therefore, the comparison between model-X and doubly robust methodologies for variable selection purposes requires more thought.

\section*{Acknowledgments}

ZN was partially supported by the grant ``Statistical Software for Single Cell CRISPR Screens'' awarded to EK by Analytics at Wharton. OD was partially supported by FWO grant 1222522N and NIH grant AG065276. EK was partially supported by NSF DMS-2113072.  We acknowledge help from Timothy Barry with our simulation studies and the underlying computational infrastructure, including his \verb|simulatr| R package and Nextflow pipeline. We acknowledge dedicated support from the staff at the Wharton High Performance Computing Cluster. We acknowledge Lucas Janson for providing details about the simulation setting in \citet{CetL16}. We acknowledge Eric Tchetgen Tchetgen for helpful discussions on hypothesis testing in the semiparametric models.

\printbibliography


\appendix

\section{The $\dCRThat$ with GCM normalization} \label{sec:ndcrt}

As an alternative to the $\dCRThat$, we consider the $\ndCRThat$. This procedure is based on a normalized statistic that coincides exactly with the GCM statistic:
\begin{equation*}
T^{\ndCRThat}_n(\srx,\sry,\srz) \equiv \frac{1}{\widehat S_{n}^{\GCM}}T_n^{\dCRThat}(\srx, \sry, \srz) \equiv T^{\GCM}_n(\srx, \sry, \srz).
\end{equation*}
The only difference with GCM is that the critical value is given by conditional resampling rather than a normal quantile:
\begin{equation}
	\begin{split}
		\phi^{\ndCRThat}_n(\srx, \sry, \srz) &\equiv \indicator(T_n^{\ndCRThat}(\srx, \sry, \srz) > \Q_{1-\alpha}[T_n^{\ndCRThat}(\srxk, \srx, \sry, \srz) \mid \srx, \sry, \srz]) \\
		&\equiv \indicator\left(T_n^{\ndCRThat}(\srx, \sry, \srz) > C^{\ndCRThat}_n(\srx, \sry, \srz)\right).
	\end{split}
\end{equation}
Here, 
\begin{equation*}
T_n^{\ndCRThat}(\srxk, \srx, \sry, \srz) \equiv \frac{T_n^{\dCRT}(\srxk, \srx, \sry, \srz)}{S_n^{\GCM}(\srxk, \srx, \sry, \srz)},
\end{equation*}
where 
\begin{equation*}
(S_n^{\GCM}(\srxk, \srx, \sry, \srz))^2 \equiv \widehat{\V}\{(\srxk_i - \widehat \mu_{n,x}(\srz_i))(\sry_i - \widehat \mu_{n,y}(\srz_i))\}
\end{equation*}
and $T_n^{\dCRT}(\srxk, \srx, \sry, \srz)$ is as defined in equation~\eqref{eq:resampled-dcrt-def}.

\begin{theorem}\label{thm:normal-limit-ndcrt} 
	Let $\law_n$ be a sequence of laws such that the nondegeneracy conditions~\eqref{eq:var-bounded-below} and~\eqref{eq:variance-bounds} and the conditional Lyapunov condition~\eqref{eq:lyapunov-condition} hold. Then,
	\begin{equation}
		T^{\ndCRThat}_n(\srxk, \srx, \sry,\srz) \mid \srx, \sry, \srz \convdp N(0,1)
		\label{eq:conditional-convergence-ndcrt}
	\end{equation}
	and therefore
	\begin{equation}
		C^{\ndCRThat}_n(\srx, \sry, \srz) \equiv \Q_{1-\alpha}\left[T^{\ndCRThat}_n(\srxk, \srx, \sry,\srz) \mid \srx, \sry, \srz\right] \rightarrow z_{1-\alpha}
		\label{eq:critical-value-convergence-ndcrt}
	\end{equation}
	and the $\ndCRThat$ and GCM tests are equivalent:
	\begin{equation}
	\lim_{n \rightarrow \infty} \P_{\law_n}[\phi_n^{\ndCRThat}(\srx, \sry, \srz) = \phi_n^{\GCM}(\srx, \sry, \srz)] = 1.
	\label{eq:ndcrt-gcm-equivalence}
	\end{equation}
\end{theorem}

\begin{corollary}
\label{cor:double-robustness-ndcrt}
Let $\regclass_n$ be a sequence of regularity conditions such that for any sequence $\law_n \in \regclass_n$, we have the the nondegeneracy conditions~\eqref{eq:var-bounded-below} and~\eqref{eq:variance-bounds}, the conditional Lyapunov condition~\eqref{eq:lyapunov-condition}, and the assumptions~\eqref{eq:sp1} and~\eqref{eq:sp2}. Then, the $\ndCRThat$ has asymptotic Type-I error control over $\nulllaws_n \cap \regclass_n$ in the sense of the definition~\eqref{eq:asymptotic-control}.
\end{corollary}
Comparing Corollary~\ref{cor:double-robustness-ndcrt} to Corollary~\ref{cor:dcrt-double-robustness}, we see that the $\ndCRThat$ controls Type-I error under weaker assumptions than required for the $\dCRThat$; in particular the variance of $\law_n(\prx|\prz)$ need not be estimated with any degree of accuracy.

\section{Conditional convergence results} \label{sec:conditional-convergence-results}

The proofs of our theoretical results rely on the conditional counterparts of several standard convergence theorems. In this section, we state these conditional convergence theorems. We defer their proofs to Appendix~\ref{sec:conditional-convergence-proofs}.

First we define a notion of conditional convergence in probability, analogous to our definition of conditional convergence in distribution (Definition~\ref{def:conditional-convergence-distribution}).

\begin{definition} \label{def:conditional-convergence-probability}
	For each $n$, let $W_n$ be a random variable and let $\mathcal F_n$ be a $\sigma$-algebra. Then, we say $W_n$ converges in probability to a constant $c$ conditionally on $\mathcal F_n$ if $W_n$ converges in distribution to the delta mass at $c$ conditionally on $\mathcal F_n$ (recall Definition~\ref{def:conditional-convergence-distribution}). We denote this convergence by ${W_n \mid \mathcal F_n \convpp c}$. In symbols, 
	\begin{equation}
		W_n \mid \mathcal F_n \convpp c \quad \text{if} \quad W_n \mid \mathcal F_n \convdp \delta_c.
	\end{equation}
\end{definition}

\noindent Now we are ready to state the conditional convergence results. 

\subsection{Statements}

For the sake of all results below, let $\mathcal F_n$ be a sequence of $\sigma$-algebras.

\begin{theorem}[Conditional Polya's theorem]\label{thm:cond_polya} 
	Let $W_n$ be a sequence of random variables. If $W_n \mid \mathcal F_n \convdp W$ for some random variable $W$ with continuous CDF, then
	\begin{equation}
		\sup_{t \in \R}|\P[W_n \leq t \mid \mathcal F_n] - \P[W \leq t]| \convp 0.
	\end{equation}

\end{theorem}

\begin{theorem}[Conditional Slutsky's theorem]\label{thm:cond_slutsky}
	Let $W_n$ be a sequence of random variables. Suppose $a_n$ and $b_n$ are sequences of random variables such that $a_n \convp 1$ and $b_n \convp 0$. If $W_n \mid \mathcal F_n \convdp W$ for some random variable $W$ with continuous CDF, then
	\begin{equation}
		a_n W_n + b_n \mid \mathcal F_n \convdp W.
	\end{equation}
	
\end{theorem}

\begin{theorem}[Conditional law of large numbers]\label{thm:wlln_cond} 
	Let $W_{in}$ be a triangular array of random variables, such that $W_{in}$ are independent conditionally on $\mathcal F_n$ for each $n$. If for some $\delta > 0$ we have
	\begin{equation}
		\frac{1}{n^{1+\delta}}\sum_{i=1}^n\E[|W_{in}|^{1+\delta}\mid\mathcal{F}_n] \convp 0,
		\label{eq:wlln_cond_assumption}
	\end{equation}
	then 
	\begin{equation}
		\frac{1}{n} \sum_{i = 1}^n (W_{in} - \E[W_{in}|\mathcal{F}_n]) \mid \mathcal F_n \convpp 0.
		\label{eq:wlln_cond_conclusion}
	\end{equation}
	The condition~\eqref{eq:wlln_cond_assumption} is satisfied when
	\begin{equation}
		\sup_{1\leq i\leq n}\E[|W_{in}|^{1+\delta} \mid \mathcal{F}_n]=o_p(n^{\delta}).
		\label{eq:wlln_cond_sufficient}
	\end{equation}
\end{theorem}

As a corollary of Theorem~\ref{thm:wlln_cond}, if we choose $\mathcal{F}_n=\{\varnothing,\Omega\}$, we are able to obtain the following version of the weak law of large numbers for triangular arrays.

\begin{corollary}[Unconditional weak law of large numbers] \label{cor:wlln} 
	Let $W_{in}$ be a triangular array of random variables, such that $W_{in}$ are independent for each $n$. If for some $\delta > 0$ we have
	\begin{equation}
		\frac{1}{n^{1+\delta}}\sum_{i=1}^n\E[|W_{in}|^{1+\delta}] \rightarrow 0,
		\label{eq:one-plus-delta-moment-condition}
	\end{equation}
	then 
	\begin{equation}
		\frac{1}{n} \sum_{i = 1}^n (W_{in} - \E[W_{in}]) \convp 0.
	\end{equation}
	The condition~\eqref{eq:one-plus-delta-moment-condition} is satisfied when
	\begin{equation}
		\sup_{1\leq i\leq n}\E[|W_{in}|^{1+\delta}]=o(n^{\delta}).
		\label{eq:wlln_sufficient}
	\end{equation}
\end{corollary}

\begin{theorem}[Conditional central limit theorem] \label{thm:conditional-clt} 
	Let $W_{in}$ be a triangular array of random variables, such that for each $n,W_{in}$ are independent conditionally on $\mathcal F_n$. Define
	\begin{equation}
		S_n^2 \equiv \sum_{i = 1}^n \V[W_{in} \mid \mathcal F_n],
	\end{equation} 
	and assume $0 < \mathrm{Var}[W_{in}|\mathcal{F}_n]<\infty $ almost surely for all $i=1,\ldots,n$ and for all $n\in\mathbb{N}$. If for some $\delta > 0$ we have
	\begin{equation}
		\frac{1}{S_n^{2+\delta}} \sum_{i = 1}^n \E[|W_{in}-\E[W_{in}|\mathcal{F}_n]|^{2+\delta} \mid \mathcal{F}_n] \convp 0,
		\label{eq:conditional-lyapunov}
	\end{equation}
	then 
	\begin{equation}
		\frac{1}{S_n} \sum_{i = 1}^n (W_{in} - \E[W_{in} \mid \mathcal{F}_n]) \mid \mathcal F_n \convdp N(0,1).
	\end{equation}
	
\end{theorem}

\begin{lemma}[Conditional convergence implies quantile convergence] \label{lem:conditional-convergence-to-quantile} 
	Let $W_n$ be a sequence of random variables and $\alpha \in (0,1)$. If $W_n \mid \mathcal F_n \convdp W$ for some random variable $W$ whose CDF is continuous and strictly increasing at $\mathbb Q_{\alpha}[W]$, then
	\begin{equation}
		\Q_{\alpha}[W_n \mid \mathcal F_n] \convp \Q_{\alpha}[W].
	\end{equation}
\end{lemma}

\subsection{Discussion}

The above definitions and results on conditional convergence are not particularly surprising, and related results are present in the existing literature. Nevertheless, we have not found any of the above results stated in the literature in exactly this form. Here we discuss the relationships of our definitions and results with existing ones.

Notions of conditional convergence in probability and in distribution have been explicitly defined by \citet{Nowak2005}. However, these notions require a single conditioning $\sigma$-algebra as well as almost sure convergences of conditional probabilities, whereas in Definitions~\ref{def:conditional-convergence-distribution} and~\ref{def:conditional-convergence-probability} we allow the conditioning $\sigma$-algebra to change with $n$ and for the conditional probabilities to converge in probability. Our Definition~\ref{def:conditional-convergence-distribution} can be viewed as formalizing the notion of conditional convergence in distribution implicitly used by \citet{Wang2020b}. Related notions of conditional convergence in distribution allowing for changing conditioning $\sigma$-algebra are present implicitly in the works of \citet{Dedecker2002} and \citet{bulinski2017conditional}, though these are based on the convergence of conditional characteristic functions as opposed to conditional cumulative distribution functions. 

Turning to the convergence results themselves, we were not able to find conditional Polya's theorem (Theorem~\ref{thm:cond_polya}) in the literature. Conditional Slutsky's theorem (Theorem~\ref{thm:cond_slutsky}) is a generalization of \citet[Lemma 5]{Wang2020b} to the case when $a_n$ is not necessarily independent of $\mathcal F_n$ and $b_n \neq 0$. Versions of the conditional law of large numbers are given by \cite{Majerek2005a} and \cite{PrakasaRao2009}, but these involve a single conditioning $\sigma$-algebra and do not allow for triangular arrays, unlike Theorem~\ref{thm:wlln_cond}. Remarkably, we could not find even the unconditional triangular array law of large numbers (Corollary~\ref{cor:wlln}) in the literature; existing results either assume a second-moment condition or use truncation \citep[Theorems 2.2.4 and 2.2.6, respectively]{Durrett2010} instead of a $1+\delta$ moment condition or are not applicable to triangular arrays \citep[Lemma 19]{Shah2018}. As for central limit theorems, \citet{Grzenda2008, PrakasaRao2009, Yuan2014} give non-triangular array versions of the conditional central limit theorem that require a single conditioning $\sigma$-algebra, unlike Theorem~\ref{thm:conditional-clt}. Versions of the conditional central limit theorem appropriate for varying conditioning $\sigma$-algebras and triangular arrays are given by \citet{Dedecker2002, bulinski2017conditional}, those these involve different notions of conditional convergence in distribution. Results similar to Theorem~\ref{thm:conditional-clt} for $\delta = 1$ are presented in a recent line of work on sample-splitting-based inference \citep{Kim2020a, Shekhar2022a, Shekhar2022}; these can be proved via the Berry-Esseen theorem. Finally, we note that our result that conditional convergence in distribution implies in-probability quantile convergence (Lemma~\ref{lem:conditional-convergence-to-quantile}) is a generalization of \citet[Lemma 3]{Wang2020b} to general conditioning $\sigma$-algebras.

\section{Proofs for Section~\ref{sec:conv-to-normal}}

\subsection{Proofs of main results}

\begin{proof}[Proofs of Theorem~\ref{thm:normal-limit} and Corollary~\ref{cor:asymptotic-equivalence}] 
We prove instead the stronger Theorem~\ref{thm:normal-limit-stronger} and Corollary~\ref{cor:asymptotic-equivalence-stronger} below.
\end{proof}

\begin{theorem} \label{thm:normal-limit-stronger}
	Let $\law_n$ be a sequence of laws and $\lawhat_n$ be a sequence of estimates. Suppose either of the following two sets of assumptions is satisfied: 
	\begin{enumerate}
		\item The nondegeneracy conditions ~\eqref{eq:var-bounded-below} and~\eqref{eq:variance-bounds} and the conditional Lyapunov condition~\eqref{eq:lyapunov-condition} hold.
		\item The assumptions of Theorem~\ref{thm:equivalence} hold.
	\end{enumerate}
	Then, the normalized $T_n^{\dCRThat}(\srxk, \srx, \sry, \srz)$ converges conditionally to a normal distribution~\eqref{eq:conditional-convergence} and therefore the $\dCRThat$ critical value $C^{\dCRThat}_n(\srx, \sry, \srz)$ converges to $z_{1-\alpha}$~\eqref{eq:critical-value-convergence}.
\end{theorem}

\begin{proof}

It suffices to prove the conditional convergence in distribution~\eqref{eq:conditional-convergence}, as the convergence of the critical value~\eqref{eq:critical-value-convergence} follows from Lemma~\ref{lem:conditional-convergence-to-quantile} because the normal distribution has continuous and strictly increasing CDF. We prove the conditional convergence~\eqref{eq:conditional-convergence} for each of the two sets of assumptions.

\paragraph*{Assumption 1}

We proceed by applying the conditional CLT (Theorem~\ref{thm:conditional-clt}) with 
\begin{equation}
W_{in} \equiv \frac{1}{\sqrt{n}}(\srxk_i - \widehat \mu_{n,x}(\srz_i))(\sry_i - \widehat \mu_{n,y}(\srz_i))
\end{equation}
and $\mathcal F_n \equiv \sigma(\srx, \sry, \srz)$. To verify the assumptions of the conditional CLT, note first that $W_{in}$ are independent conditionally on $\mathcal F_n$ by construction and satisfy $0 < \V[W_{in} \mid \mathcal F_n] < \infty$ by the nondegeneracy assumption~\eqref{eq:variance-bounds}. Next, recalling definition~\eqref{eq:conditional-variance-def}, we have
\begin{equation*}
S^2_n \equiv \sum_{i = 1}^n \V_{\lawhat_n}[W_{in}\mid \srx,\sry,\srz] \equiv (\widehat S_n^{\dCRThat})^2,
\end{equation*}
so that
\begin{align*}
		&\frac{1}{S_n^{2+\delta}} \sum_{i = 1}^n \E_{\lawhat_n}[|W_{in}|^{2+\delta} \mid \srx,\sry,\srz] \\
		&\quad= \frac{1}{(\widehat S_n^{\dCRThat})^{2+\delta}} \cdot \frac{1}{n^{1+\delta/2}} \sum_{i=1}^n |\sry_i-\widehat\mu_{n,y}(\srz_i)|^{2+\delta}\E_{\lawhat_n}\left[\left|\srxk_i-\widehat\mu_{n,x}(\srz_i)\right|^{2+\delta}\mid \srx,\srz\right].
\end{align*}
This quantity converges to zero in probability due to the nondegeneracy condition~\eqref{eq:var-bounded-below} and the Lyapunov condition~\eqref{eq:lyapunov-condition}. Hence, the conditional CLT gives the desired conditional convergence~\eqref{eq:conditional-convergence}.
 
 \paragraph*{Assumption 2}

We begin by decomposing $T_n^{\dCRThat}(\srxk, \srx, \sry, \srz)$:
 \begin{align}
 	T_n^{\dCRThat}(\srxk, \srx, \sry, \srz) &= \frac{1}{\sqrt n} \sum_{i = 1}^n (\srxk_i - \widehat \mu_{n,x}(\srz_i))(\sry_i - \mu_{n,y}(\srz_i)) \\
 	&\quad -\frac{1}{\sqrt n}\sum_{i = 1}^n (\srxk_i - \widehat \mu_{n,x}(\srz_i))(\widehat \mu_{n,y}(\srz_i) - \mu_{n,y}(\srz_i))\\
	&\equiv I_n - J_n.
 \end{align}
We claim that $J_n \convp 0$. Indeed, from $\E[J_n \mid \srx,\sry,\srz] = 0$ and the assumption $\widehat E'_{n, y} \convp 0$~\eqref{eq:sp1prime} it follows
\begin{equation*}
\E[J_n^2 \mid \srx,\sry,\srz] = \V[J_n \mid \srx,\sry,\srz] \equiv (\widehat E'_{n,y})^2 \convp 0. 
\end{equation*}
Hence by Lemma \ref{lem:conditional-expectation-to-unconditional} we have $J_n^2 \convp 0$, so that $J_n \convp 0$, as claimed. Next, we claim that an appropriately rescaled $I_n$ converges conditionally to $N(0,1)$. To this end, we apply the conditional CLT (Theorem~\ref{thm:conditional-clt}) with
 \begin{equation}
 	W_{in} \equiv \frac{1}{\sqrt{n}}(\srxk_i - \widehat \mu_{n,x}(\srz_i))(\sry_i - \mu_{n,y}(\srz_i))
 \end{equation}
 and $\mathcal{F}_n \equiv \sigma(\srx,\sry,\srz)$. To verify the assumptions of the conditional CLT, note first that $W_{in}$ are independent conditionally on $\mathcal F_n$ by construction and satisfy $0 < \V[W_{in} \mid \mathcal F_n] < \infty$ by the nondegeneracy assumption~\eqref{eq:variance-bounds}. Next, observe that
 \begin{equation*}
	S^2_n \equiv \sum_{i = 1}^n \V[W_{in}\mid \srx,\sry,\srz] =  \frac{1}{n}\sum_{i = 1}^n \V_{\lawhat_n}[\srx_i|\srz_i](\sry_i - \mu_{n,y}(\srz_i))^2 \equiv(S_n^{\dCRThat})^2,
 \end{equation*}
so that
\begin{align*}
	&\frac{1}{S_n^{2+\delta}} \sum_{i = 1}^n \E[|W_{in}|^{2+\delta} \mid \srx,\sry,\srz] \\
	&\quad= \frac{1}{(S_n^{\dCRThat})^{2+\delta}} \cdot \frac{1}{n^{1+\delta/2}} \sum_{i=1}^n |\sry_i-\mu_{n,y}(\srz_i)|^{2+\delta}\E_{\lawhat_n}[|\srxk_i-\widehat\mu_{n,x}(\srz_i)|^{2+\delta}\mid \srx,\srz].
\end{align*}
Since the first factor is stochastically bounded (conclusion~\eqref{eq:bounded-away-from-zero} from Lemma~\ref{lem:cond-sample-variance-equivalence2}), it suffices to show that the second factor converges to zero in probability. To this end, by Lemma~\ref{lem:conditional-expectation-to-unconditional} it suffices to note that $\law_n \in \nulllaws_n$ and the Lyapunov assumption~\eqref{eq:lyapunov-condition-2} give
\begin{equation}
\begin{split}
&\E\left[\left.\frac{1}{n^{1+\delta/2}} \sum_{i=1}^n |\sry_i-\mu_{n,y}(\srz_i)|^{2+\delta}\E_{\lawhat_n}[|\srxk_i-\widehat\mu_{n,x}(\srz_i)|^{2+\delta}\mid \srx,\srz] \right| \srx, \srz \right] \\
&\quad= \frac{1}{n^{1+\delta/2}} \sum_{i=1}^n \E_{\law_n}[|\sry_i-\mu_{n,y}(\srz_i)|^{2+\delta}|\srz_i]\E_{\lawhat_n}[|\srxk_i-\widehat\mu_{n,x}(\srz_i)|^{2+\delta}\mid \srx,\srz] \\
&\quad \convp 0.
\end{split}
\end{equation}
Therefore, we may apply the conditional CLT to obtain that 
\begin{equation*}
\frac{1}{S_{n}^{\dCRThat}}I_n \mid \srx, \sry, \srz \convdp N(0,1).
\end{equation*}
Furthermore, equation~\eqref{eq:dcrt-var-equivalence} from Lemma~\ref{lem:cond-sample-variance-equivalence2} gives $S_n^{\dCRThat}/\widehat S_n^{\dCRThat} \convp 1$, so by conditional Slutsky's theorem (Theorem~\ref{thm:cond_slutsky}) we conclude that
\begin{equation}
\frac{1}{\widehat S_{n}^{\dCRThat}}T_n^{\dCRThat}(\srxk, \srx, \sry, \srz) = \left.\frac{S_n^{\dCRThat}}{\widehat S_n^{\dCRThat}}\left(\frac{I_n}{S_{n}^{\dCRThat}} - \frac{J_n}{S_{n}^{\dCRThat}}\right)\right| \srx, \sry, \srz \convdp N(0,1),
\end{equation}
as desired.
\end{proof}

\begin{corollary}
	\label{cor:asymptotic-equivalence-stronger}
Under the assumptions of Theorem \ref{thm:normal-limit-stronger}, if the non-accumulation condition~\eqref{eq:non-accumulation} is satisfied then the $\dCRThat$ is asymptotically equivalent to the $\MXtwohat$ $F$-test.
\end{corollary}
\begin{proof}
The desired equivalence is given by Lemma~\ref{lem:equivalence-lemma} with $T_n(\srx, \sry, \srz) \equiv (\widehat S_{n}^{\dCRThat})^{-1} T_n^{\dCRThat}$ and $C_n(\srx, \sry, \srz) \equiv C_n^{\dCRThat}(\srx, \sry, \srz)$, since $\phi_n^1$ and $\phi_n^2$ in the lemma statement reduce to $\dCRThat$ and the $\MXtwohat$ $F$-test, respectively. This lemma is applicable because the convergence of the critical value~\eqref{eq:convergence-of-critical-value} is given by Theorem~\ref{thm:normal-limit-stronger} and the non-accumulation condition~\eqref{eq:non-accumulation-app} is assumed.
\end{proof}

\subsection{Auxiliary lemmas}

\begin{lemma}
	\label{lem:conditional-expectation-to-unconditional}
	Let $W_n$ be a sequence of nonnegative random variables and let $\mathcal F_n$ be a sequence of $\sigma$-algebras. If $\E[W_n \mid \mathcal{F}_n]\convp 0$, then $ W_n \convp 0$.
\end{lemma}
\begin{proof}
	For any $\epsilon>0$, we have
	\begin{align}
		\P[W_n \geq \epsilon] &= \P[W_n \wedge \epsilon \geq \epsilon] \\
		&\leq \epsilon^{-1} \E[W_n \wedge \epsilon] \\
		&= \epsilon^{-1} \E[\E[W_n \wedge \epsilon\mid \mathcal{F}_n]]\\
		&\leq \epsilon^{-1} \E[\E[W_n \mid \mathcal{F}_n]\wedge \epsilon] \to 0,
	\end{align}
	where the last convergence is due to bounded convergence theorem and the assumption $\E[W_n \mid \mathcal{F}_n]\convp 0$. 
\end{proof}

\begin{lemma}[Asymptotic equivalence of tests] \label{lem:equivalence-lemma}
	
Consider two hypothesis tests based on the same test statistic $T_n(\srx, \sry, \srz)$ but different critical values:
\begin{equation*}
\phi_n^1(\srx, \sry, \srz) \equiv \indicator(T_n(\srx, \sry, \srz) > C_n(\srx, \sry, \srz)); \quad \phi_n^2(\srx, \sry, \srz) \equiv \indicator(T_n(\srx, \sry, \srz) > z_{1-\alpha}). 
\end{equation*}
If the critical value of the first converges in probability to that of the second:
\begin{equation}
C_n(\srx, \sry, \srz) \convp z_{1-\alpha}
\label{eq:convergence-of-critical-value}
\end{equation}
and the test statistic does not accumulate near the limiting critical value:
\begin{equation}
\lim_{\delta \rightarrow 0}\limsup_{n \rightarrow \infty}\ \P_{\law_n}[|T_n(\srx, \sry, \srz)-z_{1-\alpha}| \leq \delta] = 0,
\label{eq:non-accumulation-app}
\end{equation}
then the two tests are asymptotically equivalent:
\begin{equation}
\lim_{n \rightarrow \infty}\P_{\law_n}[\phi_n^{1}(\srx, \sry, \srz) = \phi_n^2(\srx, \sry, \srz)] = 1.
\end{equation}
\end{lemma}

\begin{proof}

Note that for any $\delta>0$, we have
\begin{equation}
	\begin{aligned}
		\mathbb{P}_{\law_{n}}&  {\left[\phi_{n}^{1}(\srx,\sry,\srz) \neq \phi_{n}^{2}(\srx,\sry,\srz)\right] } \\
		&= \mathbb{P}_{\law_{n}}\left[\min \left(z_{1-\alpha}, C_{n}\right)<T_{n} \leq \max \left(z_{ 1-\alpha}, C_{n}\right)\right] \\
		&= \mathbb{P}_{\law_{n}}\left[\min \left(z_{ 1-\alpha}, C_{n} \right)<T_{n} \leq \max \left(z_{ 1-\alpha}, C_{n}\right),\left|C_{n}-z_{ 1-\alpha}\right| \leq \delta\right] \\
		&\quad+\mathbb{P}_{\law_{n}}\left[\min \left(z_{ 1-\alpha}, C_{n}\right)<T_{n} \leq \max \left(z_{ 1-\alpha}, C_{n}\right),\left|C_{n}-z_{ 1-\alpha}\right|>\delta\right] \\
		& \leq  \mathbb{P}_{\law_{n}}\left[\left|T_{n}-z_{1-\alpha}\right| \leq \delta\right]+\mathbb{P}_{\law_{n}}\left[\left|C_{n}-z_{ 1-\alpha}\right|>\delta\right] .
	\end{aligned}
\label{eq:tests-not-same}
\end{equation}

\noindent To justify the last step, suppose without loss of generality that $z_{1-\alpha} \leq C_n$. Then note that if $z_{1-\alpha}<T_{n} \leq C_{n}$ and $C_{n}-z_{1-\alpha} \leq \delta$ then
$$
|T_{n} - z_{1-\alpha}| = T_{n} -z_{1-\alpha} \leq C_n-z_{1-\alpha} \leq \delta.
$$
Taking a limsup on both sides in equation~\eqref{eq:tests-not-same} and using the assumed convergence~\eqref{eq:convergence-of-critical-value}, we find that
$$
\begin{aligned}
	\limsup_{n\to \infty}\mathbb{P}_{\law_{n}} & {\left[\phi_{n}^{1}(\srx,\sry,\srz) \neq \phi_{n}^{2}(\srx,\sry,\srz)\right] } \\
	\leq & \limsup_{n\to \infty}\mathbb{P}_{\law_{n}}\left[\left|T_{n}(\srx,\sry,\srz)-z_{1-\alpha}\right| \leq \delta\right]+ \limsup_{n\to \infty}\mathbb{P}_{\law_{n}}\left[\left|C_{n}(\srx,\sry,\srz)-z_{ 1-\alpha}\right|>\delta\right]\\
	= &\limsup_{n\to \infty}\mathbb{P}_{\law_{n}}\left[\left|T_{n}(\srx,\sry,\srz)-z_{1-\alpha}\right| \leq \delta\right].
\end{aligned}
$$
Letting $\delta \rightarrow 0$ and using our assumption~\eqref{eq:non-accumulation-app}, we arrive at the claimed asymptotic equivalence. This completes the proof.
\end{proof}

\section{Proofs for Section~\ref{sec:dr-and-equivalence} and Appendix~\ref{sec:ndcrt}}

For the sake of this section, we define
\begin{equation}
	s^2_n \equiv \E_{\law_n}[\V_{\law_n}[\prx|\prz]\V_{\law_n}[\pry|\prz]]. 
\end{equation}

\subsection{Proofs of main results}

\begin{proof}[Proof of Theorem~\ref{thm:equivalence}]

To show the asymptotic equivalence of variance estimates~\eqref{eq:asymptotic-variance-equivalence} it suffices to show that
\begin{equation}
(\widehat S_n^{\dCRThat})^2 - s^2_n \convp 0; \quad (\widehat S_n^{\GCM})^2 - s^2_n \convp 0; \quad \inf_{n} s^2_n > 0.
\end{equation}
The first of these statements is given by equation~\eqref{eq:variance_convg_7} in Lemma~\ref{lem:variance_convg}, the second follows from the proof of Theorem 6 in \citet{Shah2018}, and the third is a consequence of assumption~\eqref{eq:sp2} and conditional independence.
  
Given the asymptotic equivalence of the variance estimates~\eqref{eq:asymptotic-variance-equivalence}, we can show using Lemma~\ref{lem:equivalence-lemma} that the GCM test is asymptotically equivalent to the $\MXtwohat$ $F$-test. Indeed, set
\begin{equation}
	T_n \equiv T_n^\GCM \quad \text{and} \quad C_n \equiv \frac{\widehat S_n^{\dCRThat}}{\widehat S_{n}^{\GCM}}z_{1-\alpha}
\end{equation}
Then, $\phi_n^1$ is the $\MXtwohat$ $F$-test and $\phi_n^2$ is the GCM test. The asymptotic equivalence of the variance estimates~\eqref{eq:asymptotic-variance-equivalence} then implies the critical value convergence assumption~\eqref{eq:convergence-of-critical-value} of Lemma~\ref{lem:equivalence-lemma}. The non-accumulation assumption~\eqref{eq:non-accumulation-app} is a consequence of the fact that, under the assumptions of Theorem~\ref{thm:equivalence}, we have $T^{\text{\GCM}}_n \convd N(0,1)$ \citep{Shah2018}. On the other hand, by Corollary~\ref{cor:asymptotic-equivalence-stronger} (whose conclusion holds under the assumptions of Theorem~\ref{thm:equivalence}), we also know that $\dCRThat$ is asymptotically equivalent to the $\MXtwohat$ $F$-test. Hence, both the GCM test and the $\dCRThat$ are asymptotically equivalent to the $\MXtwohat$ $F$-test, so they are asymptotically equivalent to each other as well.
\end{proof}

\begin{proof}[Proof of Proposition~\ref{prop:sufficient-for-variance-consistency}]

We treat the two cases separately.

\paragraph*{Case 1.}
We have
\begin{align*}
&\frac{1}{n} \sum_{i = 1}^n ((\srx_i - \widehat \mu_{n,x}(\srz_i))^2 - \V_{\law_n}[\srx_i \mid \srz_i])\V_{\law_n}[\sry_i\mid\srz_i] \\
&\quad= \frac{1}{n} \sum_{i = 1}^n (\srx_i - \widehat \mu_{n,x}(\srz_i))^2\V_{\law_n}[\sry_i\mid\srz_i] - \frac{1}{n} \sum_{i = 1}^n \V_{\law_n}[\srx_i \mid \srz_i]\V_{\law_n}[\sry_i\mid\srz_i] \\
&\quad= (s^2_n + o_p(1)) - (s^2_n + o_p(1)) \\
&\quad= o_p(1),
\end{align*}
where the third line follows from the convergences~\eqref{eq:variance_convg_4} and~\eqref{eq:variance_convg_1} from Lemma \ref{lem:variance_convg}. This shows the variance consistency property~\eqref{eq:variance-consistency}.

\paragraph*{Case 2.}
We need to show that
$$
\frac{1}{n} \sum_{i=1}^n (f(\widehat \mu_{n,x}(Z_i))-f(\mu_{n,x}(\srx_i)))\V_{\law_n}[\sry_i\mid \srz_i] \convp 0.
$$
By the Cauchy-Schwartz inequality, we have
\begin{equation}
\begin{split} \label{eq:CS_corrolary_doubly_robust}
	&\left|\frac{1}{n} \sum_{i=1}^n (f(\widehat \mu_{n,x}(\srz_i))-f(\mu_{n,x}(\srz_i)))\V_{\law_n}[\sry_i\mid \srz_i]\right| \\ &\quad\leq \sqrt{\frac{1}{n} \sum_{i=1}^n (f(\widehat \mu_{n,x}(Z_i))-f( \mu_{n,x}(Z_i))^2\V_{\law_n}[\sry_i\mid \srz_i]}\sqrt{\frac{1}{n} \sum_{i=1}^n \V_{\law_n}[\sry_i\mid \srz_i]}.
\end{split}
\end{equation}
Given the assumption that $\sup_n \E_{\law_n}[|\pry-\mu_{n,y}(\prz)|^{2+\delta}] < \infty$ for some $\delta>0$, Jensen's inequality gives
\begin{equation*}
	\sup_n\ \E[\V_{\law_n}[\pry\mid \prz]^{1+\delta/2}] \leq \sup_n\ \E_{\law_n}[|\pry-\mu_{n,y}(\prz)|^{2+\delta}] < \infty.
\end{equation*}
Therefore, the weak law of large numbers (Corollary~\ref{cor:wlln}) gives
\begin{equation*}
\frac{1}{n} \sum_{i=1}^n \V_{\law_n}[\sry_i\mid \srz_i]- \E[\V_{\law_n}[\pry\mid \prz]]\convp 0.
\end{equation*}
Furthermore, 
\begin{equation*}
\E[\V_{\law_n}[\pry\mid \prz]] \leq \sup_n\ \E_{\law_n}[|\pry-\mu_{n,y}(\prz)|^{2+\delta}]^{\frac2{2+\delta}} < \infty,
\end{equation*}
so 
\begin{equation*}
	\frac{1}{n} \sum_{i=1}^n \V_{\law_n}[\sry_i\mid \srz_i] = O_p(1).
\end{equation*}
On the other hand, we know $\mathrm{supp}(\mu_{n,x}(Z_i))\subseteq \mathrm{Conv}(\mathrm{supp}(\law_n(\prx))$ for every $i$ and $n$ and by assumption $\widehat{\mu}_{n,x}(\prz)\subseteq\mathrm{Conv}(\mathrm{supp}(\law_n(\prx)))$ almost surely. Together with the fact that $f$ is Lipschitz (say with Lipschitz constant $L$) on $\cup_{n=1}^{\infty}\mathrm{Conv}(\mathrm{supp}(\law_n(\prx)))$, it follows that
\begin{align*}
&\frac{1}{n} \sum_{i=1}^n (f(\widehat \mu_{n,x}(Z_i))-f( \mu_{n,x}(Z_i)))^2\V_{\law_n}[\sry_i\mid \srz_i] \\
&\quad\leq 
\frac{L^2}{n} \sum_{i=1}^n (\widehat \mu_{n,x}(Z_i)- \mu_{n,x}(Z_i))^2\V_{\law_n}[\sry_i\mid \srz_i] \equiv L^2(E'_{n,x})^2 = o_p(1).
\end{align*}
Combining the last two displays with equation~\eqref{eq:CS_corrolary_doubly_robust} gives us the desired result.
\end{proof}

\begin{proof}[Proof of Corollary~\ref{cor:asymptotic-equivalence-alternative}]
Define the event
$$
A_n \equiv \{\phi^{\dCRThat}_{n}(\srx, \sry, \srz) \neq \phi^{\GCM}_{n}(\srx, \sry, \srz)\}
$$
The conclusion~\eqref{eq:asymptotic-test-equivalence} of Theorem~\ref{thm:equivalence} implies that $\P_{\law_n}[A_n] \to 0$, so by the assumed contiguity of $\law'_n$ to $\law_n$ we also have $\P_{\law'_n}[A_n] \to 0$. This shows the desired asymptotic equivalence of tests~\eqref{eq:equivalent-tests-alternative}. To show the asymptotic equivalence of powers, we derive
\begin{align*}
|\E_{\law'_n}[\phi_n^{\dCRThat}&(\srx, \sry, \srz)-\phi_n^{\GCM}(\srx, \sry, \srz)]|
\\ &\leq
\E_{\law'_n}[|\phi_n^{\dCRThat}(\srx, \sry, \srz)-\phi_n^{\GCM}(\srx, \sry, \srz)|]\\
&= \E_{\law'_n}[|\phi_n^{\dCRThat}(\srx, \sry, \srz)-\phi_n^{\GCM}(\srx, \sry, \srz)|\indicator(\phi_n^{\dCRThat} \neq \phi_n^{\GCM})]\\
&\leq\E_{\law'_n}[\indicator(\phi_n^{\dCRThat} \neq \phi_n^{\GCM})]\\
&=\P_{\law'_n}[\phi_n^{\dCRThat} \neq \phi_n^{\GCM}] \to 0.
\end{align*}
This completes the proof.
\end{proof}

\begin{proof}[Proof of Corollary~\ref{cor:dcrt-double-robustness}] 
Fix $\epsilon > 0$, and for each $n$ let $\law^*_n \in \nulllaws_n \cap \regclass_n$ be such that
\begin{align*}
&\P_{\law^*_n}[\phi_n^{\dCRThat}(\srx, \sry, \srz) \neq \phi_n^{\GCM}(\srx, \sry, \srz)] \\
&\quad\geq \sup_{\law_n \in \nulllaws_n \cap \regclass_n} \P_{\law_n}[\phi_n^{\dCRThat}(\srx, \sry, \srz) \neq \phi_n^{\GCM}(\srx, \sry, \srz)] - \epsilon.
\end{align*}
Applying Theorem~\ref{thm:equivalence} to the sequence $\law^*_n$ and using the asymptotic Type-I error control of the GCM test, we obtain
\begin{equation*}
\begin{split}
&\limsup_{n \rightarrow \infty} \sup_{\law_n \in \nulllaws_n \cap \regclass_n} \E_{\law_n}[\phi_n^{\dCRThat}(\srx, \sry, \srz)] \\
&\quad\leq \limsup_{n \rightarrow \infty} \sup_{\law_n \in \nulllaws_n \cap \regclass_n} |\E_{\law_n}[\phi_n^{\dCRThat}(\srx, \sry, \srz)] - \E_{\law_n}[\phi_n^{\GCM}(\srx, \sry, \srz)]| + \E_{\law_n}[\phi_n^{\GCM}(\srx, \sry, \srz)] \\
&\quad\leq \limsup_{n \rightarrow \infty} \sup_{\law_n \in \nulllaws_n \cap  \regclass_n} |\E_{\law_n}[\phi_n^{\dCRThat}(\srx, \sry, \srz)] - \E_{\law_n}[\phi_n^{\GCM}(\srx, \sry, \srz)]| \\
&\quad \quad + \limsup_{n \rightarrow \infty} \sup_{\law_n \in \nulllaws_n \cap \regclass_n} \E_{\law_n}[\phi_n^{\GCM}(\srx, \sry, \srz)] \\
&\quad\leq \limsup_{n \rightarrow \infty} \sup_{\law_n \in \nulllaws_n \cap  \regclass_n} \P_{\law_n}[\phi_n^{\dCRThat}(\srx, \sry, \srz) \neq \phi_n^{\GCM}(\srx, \sry, \srz)]  \\
&\quad \quad + \limsup_{n \rightarrow \infty} \sup_{\law_n \in \nulllaws_n \cap \regclass_n} \E_{\law_n}[\phi_n^{\GCM}(\srx, \sry, \srz)] \\
&\quad\leq \limsup_{n \rightarrow \infty} \sup_{\law_n \in \nulllaws_n \cap  \regclass_n} \P_{\law_n}[\phi_n^{\dCRThat}(\srx, \sry, \srz) \neq \phi_n^{\GCM}(\srx, \sry, \srz)] + \alpha \\
&\quad\leq \limsup_{n \rightarrow \infty}\ \P_{\law^*_n}[\phi_n^{\dCRThat}(\srx, \sry, \srz) \neq \phi_n^{\GCM}(\srx, \sry, \srz)] + \epsilon +\alpha \\
&\quad=\epsilon + \alpha.
\end{split}
\end{equation*}
Sending $\epsilon \rightarrow 0$ gives the desired conclusion.
\end{proof}

\begin{proof}[Proof of Theorem~\ref{thm:normal-limit-ndcrt}]
We have
\begin{equation*}
T_n^{\ndCRThat}(\srxk, \srx, \sry, \srz) \equiv \frac{T_n^{\dCRT}(\srxk, \srx, \sry, \srz)}{S_n^{\GCM}(\srxk, \srx, \sry, \srz)} \equiv  \frac{S_n^{\dCRThat}(\srx, \sry, \srz)}{S_n^{\GCM}(\srxk, \srx, \sry, \srz)} \cdot \frac{T_n^{\dCRT}(\srxk, \srx, \sry, \srz)}{S_n^{\dCRThat}(\srx, \sry, \srz)}
\end{equation*}
The first factor converges to 1 in probability (Lemma~\ref{lem:var-ratio-ndcrt-dcrt}), whereas the second factor converges conditionally on $\srx, \sry, \srz$ to $N(0,1)$ (Theorem \ref{thm:normal-limit}). Putting these two statements together with conditional Slutsky's theorem (Theorem~\ref{thm:cond_slutsky}), we arrive at the convergence~\eqref{eq:conditional-convergence-ndcrt}. Since the standard normal has continuous CDF we can use Lemma~\ref{lem:conditional-convergence-to-quantile} to conclude the convergence of the critical value~\eqref{eq:critical-value-convergence-ndcrt}. The equivalence statement~\eqref{eq:ndcrt-gcm-equivalence} follows from the convergence~\eqref{eq:critical-value-convergence-ndcrt} and Lemma~\ref{lem:equivalence-lemma} applied with $T_n = T_n^{\GCM}$ and $C_n = C_n^{\ndCRThat}$.
\end{proof}

\begin{proof}[Proof of Corollary~\ref{cor:double-robustness-ndcrt}] 
 
The proof of this corollary is directly analogous to that of Corollary~\ref{cor:dcrt-double-robustness}, so we omit it for the sake of brevity.

\end{proof}

\subsection{Auxiliary lemmas}

\begin{lemma}[Conditional Jensen inequality, \cite{Davidson2003}, Theorem 10.18] \label{lem:conditional-jensen}
Let $W$ be a random variable and let $\phi$ be a convex function, such that $W$ and $\phi(W)$ are integrable. For any $\sigma$-algebra $\mathcal F$, we have the inequality
\begin{equation*}
	\phi(\E[W \mid \mathcal F]) \leq  \E[\phi(W) \mid \mathcal F] \quad \text{almost surely}.
\end{equation*}
\end{lemma}

\begin{lemma} \label{lem:conditional-convergence-to-unconditional}
	Let $W_n$ be a sequence of random variables and $\mathcal F_n$ a sequence of $\sigma$-algebras. If $W_n \mid \mathcal F_n \convpp 0$, then $W_n \convp 0$. 
\end{lemma}

\begin{proof}
	Let $\epsilon > 0$. Because of the assumed conditional convergence in probability, we have
	\begin{equation*}
	\P[|W_n| > \epsilon \mid \mathcal F_n] \convp 0.
	\end{equation*}
	By the bounded convergence theorem, it follows that
	\begin{equation*}
	\P[|W_n| > \epsilon] = \E[\P[|W_n| > \epsilon \mid \mathcal F_n]] \rightarrow 0,
	\end{equation*}
	from which the conclusion follows.
\end{proof}

\begin{lemma}
	\label{lem:variance_convg}
	Consider a sequence of laws $\law_n \in \nulllaws_n$. Given assumption~\eqref{eq:sp2}, we have
	\begin{align}
	\frac{1}{n} \sum_{i=1}^n \V_{\law_n}[\srx_i\mid \srz_i]\V_{\law_n}[\sry_i \mid \srz_i] -s^2_n &\convp 0; \label{eq:variance_convg_1} \\ 
	\frac{1}{n}\sum_{i = 1}^n (\srx_i - \mu_{n,x}(\srz_i))^2\V_{\law_n}[\sry_i\mid\srz_i]-s^2_n &\convp 0; \label{eq:variance_convg_2} \\
	\frac{1}{n}\sum_{i = 1}^n \V_{\law_n}[\srx_i\mid\srz_i](\sry_i - \mu_{n,y}(\srz_i))^2-s^2_n &\convp 0. \label{eq:variance_convg_3}
	\end{align}
	Given additionally assumption~\eqref{eq:sp1}, we also have
	\begin{align}
	\frac{1}{n} \sum_{i = 1}^n (\srx_i - \widehat \mu_{n,x}(\srz_i))^2\V_{\law_n}[\sry_i\mid\srz_i] - s^2_n &\convp 0; \label{eq:variance_convg_4} \\
	\frac{1}{n}\sum_{i = 1}^n \V_{\law_n}[\srx_i\mid\srz_i](\sry_i - \widehat \mu_{n,y}(\srz_i))^2-s^2_n &\convp 0. \label{eq:variance_convg_5}
	\end{align}
	Given additionally the conditional Lyapunov condition~\eqref{eq:lyapunov-condition-2} and the variance consistency condition~\eqref{eq:variance-consistency}, we also have
	\begin{align}
	(S^{\dCRThat}_n)^2 - s_n^2 \equiv \frac{1}{n}\sum_{i = 1}^n \V_{\lawhat_n}[\srx_i\mid\srz_i](\sry_i - \mu_{n,y}(\srz_i))^2-s^2_n &\convp 0. \label{eq:variance_convg_6}
	\end{align}
	Given additionally the assumption $\widehat E'_{n,y} \convp 0$, we also have
	\begin{align}
		(\widehat S^{\dCRThat}_n)^2 - s_n^2 \equiv \frac{1}{n}\sum_{i = 1}^n \V_{\lawhat_n}[\srx_i\mid\srz_i](\sry_i - \widehat \mu_{n,y}(\srz_i))^2-s^2_n &\convp 0. \label{eq:variance_convg_7}
	\end{align}

\end{lemma}

\begin{proof}

	We prove the convergence statements in order.

\paragraph*{Proofs of statements~\eqref{eq:variance_convg_1}, \eqref{eq:variance_convg_2}, \eqref{eq:variance_convg_3}.}

These statements are consequences of the weak law of large numbers (Corollary~\ref{cor:wlln}). To verify the statement~\eqref{eq:variance_convg_1}, we note that
\begin{equation*}
	\E_{\law_n}[\V_{\law_n}[\srx_i\mid \srz_i]\V_{\law_n}[\sry_i \mid \srz_i]] = \E_{\law_n}[\V_{\law_n}[\prx\mid \prz]\V_{\law_n}[\pry \mid \prz]] \equiv s^2_n, 
\end{equation*}
and
\begin{equation}
	\begin{split}
	&\sup_{i,n} \ \E_{\law_n}[|\V_{\law_n}[\srx_i\mid \srz_i]\V_{\law_n}[\sry_i \mid \srz_i]|^{1+\delta/2}] \\
	&\quad= \sup_n\ \E_{\law_n}[|\V_{\law_n}[\prx\mid \prz]\V_{\law_n}[\pry \mid \prz]|^{1+\delta/2}] \\
	&\quad\leq \sup_n\ \E_{\law_n}[\E_{\law_n}[|\prx - \mu_{n,x}(\prz)|^{2+\delta} \mid \prz]\E_{\law_n}[|\pry - \mu_{n,y}(\prz)|^{2+\delta} \mid \prz]] \\
	&\quad= \sup_n\ \E_{\law_n}[|(\prx - \mu_{n,x}(\prz))(\pry - \mu_{n,y}(\prz))|^{2 + \delta}] \\
	&\quad< c_2 < \infty.
	\end{split}
	\label{eq:variance-bound-derivation}
\end{equation}
The inequality in the third line follows from the conditional Jensen inequality (Lemma~\ref{lem:conditional-jensen}), the equality in the fourth line follows from the conditional independence assumption, and the inequality in the fifth line follows from the $2+\delta$ moment assumption~\eqref{eq:sp2}. Hence we have verified the sufficient condition~\eqref{eq:wlln_sufficient} for the WLLN, so the convergence~\eqref{eq:variance_convg_1} follows. Statements~\eqref{eq:variance_convg_2} and~\eqref{eq:variance_convg_3} can be verified with similar arguments.

\paragraph*{Proofs of statements~\eqref{eq:variance_convg_4} and~\eqref{eq:variance_convg_5}.} 

We prove only the first, as the second will follow by symmetry. Given the convergence~\eqref{eq:variance_convg_2}, the statement~\eqref{eq:variance_convg_4} will follow if we show that
\begin{equation}
	\begin{split}
	&\frac{1}{n} \sum_{i = 1}^n (\srx_i - \widehat \mu_{n,x}(\srz_i))^2\V_{\law_n}[\sry_i\mid\srz_i] - \frac{1}{n} \sum_{i = 1}^n (\srx_i - \mu_{n,x}(\srz_i))^2\V_{\law_n}[\sry_i\mid\srz_i] \\
	&\quad= \frac{1}{n} \sum_{i = 1}^n (\widehat \mu_{n,x}(\srz_i) - \mu_{n,x}(\srz_i))^2\V_{\law_n}[\sry_i\mid\srz_i] \\
	&\quad \quad -\frac{2}{n}\sum_{i = 1}^n (\srx_i - \mu_{n,x}(\srz_i))(\widehat \mu_{n,x}(\srz_i) - \mu_{n,x}(\srz_i))\V_{\law_n}[\sry_i\mid\srz_i] \\
	&\quad \equiv (E'_{n,x})^2 - \frac{2}{n}\sum_{i = 1}^n (\srx_i - \mu_{n,x}(\srz_i))(\widehat \mu_{n,x}(\srz_i) - \mu_{n,x}(\srz_i))\V_{\law_n}[\sry_i\mid\srz_i] \\
	&\quad\convp 0.
	\label{eq:sufficient-convergence}
	\end{split}
\end{equation}
The convergence $E'_{n,x} \convp 0$ is assumed~\eqref{eq:sp1}. To verify the convergence of the second term, note first that conditional H\"older (Lemma~\ref{lem:cond_holder}) and the derivation~\eqref{eq:variance-bound-derivation} give
\begin{equation}
\sup_n\ s^2_n \leq \sup_n\ \{\E_{\law_n}[|\V_{\law_n}[\prx\mid \prz]\V_{\law_n}[\pry \mid \prz]|^{1+\delta/2}]\}^{2/(2+\delta)} < \infty,
\end{equation}
which combined with the convergence~\eqref{eq:variance_convg_2} implies that
\begin{equation}
	\frac{1}{n}\sum_{i = 1}^n (\srx_i - \mu_{n,x}(\srz_i))^2\V_{\law_n}[\sry_i\mid\srz_i] = O_p(1).
\end{equation}
Therefore, by the Cauchy-Schwartz inequality we find that 
\begin{equation*}
	\begin{split}
		&\left(\frac{1}{n}\sum_{i = 1}^n (\srx_i - \mu_{n,x}(\srz_i))(\widehat \mu_{n,x}(\srz_i) - \mu_{n,x}(\srz_i))\V_{\law_n}[\sry_i\mid\srz_i]\right)^2 \\
		&\quad\leq \left(\frac{1}{n}\sum_{i = 1}^n (\srx_i - \mu_{n,x}(\srz_i))^2 \V_{\law_n}[\sry_i\mid\srz_i]\right) \left(\frac{1}{n}\sum_{i = 1}^n (\widehat \mu_{n,x}(\srz_i) - \mu_{n,x}(\srz_i))^2\V_{\law_n}[\sry_i\mid\srz_i]\right) \\
		&\quad\equiv \left(\frac{1}{n}\sum_{i = 1}^n (\srx_i - \mu_{n,x}(\srz_i))^2 \V_{\law_n}[\sry_i\mid\srz_i]\right) \cdot (E'_{n,x})^2 \\
		&\quad= O_p(1) \cdot o_p(1) = o_p(1).
	\end{split}
\end{equation*}
This proves the convergence~\eqref{eq:sufficient-convergence}, which in turn implies the claimed convergence~\eqref{eq:variance_convg_4}.

\paragraph*{Proof of statement~\eqref{eq:variance_convg_6}.}

Note that
\begin{equation*}
	\begin{split}
		&\frac{1}{n}\sum_{i = 1}^n \V_{\lawhat_n}[\srx_i\mid\srz_i](\sry_i - \mu_{n,y}(\srz_i))^2-s^2_n \\
		&\quad= \frac{1}{n}\sum_{i = 1}^n \V_{\lawhat_n}[\srx_i\mid\srz_i](\sry_i - \mu_{n,y}(\srz_i))^2 - \frac{1}{n}\sum_{i = 1}^n \V_{\lawhat_n}[\srx_i\mid\srz_i]\V[\sry_i \mid \srz_i]  \\
		&\quad \quad + \frac{1}{n}\sum_{i = 1}^n \V_{\lawhat_n}[\srx_i\mid\srz_i]\V[\sry_i \mid \srz_i] - \frac{1}{n}\sum_{i = 1}^n \V_{\law_n}[\srx_i\mid\srz_i]\V[\sry_i \mid \srz_i] \\
		&\quad \quad + \frac{1}{n}\sum_{i = 1}^n \V_{\law_n}[\srx_i\mid\srz_i]\V[\sry_i \mid \srz_i] - s^2_n \\
		&\quad = \frac{1}{n}\sum_{i = 1}^n \V_{\lawhat_n}[\srx_i\mid\srz_i](\sry_i - \mu_{n,y}(\srz_i))^2 - \frac{1}{n}\sum_{i = 1}^n \V_{\lawhat_n}[\srx_i\mid\srz_i]\V[\sry_i \mid \srz_i] + o_p(1),
	\end{split}
\end{equation*}
where we used the variance consistency assumption~\eqref{eq:variance-consistency} and the convergence result~\eqref{eq:variance_convg_1} to obtain the last line. Hence, it suffices to show that
\begin{equation}
	\frac{1}{n}\sum_{i = 1}^n \left(\V_{\lawhat_n}[\srx_i\mid\srz_i](\sry_i - \mu_{n,y}(\srz_i))^2 - \V_{\lawhat_n}[\srx_i\mid\srz_i]\V[\sry_i \mid \srz_i]\right) \convp 0.
	\label{eq:sufficient-convergence-2}
\end{equation}
To this end, we apply the conditional WLLN (Theorem~\ref{thm:wlln_cond}) with $\mathcal{F}_n = \sigma(X,Z)$ and
	\begin{equation*}
	W_{in} = \V_{\lawhat_n}[\srx_i\mid \srz_i](\sry_i-\mu_{n,y}(\srz_i))^2.
	\end{equation*}
	We check the required $1+\delta$ moment condition~\eqref{eq:wlln_cond_assumption}:
	\begin{align*}
		&\frac{1}{n^{1+\delta/2}} \sum_{i = 1}^n \E_{\law_n}[|W_{in}|^{1+\delta/2} \mid \mathcal{F}_n] \\
		&\quad\equiv \frac{1}{n^{1+\delta/2}} \sum_{i = 1}^n \E_{\law_n}\left[\V_{\lawhat_n}[\srx_i\mid \srz_i]^{1+\delta/2} \cdot |\sry_i-\mu_{n,y}(\srz_i)|^{2+\delta}\mid \srx,\srz\right]\\
		&\quad\leq \frac{1}{n^{1+\delta/2}} \sum_{i = 1}^n \E_{\law_n}\left[\E_{\lawhat_n}[|\srxk_i-\widehat{\mu}_{n,x}(Z_i)|^{2+\delta} \mid \srx, \srz] \cdot |\sry_i-\mu_{n,y}(Z_i)|^{2+\delta} \mid \srx,\srz\right]\\
		&\quad= \frac{1}{n^{1+\delta/2}} \sum_{i = 1}^n \E_{\lawhat_n}[|\srxk_i-\widehat\mu_{n,x}(Z_i)|^{2+\delta}\mid \srx, \srz] \cdot \E_{\law_n}\left[|\sry_i-\mu_{n,y}(Z_i)|^{2+\delta} \mid Z_i\right]\\
		&\quad\convp 0. 
	\end{align*}
	The inequality in the third line follows from the conditional Jensen inequality (Lemma~\ref{lem:conditional-jensen}), the equality in the fourth line from the assumed conditional independence, and the convergence in the fifth line from the conditional Lyaponov assumption~\eqref{eq:lyapunov-condition-2}. Therefore, the conditional WLLN gives
	\begin{equation*}
	\begin{split}
	&\frac{1}{n} \sum_{i=1}^n \left(\V_{\lawhat_n}[\srx_i\mid \srz_i](\sry_i-\mu_{n,y}(\srz_i))^2- \V_{\lawhat_n}[\srx_i\mid \srz_i]\V_{\law_n}[\sry_i \mid \srz_i]\right) \\
	&\quad = \frac{1}{n} \sum_{i=1}^n \left(\V_{\lawhat_n}[\srx_i\mid \srz_i](\sry_i-\mu_{n,y}(\srz_i))^2- \E_{\law_n}[\V_{\lawhat_n}[\srx_i\mid \srz_i](\sry_i-\mu_{n,y}(\srz_i))^2 | \srx, \srz]\right) \\
	&\quad \convpp 0,
	\end{split}
	\end{equation*}
	where the equality follows from the assumed conditional independence. Since conditional convergence in probability implies unconditional convergence in probability (Lemma~\ref{lem:conditional-convergence-to-unconditional}), this verifies the claimed convergence statement~\eqref{eq:sufficient-convergence-2} and completes the proof of the statement~\eqref{eq:variance_convg_6}.

\paragraph*{Proof of statement~\eqref{eq:variance_convg_7}.} 

Given the convergence~\eqref{eq:variance_convg_6}, it suffices to show that
\begin{equation}
\frac{1}{n}\sum_{i=1}^n \V_{\lawhat_n}[\srx_i\mid \srz_i](\sry_i-\widehat \mu_{n,y}(\srz_i))^2-\frac{1}{n}\sum_{i=1}^n \V_{\lawhat_n}[\srx_i\mid \srz_i](\sry_i-\mu_{n,y}(\srz_i))^2
\convp 0
\end{equation}
Given the assumption that $\widehat E'_{n,y} \convp 0$, this statement's proof is analogous to that of statement~\eqref{eq:sufficient-convergence}, so we omit it for the sake of brevity. This completes the proof of the lemma.
\end{proof}

\begin{lemma}
	\label{lem:cond-sample-variance-equivalence2}
	Define
	\begin{equation}
	(S_n^{\dCRThat})^2 \equiv \frac{1}{n}\sum_{i = 1}^n \V_{\lawhat_n}[\srx_i|\srz_i](\sry_i - \mu_{n,y}(\srz_i))^2.		
	\end{equation}
	Under the assumptions of Theorem~\ref{thm:equivalence}, we have
	\begin{equation}
	\frac{(S^{\dCRThat}_n)^2}{(\widehat S^{\dCRThat}_n)^2}  \convp 1
	\label{eq:dcrt-var-equivalence}
	\end{equation}
	and
	\begin{equation}
	\P[(S^{\dCRThat}_n)^2 > \epsilon] \rightarrow 1 \quad \text{for some } \epsilon > 0.
	\label{eq:bounded-away-from-zero}
	\end{equation}
\end{lemma}

\begin{proof}
The equivalence of variances~\eqref{eq:dcrt-var-equivalence} follows from the convergences~\eqref{eq:variance_convg_6} and~\eqref{eq:variance_convg_7} (Lemma~\ref{lem:variance_convg}), as well as the observation that conditional independence and the assumption~\eqref{eq:sp2} imply 
\begin{equation}
	\inf_n s^2_n = \inf_{n}\ \E_{\law_n}[(\prx - \mu_{n,x}(\prz))^2(\pry - \mu_{n,y}(\prz))^2] > 0.
\end{equation}
The stochastic boundedness from below~\eqref{eq:bounded-away-from-zero} follows from the latter fact and the convergence~\eqref{eq:variance_convg_6}.
\end{proof}


\begin{lemma} \label{lem:var-ratio-ndcrt-dcrt}
If the nondegeneracy conditions~\eqref{eq:var-bounded-below} and~\eqref{eq:variance-bounds} and the conditional Lyapunov condition~\eqref{eq:lyapunov-condition} hold, then the variance estimates $(\widehat S_n^{\GCM})^2$ and $(\widehat S_n^{\dCRThat})^2$ are equivalent under resampling:
\begin{equation}
\frac{(\widehat S_n^{\dCRThat}(\srx, \sry, \srz))^2}{(\widehat S_{n}^{\GCM}(\srxk, \srx, \sry, \srz))^2} \convp 1.
\label{eq:var-ratio-ndcrt-dcrt}
\end{equation}
\end{lemma}
\begin{proof}
Define 
\begin{equation*}
W_{in}  \equiv  (\srxk_i - \widehat \mu_{n,x}(\srz_i))(\sry_i - \widehat \mu_{n,y}(\srz_i)),
\end{equation*}
so that 
\begin{equation*}
(\widehat S_{n}^{\GCM}(\srxk, \srx, \sry, \srz))^2 \equiv \frac{1}{n} \sum_{i=1}^n W_{in}^2 - \left(\frac{1}{n} \sum_{i=1}^n W_{in}\right)^2.
\end{equation*} 
First we claim that $\frac{1}{n}\sum_{i = 1}^n W_{in} \convp 0$.  We will use conditional WLLN (Theorem~\ref{thm:wlln_cond}) with $\mathcal{F}_n \equiv \sigma(\srx,\sry,\srz)$. First note that $\E[W_{in} \mid \mathcal{F}_n] = 0$ by construction. We also check the moment condition~\eqref{eq:wlln_cond_assumption}:
\begin{align*}
	\frac{1}{n^{2+\delta}} \sum_{i = 1}^n &\E[|W_{in}|^{2+\delta} \mid \mathcal{F}_n] \\&= \frac{1}{n^{1+\delta/2}}\frac{1}{n^{1+\delta/2}} \sum_{i=1}^n |\sry_i-\widehat\mu_{n,y}(\srz_i)|^{2+\delta}\E\left[|\srxk_i-\widehat\mu_{n,x}(\srz_i)|^{2+\delta}\mid \srx,\srz\right] \convp 0,
\end{align*} 
where the latter convergence is by the assumption~\eqref{eq:lyapunov-condition}.
Hence we have that $\frac{1}{n} \sum_{i = 1}^n W_{in}  \mid \mathcal F_n \convpp 0$ which by Lemma \ref{lem:conditional-convergence-to-unconditional} implies $\frac{1}{n} \sum_{i=1}^n W_{in} \convp 0$.

Next we show that $\frac{1}{n} \sum_{i=1}^n W^2_{in} - (\widehat S_{n}^{\dCRThat})^2 \convp 0$. We will use conditional WLLN with $\mathcal{F}_n = \sigma(\srx,\sry,\srz)$, observing that $\E\left[\frac{1}{n} \sum_{i=1}^n W^2_{in} \mid \mathcal{F}_n\right] = (\widehat S_{n}^{\dCRThat})^2 $. Next we verify the moment condition~\eqref{eq:wlln_cond_assumption}:
\begin{align}
	\frac{1}{n^{1+\delta/2}} \sum_{i = 1}^n &\E[|W_{in}|^{2+\delta} \mid \mathcal{F}_n] \\&= \frac{1}{n^{1+\delta/2}} \sum_{i=1}^n |\sry_i-\widehat\mu_{n,y}(\srz_i)|^{2+\delta}\E\left[|\srxk_i-\widehat\mu_{n,x}(\srz_i)|^{2+\delta}\mid \srx,\srz\right] \convp 0,
\end{align} 
where the latter convergence is by the assumption~\eqref{eq:lyapunov-condition}.
Hence we have that $\frac{1}{n} \sum_{i = 1}^n W^2_{in} - (\widehat S_{n}^{\dCRThat})^2  \mid \mathcal F_n \convpp 0$ which by Lemma \ref{lem:conditional-convergence-to-unconditional} implies that $\frac{1}{n} \sum_{i=1}^n W^2_{in} - (\widehat S_{n}^{\dCRThat})^2 \convp 0$.

Combining both of these results we find that 
\begin{equation*}
(\widehat S_{n}^{\GCM}(\srxk, \srx, \sry, \srz))^2 - (\widehat S_{n}^{\dCRThat}(\srx, \sry, \srz))^2 \convp 0. 
\end{equation*}
Now using the nondegeneracy condition~\eqref{eq:var-bounded-below} we can conclude that~\eqref{eq:var-ratio-ndcrt-dcrt} holds true, as desired.
\end{proof}

\section{Proofs for Section~\ref{sec:optimality}} \label{sec:optimality-proofs}

The goal of this section is to prove our main optimality result (Theorem~\ref{thm:optimality}) and Corollary~\ref{cor:RKHS_example}. The idea of the proof of Theorem~\ref{thm:optimality} is to reduce the problem to a semiparametric testing problem, and then to use existing semiparametric optimality theory. To this end, we first review the relevant semiparametric theory (Section~\ref{sec:semiparametric-preliminaries}). Then we leverage this theory to prove Theorem~\ref{thm:optimality} (Section~\ref{sec:optimality-proof}) and verify Corollary~\ref{cor:RKHS_example} (Section~\ref{sec:proof-rkhs}). Finally, we carry out deferred semiparametric computations (Section~\ref{sec:proof-of-lemma-9}).

\subsection{Semiparametric preliminaries} \label{sec:semiparametric-preliminaries}

Consider a semiparametric model parameterized by 
\begin{equation}
	(\beta, g) \in \R \times \H_g \subseteq \R \times L^2(\nu),
\end{equation}
where $\nu$ is a measure on $\R^p$ and $\H_g\subseteq L^2(\nu)$ is a linear subspace. First, we define a notion of local Type-I error control within the context of the semiparametric model.
\begin{definition} \label{def:type-I-control}
Fix a point $g_0 \in \H_g$, and define $\theta_0 \equiv (0, g_0)$. A sequence of tests $\phi_n$ of $H_0: \beta = 0$ has asymptotic Type-I error control at $\theta_0$ relative to the tangent space $\dot \law_{\theta_0}$ if, for each submodel $t \mapsto \law_{(0, g_t)}$ with score in $\dot \law_{\theta_0}$ along which $\beta$ is differentiable, we have
\begin{equation}
\limsup_{n \rightarrow \infty}\ \E_{\law_{(0, g_{1/\sqrt{n}})}}[\phi_n(W)] \leq \alpha.
\end{equation}
\end{definition}

This definition is most similar to that of \citet{Choi1996}, except the latter paper does not explicitly use the language of tangent spaces; our definition accommodates Type-I error control over more restricted sets of null distributions reflecting regularity conditions. Next we state a version of the classic semiparametric optimality result:

\begin{theorem}[Theorem 1 in \cite{Choi1996}, Theorem 25.44 in \cite{VDV1998}, Theorem 18.12 in \cite{Kosorok2008}] \label{thm:classic-semiparametric-optimality}
Consider a semiparametric model $\{\law_{\beta, g}: (\beta, g) \in \R \times \H_g\}$ and a point $\theta_0 \equiv (0, g_0)$ for some $g_0 \in \H_g$. Suppose $\beta$ is differentiable at $\law_{\theta_0}$ relative to the tangent space $\dot \law_{\theta_0}$ with efficient influence function $\widetilde S/\widetilde I(\theta_0)$, where $\widetilde S$ is the efficient score and $\widetilde I(\theta_0) > 0$ is the efficient information. For any sequence of tests $\phi_n$ of $H_0: \beta = 0$ with asymptotic Type-I error control at $\theta_0$ relative to the tangent space $\dot \law_{\theta_0}$ and any differentiable submodel $\law_t = \law_{(th_\beta, g_t)}$ with score in $\dot \law_{\theta_0}$ we have
\begin{equation}
\limsup_{n \rightarrow \infty}\ \E_{\law_{1/\sqrt{n}}}[\phi_n(W)] \leq 1 - \Phi(z_{1-\alpha} - h_\beta \cdot \widetilde I(\theta_0)^{1/2}).
\end{equation}
This bound is achieved by the efficient score test $\phi_n^{\textnormal{opt}}(\srx, \sry, \srz) \equiv \indicator(T_n^{\textnormal{opt}}(\srx, \sry, \srz) > z_{1-\alpha})$, where
\begin{equation}
T_n^{\textnormal{opt}}(\srx, \sry, \srz) \equiv \frac{1}{\widetilde I(\theta_0)^{1/2}n^{1/2}}\sum_{i = 1}^n \widetilde S(\srx_i, \sry_i, \srz_i).
\label{eq:normalized-efficient-score}
\end{equation}
In other words,
\begin{equation}
\begin{split}
\lim_{n \rightarrow \infty}\E_{\law_{1/\sqrt{n}}}[\phi_n^{\textnormal{opt}}(W)] &= \lim_{n \rightarrow \infty}\P_{\law_{1/\sqrt{n}}}[T_n^{\textnormal{opt}}(\srx, \sry, \srz) > z_{1-\alpha}] \\
&= 1 - \Phi(z_{1-\alpha} - h_\beta \cdot \widetilde I(\theta_0)^{1/2}).
\end{split}
\label{eq:power-of-efficient-score-test}
\end{equation}
\end{theorem}

This result is like \citet[Theorem 1]{Choi1996}, except it explicitly deals with tangent spaces. On the other hand, the result is like \citet[Theorem 25.44]{VDV1998} or \citet[Theorem 18.12]{Kosorok2008}, except it is written in terms of semiparametric models and assumes Type-I error control in the sense of Definition~\ref{def:type-I-control} above. By comparison, \citet{VDV1998, Kosorok2008} assume Type-I error control at each point $(0, g)$ for $g \in \H_g$. By inspection of the proof of \citet[Theorem 25.44]{VDV1998}, only local Type-I error control (Definition~\ref{def:type-I-control}) is actually needed. In this sense, Theorem~\ref{thm:classic-semiparametric-optimality} can be verified using the same proof as that of \citet[Theorem 25.44]{VDV1998}, specializing to the case of semiparametric models.

\subsection{Proof of Theorem~\ref{thm:optimality}}  \label{sec:optimality-proof}

To apply the semiparametric theory from the previous section, the following lemma (proved in Section~\ref{sec:proof-of-lemma-9}) identifies the tangent space, the efficient score, and the efficient information at $\law_{\theta_0}$. These results are not novel or surprising; similar results are stated, for example, by \citet{Robins2001} in the cases of linear, logistic, and Poisson regressions. Nevertheless, we state and prove Lemma~\ref{lem:semiparametric-results} for a self-contained exposition and for precisely tracking the technical assumptions used. 

\begin{lemma} \label{lem:semiparametric-results}
In the context of the semiparametric model \eqref{eq:alternatives-1}, suppose the following assumptions hold:
\begin{align}
	s^2(\theta_0) \equiv \E_{\law_{\theta_0}}[\V_{\law_{\theta_0}}[\prx|\prz]\V_{\law_{\theta_0}}[\pry|\prz]] > 0; \label{eq:nonsingular-fisher-info} \\
	\ddot{\psi} = K > 0 \text{ and } \E_{\law_{x,z}}[\prx^2] < \infty \text{ OR } \mathrm{supp}(\prx, \prz) \text{ is compact and } \H_g \subseteq C(\R^p), \label{eq:moment-assumptions-app}\\
\E_{\law_{x,z}}[\prx|\ \cdot \ ] \in \H_g. \label{eq:conditional-expectation-app}
\end{align}
For each $h = (h_\beta, h_g) \in \R \times \H_g$, the parametric submodel $t \mapsto \law_{(th_\beta, g_0 + th_g)}$ is differentiable in quadratic mean at $t = 0$ with score function
\begin{equation}
S(\prx, \pry, \prz) = \big(\prx h_\beta+h_g(\prz)\big)(\pry-\E_{\law_{\theta_0}}[\pry \mid \prz])
\label{eq:gcm-score}
\end{equation}
and satisfies the following local asymptotic normality:
\small
\begin{align}
	\log\prod_{i=1}^n\frac{\mathrm{d}\law_{\theta_n(h)}}{\mathrm{d}\law_{\theta_0}}(X_i,Y_i,Z_i)
	&
	=\frac{1}{\sqrt{n}}\sum_{i=1}^n S(\srx_i, \sry_i, \srz_i) -\frac{1}{2}\V_{\law_{\theta_0}}[S(\prx, \pry, \prz)]+o_{\law_{\theta_0}}(1).
	\label{eq:LAN}
\end{align}
\normalsize
The parameter $\beta$ is differentiable at $\law_{\theta_0}$ relative to the tangent space
\begin{align}
\dot \law_{\theta_0} \equiv \{\big(\prx h_\beta+h_g(\prz)\big)(\pry-\E_{\law_{\theta_0}}[\pry \mid \prz]) : (h_\beta,h_g)\in\mathbb{R}\times \mathcal{H}_g\}
	\label{eq:GCM-tangent-space}
\end{align}
with efficient score function
\begin{equation}
\widetilde S(\prx, \pry, \prz) = (\prx - \E_{\law_{\theta_0}}[\prx \mid \prz])(\pry- \E_{\law_{\theta_0}}[\pry\mid \prz]),
\label{eq:GCM-efficient-score}
\end{equation}
efficient information 
\begin{equation}
\widetilde I(\theta_0) = s^2(\theta_0) \equiv \E_{\law_{\theta_0}}[\V_{\law_{\theta_0}}[\prx|\prz]\V_{\law_{\theta_0}}[\pry|\prz]],
\label{eq:GCM-efficient-info}
\end{equation}
and efficient influence function equal to the ratio of the efficient information and the efficient score.

\end{lemma}

Note that assumptions~\eqref{eq:moment-assumptions-app} and~\eqref{eq:conditional-expectation-app} of Lemma~\ref{lem:semiparametric-results} are the same as assumptions~\eqref{eq:moment-assumptions} and~\eqref{eq:conditional-expectation}
of Theorem~\ref{thm:optimality} in the main text; they are restated here for the reader's convenience. Using Lemma~\ref{lem:semiparametric-results} in conjunction with Theorem~\ref{thm:classic-semiparametric-optimality}, we can prove Theorem~\ref{thm:optimality}.

\begin{proof}[Proof of Theorem~\ref{thm:optimality}]

Let $\phi_n$ be a level $\alpha$ test of $H_0$ as defined in equation~\eqref{eq:testing-problem}, and fix $g_0 \in \mathcal S$. By assumption~\eqref{eq:interior-point}, $\theta_n(0, h_g) \in \regclass$ for all $h_g \in \H_g$ for all sufficiently large $n$. Therefore, $\phi_n$ also has asymptotic Type-I error control at $\theta_0 \equiv (0, g_0)$ relative to the tangent space $\dot \law_{\theta_0}$~\eqref{eq:GCM-tangent-space} in the sense of Definition~\ref{def:type-I-control}. Indeed, it suffices to take submodels $t \mapsto \law_{(0, g_t)}$ for $g_t = g_0 + th_g$ and $h_g \in \H_g$, so that $\law_{(0, g_{1/\sqrt{n}})} = \law_{\theta_n(0, h_g)}$. By Lemma~\ref{lem:semiparametric-results} (applicable because its first assumption~\eqref{eq:nonsingular-fisher-info} is implied by assumption~\eqref{eq:sp-assumptions} of Theorem~\ref{thm:optimality} and its last two assumptions are also assumed by Theorem~\ref{thm:optimality}), the assumptions of Theorem~\ref{thm:classic-semiparametric-optimality} are met with efficient score $\widetilde S$~\eqref{eq:GCM-efficient-score} and efficient information $s^2(\theta_0)$~\eqref{eq:GCM-efficient-info}, so taking submodels $t \mapsto \law_{(th_\beta, g_0 + th_g)}$ we find
\begin{equation}
	\limsup_{n \rightarrow \infty}\ \E_{\law_{\theta_n(h)}}[\phi_n(\srx, \sry, \srz)] \leq 1 - \Phi(z_{1-\alpha} - h_\beta \cdot s(\theta_0)).
\end{equation}
On the other hand, because $\law_{\theta_0} \in \mathscr{R}$, it follows that
\begin{equation}
\begin{split}
T^{\GCM}_n(\srx, \sry, \srz) &= \frac{1}{s(\theta_0)\sqrt{n}}\sum_{i = 1}^n (\srx_i - \E_{\law_{\theta_0}}[\srx_i \mid \srz_i])(\sry_i- \E_{\law_{\theta_0}}[\sry_i \mid \srz_i]) + o_{\law_{\theta_0}}(1) \\
&= \frac{1}{\widetilde I(\theta_0)^{1/2}n^{1/2}}\sum_{i = 1}^n \widetilde S(\srx_i, \sry_i, \srz_i) + o_{\law_{\theta_0}}(1) \\
&= T_n^{\text{opt}}(\srx, \sry, \srz) + o_{\law_{\theta_0}}(1).
\end{split}
\label{eq:asymptotic-equivalence-gcm-opt}
\end{equation}
The first equality follows from the proof of Theorem 6 in \citet{Shah2018}, the second follows from the derivations of the efficient score~\eqref{eq:GCM-efficient-score} and efficient information~\eqref{eq:GCM-efficient-info} in Lemma~\ref{lem:semiparametric-results}, and the third from equation~\eqref{eq:normalized-efficient-score} in Theorem~\ref{thm:classic-semiparametric-optimality}. From the local asymptotic normality~\eqref{eq:LAN} it follows that $\prod_{i=1}^n\law_{\theta_{n}(h)}$ and $\prod_{i=1}^n\law_{\theta_0}$ are contiguous by Le Cam's first lemma \citep[Example 6.5]{VDV1998}. It follows that 
\begin{equation*}
	T^{\GCM}_n(\srx, \sry, \srz) = T_n^{\text{opt}}(\srx, \sry, \srz) + o_{\law_{\theta_{n}(h)}}(1)
\end{equation*}
We therefore find that
\begin{align*}
	1 - \Phi(z_{1-\alpha} - h_\beta \cdot s(\theta_0)) = \lim_{n\rightarrow\infty}\P_{\law_{\theta_{n}(h)}}[T_n^{\text{opt}}>z_{1-\alpha}]=\lim_{n\rightarrow\infty}\P_{\law_{\theta_{n}(h)}}[T_n^{\GCM}>z_{1-\alpha}].
\end{align*}
The first inequality follows from the conclusion~\eqref{eq:power-of-efficient-score-test} of Theorem~\ref{thm:classic-semiparametric-optimality} and the second equality follows from equation~\eqref{eq:asymptotic-equivalence-gcm-opt} and Le Cam's first lemma. Therefore we have shown that for any $h \in (0, \infty) \times \H_g$ and any level $\alpha$ conditional independence test $\phi_n$, we have
\begin{equation*}
\limsup_{n \rightarrow \infty}\ \E_{\law_{\theta_n(h)}}[\phi_n(\srx, \sry, \srz)] \leq 1 - \Phi(z_{1-\alpha} - h_\beta \cdot s(\theta_0)) = \lim_{n\rightarrow\infty}\P_{\law_{\theta_{n}(h)}}(T_n^{\GCM}>z_{1-\alpha}).
\end{equation*}
This shows that $\phi_n^{\GCM}$ is LAUMP($g_0$) and verifies the claimed asymptotic power~\eqref{eq:power-of-gcm-test}. Furthermore, since $g_0 \in \mathcal S$ was chosen arbitrarily, it follows that $\phi_n^{\GCM}$ is also LAUMP($\mathcal S$). This completes the proof.
\end{proof}

\subsection{Proof of Corollary \ref{cor:RKHS_example}} \label{sec:proof-rkhs}

It suffices to verify each of the four assumptions of Theorem~\ref{thm:optimality}.

\paragraph*{Verification of assumption~\eqref{eq:sp-assumptions}.}

Note that assumption~\eqref{eq:sp2} is satisfied because by construction, $\prx - \mu_{x}(\prz)$ and $\pry - \mu_y(\prz)$ are independent standard normal random variables for any $\law \in \regclass$. Next, let 
\begin{align}\label{eq:eigen_decomp}
	k(z,z')=\sum_{j=1}^{\infty}\lambda_{j}e_{j}(z)e_{j}(z')
\end{align}
be an eigendecomposition of the Sobolev kernel $k$ with eigenfunctions $e_{j}$ orthonormal with respect to $\text{Unif}[0,1]$. To verify assumption~\eqref{eq:sp1} given assumption~\eqref{eq:sp2}, it suffices to prove the following statements \citep[Theorem 11 and Remark 12]{Shah2018}: 
\begin{align*}
\V_{\law}[\prx|\prz], \V_{\law}[\pry|\prz] \leq \sigma^2 < \infty \quad \text{almost surely, for all } \law \in \regclass;& \\
\sup_{\law \in \regclass}\ \max(\norm{\mu_{n,x}}_{W^{1,2}[0,1]}, \norm{\mu_{n,y}}_{W^{1,2}[0,1]}) < \infty;& \\
\sum_{j = 1}^\infty \lambda_{j} < \infty.&
\end{align*}
The first two of these statements follow directly from the construction of $\regclass$. The third follows from the eigendecomposition of the Sobolev kernel under the uniform measure on $[0,1]$ \citep[Example 12.23]{Wainwright2019} with $\lambda_j = (\frac{2}{(2j-1)\pi})^2$.

\paragraph*{Verification of assumption~\eqref{eq:moment-assumptions}.}

Since we are using the normal exponential family, we have $\ddot \psi = 1$. Furthermore, 
\begin{equation}
\E_{\law_{x,z}}[\prx^2] = \E_{\law_{x,z}}[\E_{\law_{x,z}}[\prx^2|\prz]] = \E_{\law_{x,z}}[\mu_{0x}(\prz)^2 + 1] < \infty,
\end{equation}
since $\mu_{0x} \in W^{(1,2)}([0,1]) \subseteq L^2([0,1])$.

\paragraph*{Verification of assumption~\eqref{eq:conditional-expectation}.}

By construction, $\E_{\law_{x,z}}[\prx |\ \cdot\ ] = \mu_{0x} \in W^{1,2}([0,1])$.

\paragraph*{Verification of assumption~\eqref{eq:interior-point}.} This assumption is a consequence of the openness of the ball $B_{W^{1,2}}(0, C)$ in $W^2([0,1])$.

\subsection{Proof of Lemma~\ref{lem:semiparametric-results}} \label{sec:proof-of-lemma-9}

\paragraph*{Differentiability of parametric submodels.}

Consider the parametric submodel $t \mapsto \law_{(th_\beta, g_0 + th_g)}$ for some $(h_\beta, h_g) \in \R \times \H_g$, and denote 
\begin{equation*}
\eta_t(x, z) \equiv x th_\beta+g_0(z)+th_g(z). 
\end{equation*}
Letting $\lambda_{y}$ be the base measure of the exponential family $f_\eta$, we denote $\lambda \equiv \law_{x,z} \times \lambda_{y}$ and $\mathrm d \law_{(th_\beta, g_0 + th_g)}(x, y, z)/\mathrm d \lambda$ the density of the parametric model for $(\prx, \pry, \prz)$ with respect to $\lambda$. According to \citet[Lemma 7.6]{VDV1998}, this submodel is differentiable in quadratic mean at $t = 0$ if the map 
\begin{equation*}
t \mapsto \sqrt{\frac{\mathrm d \law_{(th_\beta, g_0 + th_g)}}{\mathrm d\lambda}(x, y, z)} = \sqrt{\frac{\mathrm df_{\eta_t}}{\mathrm d\lambda_y}(y)} = \exp\left(\frac{y\eta_t-\psi(\eta_t)}{2}\right)
\end{equation*}
is continuously differentiable at $t = 0$ for each $(x,  y, z) \in \R^{1 + 1 + p}$ and the elements of the Fisher information matrix are well-defined and continuous at $t = 0$. To show continuous differentiability of the square root density, we compute that
\begin{align*}
	\frac{\partial}{\partial t}\sqrt{\frac{\mathrm d \law_{(th_\beta, g_0 + th_g)}}{\mathrm d\lambda}(x, y, z)} = \exp\left(\frac{y\eta_t-\psi(\eta_t)}{2}\right)\cdot\frac{(y-\dot{\psi}(\eta_t))}{2}\cdot (xh_\beta + h_g(z)).
\end{align*}
The linearity of $\eta_t$ in $t$ and the smoothness of $\psi$ imply the continuous differentiability of the above function in $t$.

Next consider the information matrix
\begin{align*}
	I_t &\equiv \E_{\law_{(th_\beta, g_0 + th_g)}}\left[\left(\frac{\partial}{\partial t}\log \frac{\mathrm d \law_{(th_\beta, g_0 + th_g)}}{\mathrm d\lambda}(\prx, \pry, \prz)\right)^2\right] \\
	&= \E_{\law_{(th_\beta, g_0 + th_g)}}\left[\left(\frac{\partial}{\partial t}(\pry \eta_t(\prx, \prz) - \psi(\eta_t(\prx, \prz)))\right)^2\right] \\
	&= \E_{\law_{(\beta_t,g_t)}}[(\pry-\dot{\psi}(\eta_t(\prx, \prz)))^2(\prx h_\beta+h_g(\prz))^2] \\
	&=\E_{\law_{x,z}}\left[(\prx h_\beta+h_g(\prz))^2\ddot{\psi}(\eta_t(\prx, \prz))\right].
\end{align*}
We must show that $I_t$ is well-defined and continuous at $t = 0$. By assumption~\eqref{eq:moment-assumptions-app}, either $\ddot \psi = K > 0$ and $\E_{\law_{x,z}}[\prx^2] < \infty$ or $(\prx,\prz)$ is compactly supported and $\H_g \subseteq C(\R^p)$. If $\ddot \psi = K > 0$ and $\E_{\law_{x,z}}[\prx^2] < \infty$, then we have 
\begin{align*}
I_0=C\E_{\law_{x,z}}\left[(\prx h_\beta+h_g(\prz))^2\right]&\leq 2Kh_\beta^2\E_{\law_{x,z}}[\prx^2]+2K\E_{\law_{x,z}}[h_g^2(\prz)] \\
&= 2Kh_\beta^2\E_{\law_{x,z}}[\prx^2]+2K\|h_g\|_{L^2(\nu)}^2 <\infty.
\end{align*}
Note that, for the sake of this proof, we denote
\begin{equation}
\nu \equiv \law_{x,z}(\prz), \quad \text{so that} \quad \|h_g\|_{L^2(\nu)}^2 \equiv \E_{\law_{x,z}}[h_g(\prz)^2] < \infty \quad \text{for } h_g \in \H_g. 
\end{equation}
Hence, $I_0$ is well-defined. $I_t$ is also continuous at $t = 0$ because it does not depend on $t$. 

On the other hand, suppose $(\prx,\prz)$ is compactly supported and $\H_g \subseteq C(\R^p)$. The quantity inside the expectation defining $I_0$ is a bounded random variable, because the assumed continuity of $h_g$ implies that this quantity is a continuous function of a random vector $(\prx, \prz)$ with compact support. Hence, $I_0$ is well-defined because it is the expectation of a bounded random variable. To show continuity of $I_t$, note that by the assumed continuity of $h_g$ and compact support of $(\prx,\prz)$ we have $\sup_{t \leq 1} \eta_t(\prx,\prz) \leq B < \infty$ almost surely. Therefore, for $t \leq 1$, we have
\begin{align*}
	|I_t-I_0|
	&
	=\left|\E_{\law_{x,z}}\left[(\prx h_\beta+h_g(\prz))^2(\ddot{\psi}(\eta_t)-\ddot{\psi}(\eta_0))\right]\right|\\
	&
	\leq \E_{\law_{x,z}}\left[(\prx h_\beta+h_g(\prz))^2|\ddot{\psi}(\eta_t)-\ddot{\psi}(\eta_0)|\right]\\
	&
	\leq \E_{\law_{x,z}}\left[(\prx h_\beta+h_g(\prz))^2\sup_{|b| \leq B}|\dddot{\psi}(b)| \cdot |\eta_t-\eta_0|\right]\\
	&
	\leq \sup_{|b| \leq B}|\dddot{\psi}(b)| \cdot \E_{\law_{x,z}}\left[|\prx h_\beta+h_g(\prz)|^3\right] \cdot t.
\end{align*}
We have $\sup_{|b| \leq B}|\dddot{\psi}(b)| < \infty$ because $\dddot \psi$ is a continuous function, and $\prx h_\beta+h_g(\prz)$ almost surely bounded as before. Therefore, we conclude that $|I_t - I_0| \rightarrow 0$ as $t \rightarrow 0$, so $I_t$ is indeed continuous at 0.

Hence, we conclude by \citet[Lemma 7.6]{VDV1998} that the parametric submodel $t \mapsto \law_{(th_\beta, g_0 + th_g)}$ is differentiable in quadratic mean at $t = 0$ with score function
\begin{align*}
S(\prx, \pry, \prz) &= \left.\frac{\partial}{\partial t}(\pry \eta_t(\prx, \prz) - \psi(\eta_t(\prx, \prz)))\right|_{t = 0} \\
&= \left.(\pry-\dot{\psi}(\eta_t(\prx, \prz)))(\prx h_\beta+h_g(\prz))\right|_{t = 0} \\
&= (\prx h_\beta+h_g(\prz))(\pry-\E_{\law_{\theta_0}}[\pry|\prz]),
\end{align*}
as claimed~\eqref{eq:gcm-score}.

\paragraph*{Local asymptotic normality of parametric submodels.}

The local asymptotic normality of parametric submodels~\eqref{eq:LAN} follows from the previously established quadratic mean differentiability \citep[Theorem 7.2]{VDV1998}.

\paragraph*{Efficient score and information.}

For $(h_\beta, h_g) \in \R \times \H_g$, define the score operator
\begin{equation}
	\begin{split}
	A(h) &\equiv \big(\prx h_\beta+h_g(\prz)\big)(\pry-\E_{\law_{\theta_0}}[\pry \mid \prz]) \\
	&=  \prx (\pry-\E_{\law_{\theta_0}}[\pry \mid \prz]) h_\beta + (\pry-\E_{\law_{\theta_0}}[\pry \mid \prz])h_g(\prz) \\
	&\equiv A_{\beta} (h_\beta) + A_{g} (h_g).
	\end{split}
\end{equation}
The tangent space $\dot \law_{\theta_0}$~\eqref{eq:GCM-tangent-space} can then be expressed as the range of $A$: 
\begin{equation}
	\dot \law_{\theta_0} = A(\R \times \H_g).
	\label{eq:GCM-tangent-space-old}
\end{equation}
As discussed in \cite[Section 25.4]{VDV1998}, the efficient score for $\beta$ is
\begin{equation}
\tilde{S} = A_{\beta} - \Pi_{\beta, g}A_{\beta},
\end{equation}
where $\Pi_{\beta, g}$ is the orthogonal projection onto the closure $\overline{A_{g}(\H_g)}$ of the nuisance tangent space $A_{g}(\H_g)$ in $L^2(\law_{\theta_0})$. In other words,
\begin{equation}
\Pi_{\beta, g}A_{\beta} = \argmin{W \in \overline{A_{g}(\H_g)}}\ \|A_{\beta}-W\|_{L^2(\law_{\theta_0})}.
\label{eq:projection-of-score}
\end{equation}
To compute this projection, we first claim that the extended operator $A_{g}: L^2(\nu) \rightarrow L^2(\law_{\theta_0})$ is continuous and that $A_{g}^* A_{g}$ is continuously invertible. To verify continuity of $A_{g}$, note that for $h_g \in L^2(\nu)$ we have
\begin{equation}
\begin{split}
\|A_{g}(h_g)\|^2_{L^2(\law_{\theta_0})} &= \E_{\law_{\theta_0}}\left[\left( (\pry-\E_{\law_{\theta_0}}[\pry \mid \prz])h_g(\prz)\right)^2\right] \\
&= \E_{\law_{\theta_0}}[(\pry-\E_{\law_{\theta_0}}[\pry \mid \prz])^2h^2_g(\prz)] \\
&= \E_{\law_{\theta_0}}[\ddot \psi(g(\prz)) h^2_g(\prz)] \\
&\leq C\E_{\law_{\theta_0}}[h^2_g(\prz)] \\
&= C \|h_g\|^2_{L^2(\nu)}.
\end{split}
\end{equation}
Next, we derive the adjoint operator $A^*_{g}: L^2(\law_{\theta_0}) \rightarrow L^2(\nu)$. For a random variable $W \in L^2(\law_{\theta_0})$, we have
\begin{equation}
\begin{split}
	\langle W, A_{g} h_g \rangle_{L^2(\law_{\theta_0})} & = \E_{\law_{\theta_0}}\left[W (\pry-\E_{\law_{\theta_0}}[\pry \mid \prz])h_g(\prz)\right] \\
	&= \E_{\law_{\theta_0}}[\E_{\law_{\theta_0}}[W(\pry-\E_{\law_{\theta_0}}[\pry\mid \prz])\mid \prz] h_g(\prz)] \\
	&= \langle \E_{\law_{\theta_0}}[W(\pry-\E_{\law_{\theta_0}}[\pry \mid \cdot\ ])\mid \cdot\ ],  h_g \rangle_{L^2(\nu)}.
\end{split}
\end{equation}
It follows that
\begin{equation}
(A^*_{g}W)(z) = \E_{\law_{\theta_0}}[W(\pry-\E_{\law_{\theta_0}}[\pry \mid \prz = z])\mid \prz = z].
\end{equation}
Next we derive that
\begin{equation*}
\begin{split}
	(A^*_{g}A_{g} h_{\eta})(z) &= \E_{\law_{\theta_0}}[A_{g} h_{\eta}(\pry-\E_{\law_{\theta_0}}[\pry\mid \prz= z])\mid \prz= z] \\
	&= \E_{\law_{\theta_0}}[(\pry-\E_{\law_{\theta_0}}[\pry \mid \prz])h_g(\prz)(\pry-\E_{\law_{\theta_0}}[\pry \mid \prz = z])\mid \prz = z] \\
	&= 	\V_{\law_{\theta_0}}[\pry \mid \prz= z] h_g(z) \\
	&= \ddot \psi(g(z))h_g(z).
\end{split}
\end{equation*}
The assumption~\eqref{eq:moment-assumptions-app} implies that $\ddot \psi(g_0(\prz)) \geq c > 0$ almost surely, since either $\ddot \psi$ is a nonzero constant or $g_0(\prz)$ belongs to a compact set almost surely and therefore $\ddot \psi(g_0(\prz))$ belongs to the range of a positive continuous function applied to a compact set. From this it follows that $S^*_{g}S_{g}$ is continuously invertible. Because $A_{g}$ is a continuous linear operator with continuously invertible $A^*_{g}A_{g}$, it follows that $A_{g}(L^2(\nu))$ is closed and that $A_{g}(A^*_{g}A_{g})^{-1}A^*_{g}$ is the orthogonal projection onto this space. Next let us compute the orthogonal projection of the score $A_{\beta}$ onto $A_{g}(L^2(\nu))$. We have
\begin{equation}
	\begin{split}
	(A^*_{g}A_{\beta})(z) &= \E_{\law_{\theta_0}}[A_{\beta}(\pry-\E_{\law_{\theta_0}}[\pry \mid \prz = z])\mid \prz = z] \\
	&= \E_{\law_{\theta_0}}[\prx(\pry- \E_{\law_{\theta_0}}[\pry\mid \prz])(\pry-\E_{\law_{\theta_0}}[\pry\mid \prz = z])\mid \prz = z] \\
	&=  \E_{\law_{\theta_0}}[\prx(\pry-\E_{\law_{\theta_0}}[\pry \mid \prz = z])^2\mid \prz = z] \\
	&= \E_{\law_{\theta_0}}[\prx \mid \prz = z]\ddot\psi(g(z)),
	\end{split}
\end{equation}
and therefore
\begin{equation}
	\begin{split}
	A_{g}(A^*_{g}A_{g})^{-1}A^*_{g}A_{\beta} &= (\pry-\E_{\law_{\theta_0}}[\pry \mid \prz])((A^*_{g}A_{g})^{-1}A^*_{g}A_{\beta})(\prz) \\
	&= (\pry-\E_{\law_{\theta_0}}[\pry \mid \prz])\E_{\law_{\theta_0}}[\prx \mid \prz] \\
	&= A_{g}(\E_{\law_{\theta_0}}[\prx \mid \cdot\ ]).
	\end{split}
\end{equation}
Therefore, we have
\begin{equation}
\argmin{W \in A_{g}(L^2(\nu))} \|A_{\beta}-W\|_{L^2(\law_{\theta_0})} = A_{g}(\E_{\law_{\theta_0}}[\prx \mid \cdot\ ]).
\end{equation}
Since $A_{g}(L^2(\nu))$ is closed, it follows that $\overline{A_{g}(\H_g)} \subseteq A_{g}(L^2(\nu))$. Together with the assumption~\eqref{eq:conditional-expectation-app} that $\E_{\law_{\theta_0}}[\prx \mid \cdot\ ] \in \H_g$ and the definition of the effective score as a projection onto $A_{g}(\H_g)$~\eqref{eq:projection-of-score}, we deduce that
\begin{equation*}
	\Pi_{\beta, g}A_{\beta} = \argmin{W \in \overline{A_{g}(\H_g)}}\ \|A_{\beta}-W\|_{L^2(\law_{\theta_0})} = \argmin{W \in A_{g}(L^2(\nu))} \|A_{\beta}-W\|_{L^2(\law_{\theta_0})} = A_{g}(\E_{\law_{\theta_0}}[\prx \mid \cdot\ ]).
\end{equation*}
Therefore, the efficient score is 
\begin{equation}
	\begin{split}
		\tilde{S} &= A_{\beta} - \Pi_{\beta, g}A_{\beta} \\
		&= \prx (\pry-\E_{\law_{\theta_0}}[\pry \mid \prz]) -   (\pry-\E_{\law_{\theta_0}}[\pry\mid \prz])\E_{\law_{\theta_0}}[\prx \mid \prz] \\
		&=  (\prx- \E_{\law_{\theta_0}}[\prx\mid \prz])(\pry- \E_{\law_{\theta_0}}[\pry \mid \prz]),
	\end{split}
\end{equation}
as claimed~\eqref{eq:GCM-efficient-score}. From there we find that the efficient information is
\begin{equation}
\widetilde I_{\theta_0} = \V_{\law_{\theta_0}}[\tilde{S}] = \E_{\law_{\theta_0}}[\V_{\law_{\theta_0}}[\prx|\prz]\V_{\law_{\theta_0}}[\pry|\prz]] \equiv s^2(\theta_0),
\end{equation}
also as claimed~\eqref{eq:GCM-efficient-info}.

\paragraph*{Differentiability of $\beta$ and efficient influence function.}

By \citet[Lemma 25.25]{VDV1998}, the differentiability of $\beta$ at $\law_{\theta_0}$ with respect to the tangent set $\dot \law_{\theta_0}$ follows from the quadratic mean differentiability proved above and the assumption that $\widetilde I_{\theta_0} = s^2(\theta_0) > 0$~\eqref{eq:nonsingular-fisher-info}. The same lemma gives the efficient influence function as the ratio of the efficient score and the efficient information. This completes the proof.

\section{Additional material related to simulations.}\label{sec:sim_add}

In this section, we present details about existing robustness simulation setups (Section~\ref{sec:sim_liter}), investigate the trade-off between using lasso and post-lasso (Section~\ref{sec:lasso-vs-plasso}) and present the complete simulation results (Section~\ref{sec:additional-simulation-results}).

\subsection{Simulation setup in literature}\label{sec:sim_liter}

Here we provide details about the simulation setups considered in \cite{CetL16, Liu2022a, Li2022}.

\paragraph*{\citet{Liu2022a}}
In this paper, the authors consider the double high-dimensional linear model. Suppose $\{Z_i,Y_i\}_{i=1}^{n}$ is $n=800$ iid data and $Z_i\sim N(\bm{0},\bm{\Sigma}_{p})$ where $p=800$ and $\bm{\Sigma}_p$ is chosen to be AR(1) and the autocorrelation is set to be $0.5$. Then they consider $Y_i=Z_i^{\top}\beta+\epsilon$ where $\beta$ is vector of dimension $p$ and only $s=50$ is set to be nonzero with magnitude $\nu=0.175$ and random sign. They consider two ways to set the nonzero components of $\beta$: spacing these non-zero coefficients equally or choosing them to be the first 50 coefficients of $\beta$. The authors consider, rather than testing conditional independence, the false discovery rate (FDR) of variable selection.

\paragraph*{\citet{CetL16}}
In this paper, the authors consider a bit different setting where $Y|Z$ is now a high-dimensional logistic model. The sample size $n=800$ and $Z_i\sim N(\bm{0},\bm{\Sigma}_{p})$ where $p=1500$ and $\bm{\Sigma}_p$ is chosen to be AR(1) and the autocorrelation is set to be $0.3$. After the sampling, the design matrix is centered and every column is normalized to have norm 1. Similarly, only $s=50$ coordinates of $\beta$ are set to be nonzero and the sign is random whereas the magnitude is set to be $\nu=20$. They set a randomly-chosen set of $s=50$ coefficients of $\beta$ to be nonzero and consider again the FDR control.

\paragraph*{\citet{Li2022}}
In this paper, a similar setting is considered while $Z$ is a data matrix with row $n=250$ and column $p=500$ where each row is sampled from $N(0,\bm{\Sigma}_p)$ and $\bm{\Sigma}_p$ is an AR(1) matrix with autocorrelation $0.5$. However, a crucial difference in this paper is the way to set $\E(X|Z)$ and $\E(Y|Z)$. As for $X$, it is generated by a linear predictor $X=Z\gamma+\epsilon$ where $\gamma$ is a $p$ dimensional vector with first $s=5$ is nonzero and the other coordinate remains zero and $\epsilon$ follows the standard normal distribution. The sign of each coordinate is randomly chosen and the magnitude is set to be $\nu=0.3$. As for $Y$, we set $\beta=\gamma$ and $Y=Z\beta+\xi$ where $\xi$ follows standard normal distribution such that $\xi$ is independent of $\epsilon$. We can see that $X$ and $Y$ support on the same subset of $Z$ so that the marginal association between $\prx$ and $\pry$ is much larger than that in first two simulation designs.

\subsection{Comparing the lasso and post-lasso estimation methods} \label{sec:lasso-vs-plasso}

Compare to the lasso estimation method for $\mu_{n,x}$ and $\mu_{n,y}$, the post-lasso estimation method results in estimates with lower bias but higher variance. This impacts the Type-I error and power of the inferential methods in different ways. For Type-I error, it is only important to have good estimates for $\E[\prx|\prz_{\mathcal A}] \equiv \prz_{\mathcal A}^T\beta_{\mathcal A}$ and $\E[\pry|\prz_{\mathcal A}] \equiv \prz_{\mathcal A}^T \gamma_{\mathcal A}$, where $\mathcal A \subseteq \{1, \dots, p\}$ denotes the set of variables active in both $\E[\prx|\prz]$ and $\E[\pry|\prz]$. Indeed, only the shared active coordinates of $\prz$ act as confounders. On the other hand, for power, it is important to have a good estimate for the entire function $\E[\pry|\prz] = \prz^T \gamma$  \citep{Katsevich2020a}. Therefore, we examine the mean-squared estimation error for $\prz_{\mathcal A}^T\beta_{\mathcal A}$ and $\prz_{\mathcal A}^T\gamma_{\mathcal A}$ in one of our null simulation settings as well as the mean-squared estimation error for $\prz^T \beta$ and $\prz^T\gamma$ in one of our alternative simulation settings (Figure~\ref{fig:MSE}). We find that the post-lasso does a better job estimating the shared active coefficients in the null setting, so the reduced bias in estimating these shared coefficients outweighs the increased variance. On the other hand, the lasso does a better job estimating the entire set of coefficients, so in this case the increased variance outweighs the reduced bias. This explains why the post-lasso-based methods have improved Type-I error control but worse power than the lasso-based methods.

\begin{figure}[!ht]
	\centering
	\includegraphics[width = 0.75\textwidth]{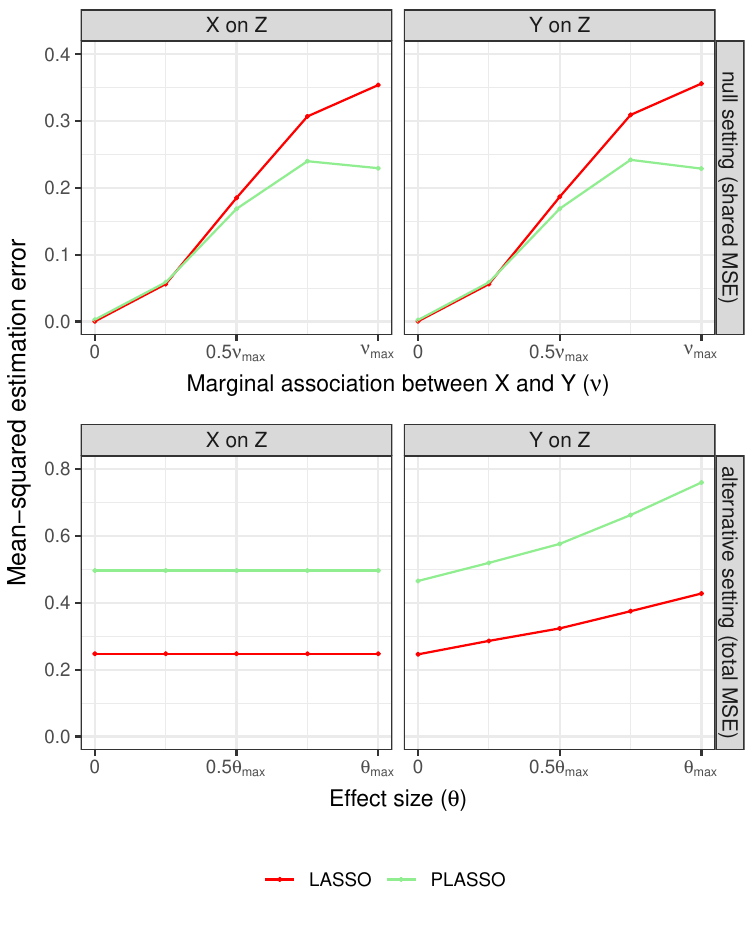}
	\caption{MSE on shared variables and total variables: first column displays the MSE of lasso and post-lasso on shared active variables $Z_{\mathcal{A}}$ and the second column displays the MSE of lasso and post-lasso on total variables. All the experiments are carried out with GCM statistic with $n=100,d=400,s=5,\rho=0.4$.}
	\label{fig:MSE}
\end{figure}

\subsection{Additional simulation results} \label{sec:additional-simulation-results}

Figures \ref{fig:gaussian_supervised_null}-\ref{fig:binomial_semi-supervised_alternative} present the complete simulation results across the null and alternative, Gaussian and binary, and supervised and unsupervised settings.

\begin{figure}[!ht]
	\centering
	\includegraphics[width = \textwidth]{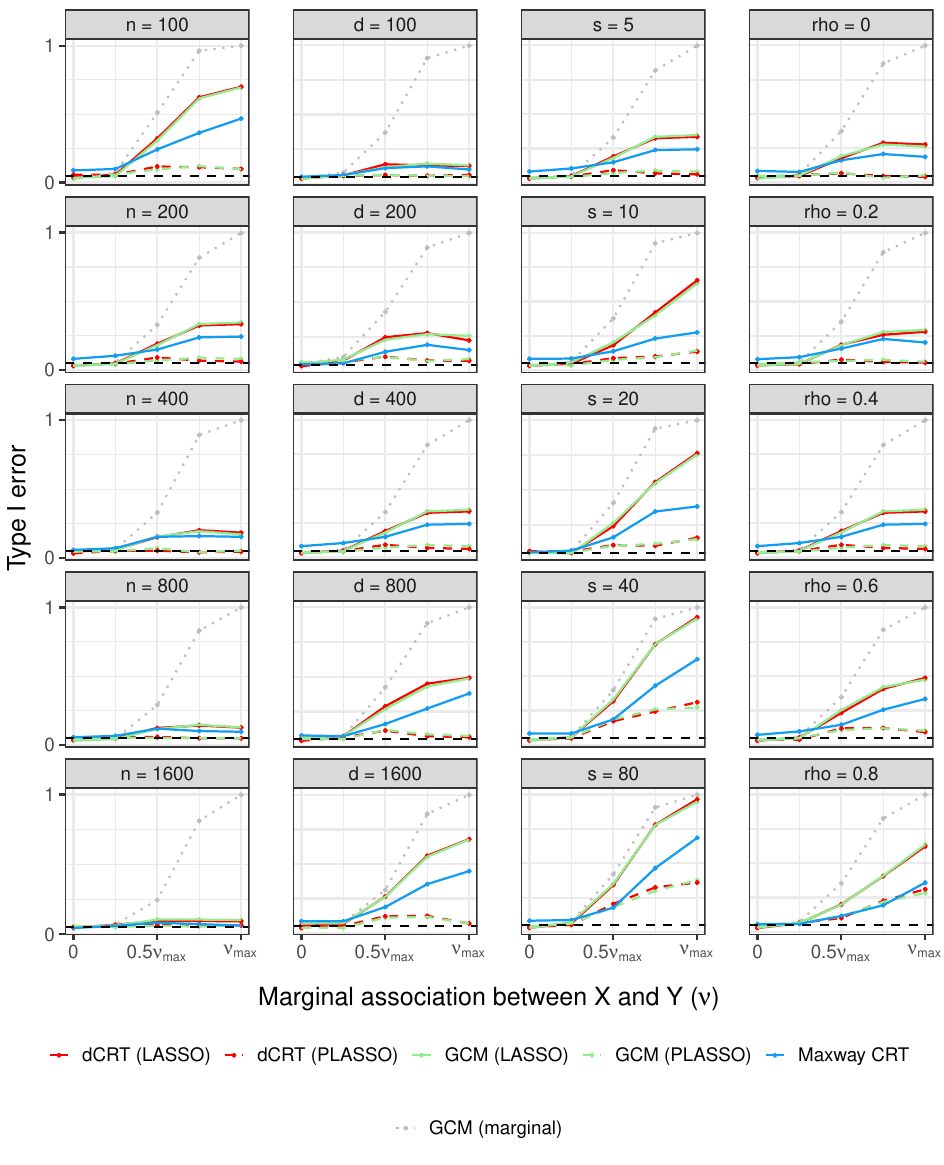}
	\caption{Type-I error in the Gaussian supervised setting.}
	\label{fig:gaussian_supervised_null}
\end{figure}

\begin{figure}[!ht]
	\centering
	\includegraphics[width = \textwidth]{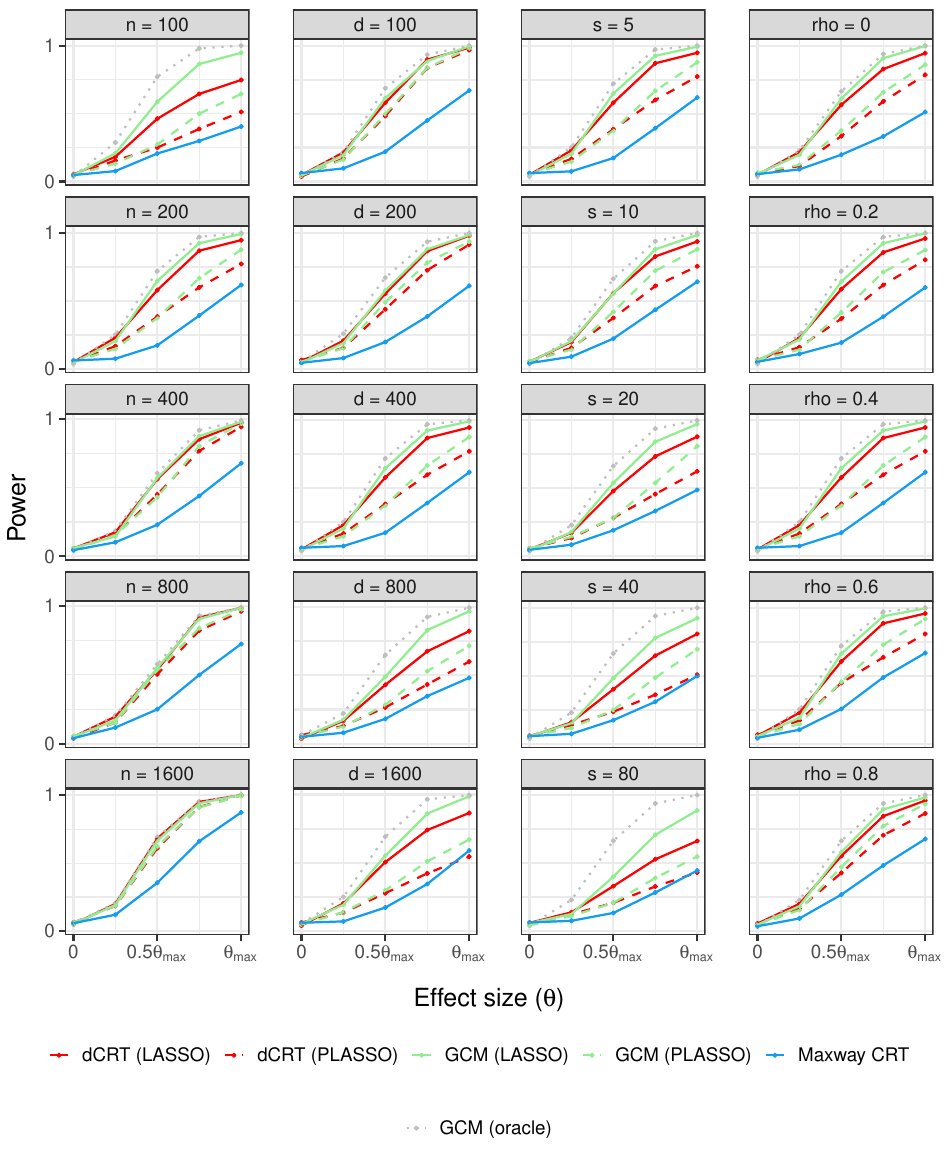}
	\caption{Power in the Gaussian supervised setting.}
	\label{fig:gaussian_supervised_alternative}
\end{figure}

\begin{figure}[!ht]
	\centering
	\includegraphics[width = \textwidth]{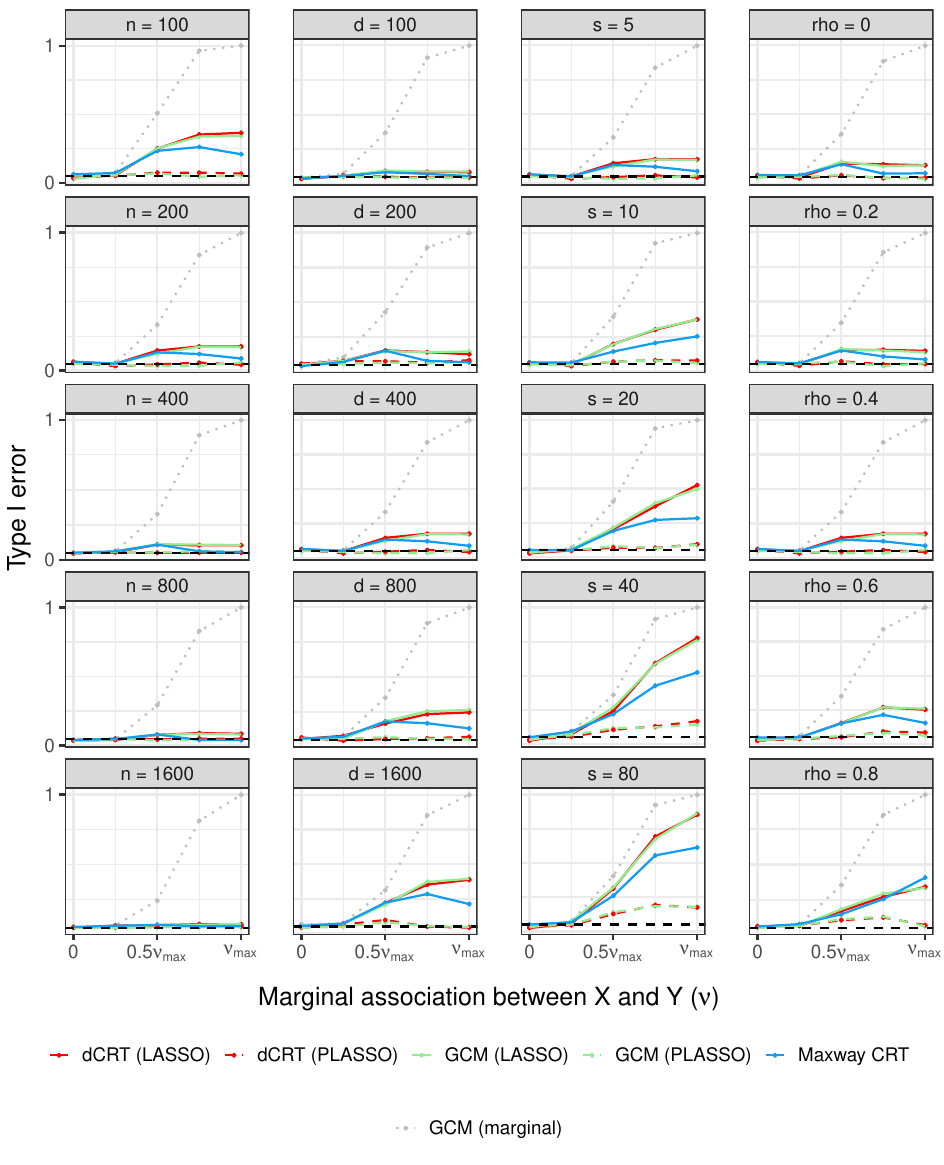}
	\caption{Type-I error in the Gaussian semi-supervised setting.}
	\label{fig:gaussian_semi-supervised_null}
\end{figure}

\begin{figure}[!ht]
	\centering
	\includegraphics[width = \textwidth]{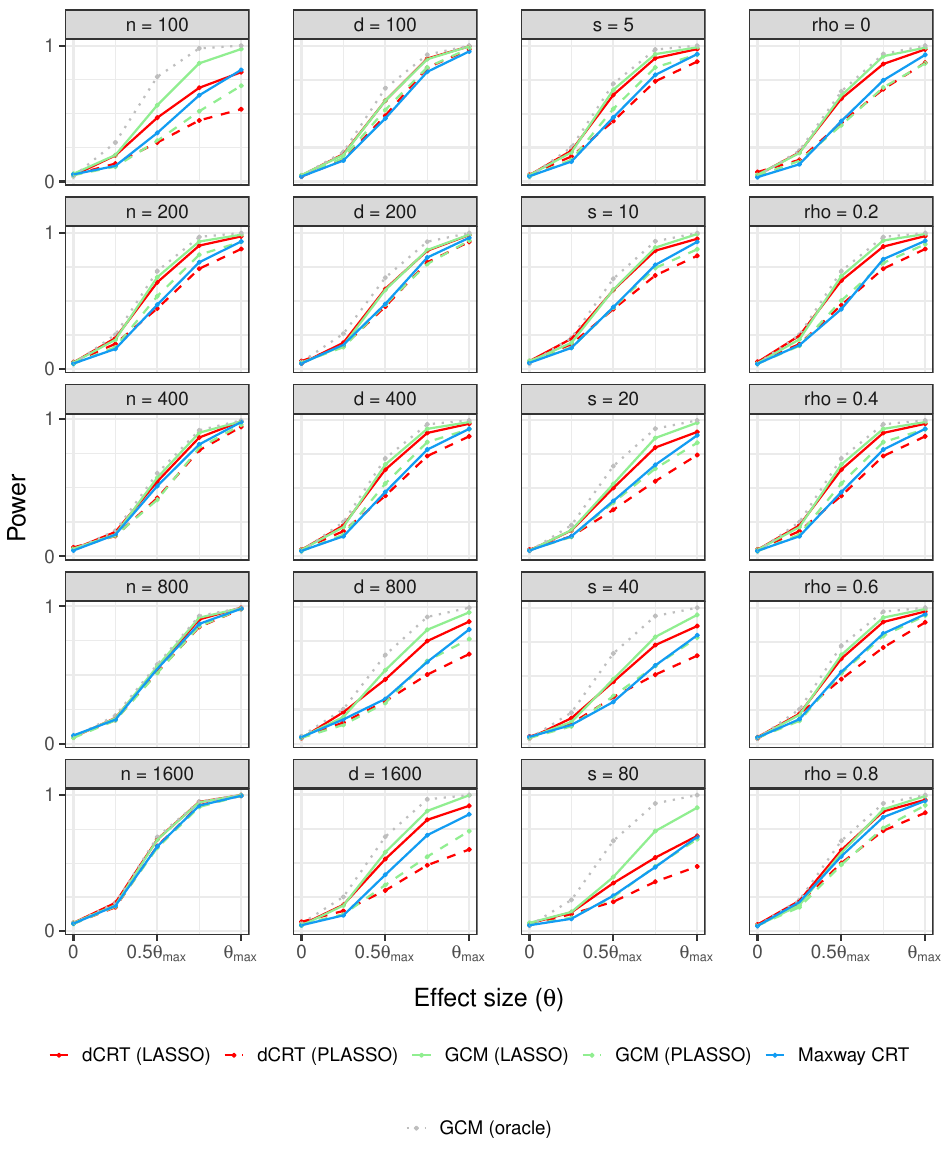}
	\caption{Power in the Gaussian semi-supervised setting.}
	\label{fig:gaussian_semi-supervised_alternative}
\end{figure}

\begin{figure}[!ht]
	\centering
	\includegraphics[width = \textwidth]{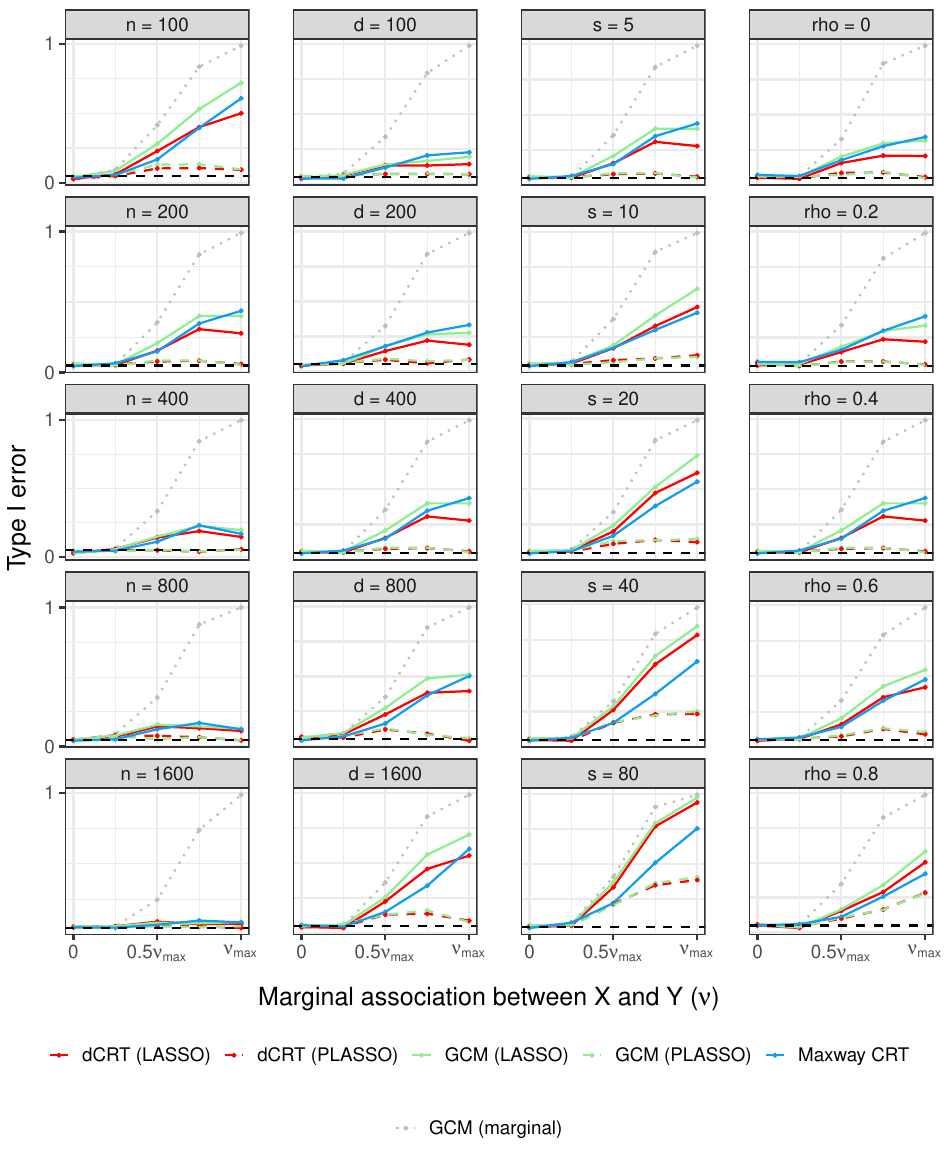}
	\caption{Type-I error in the binary supervised setting.}
	\label{fig:binomial_supervised_null}
\end{figure}

\begin{figure}[!ht]
	\centering
	\includegraphics[width = \textwidth]{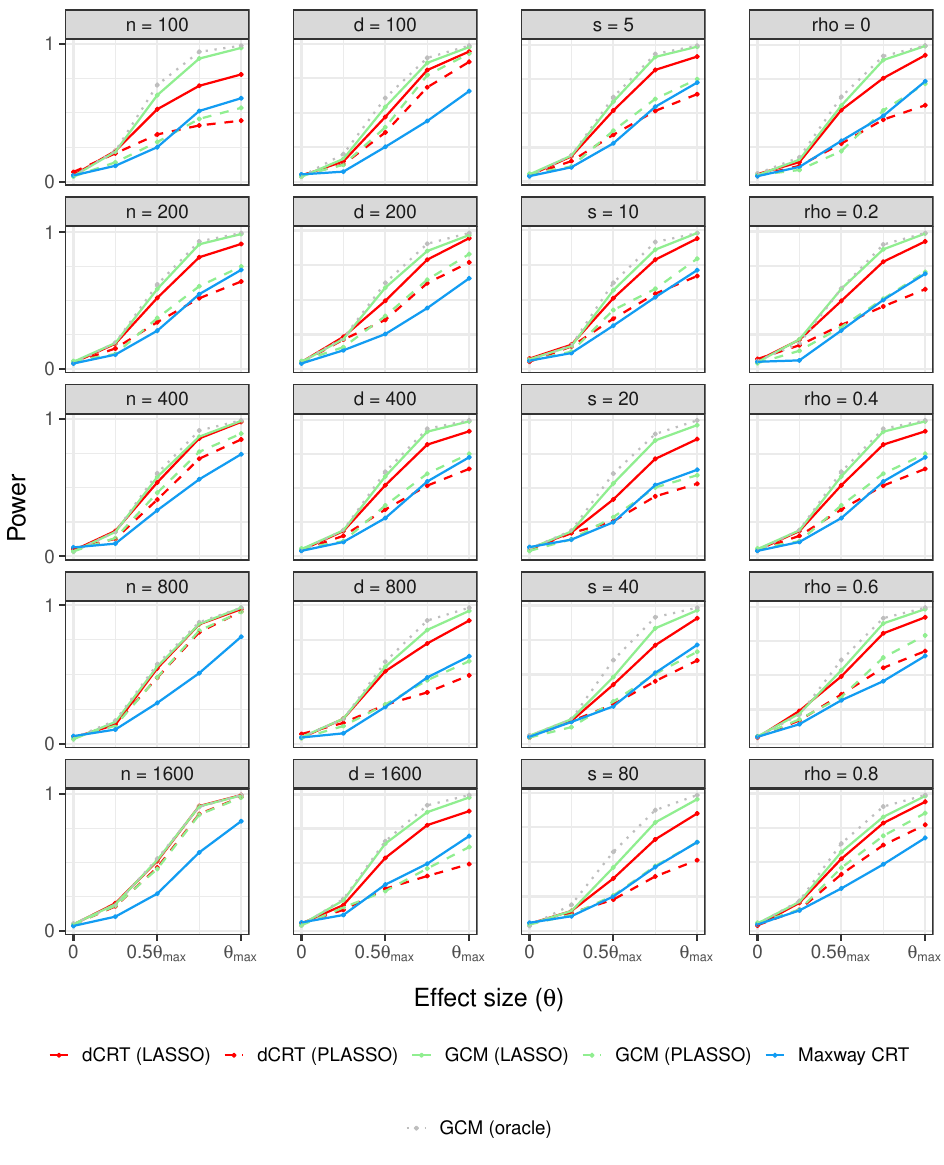}
	\caption{Power in the binary supervised setting.}
	\label{fig:binomial_supervised_alternative}
\end{figure}

\begin{figure}[!ht]
	\centering
	\includegraphics[width = \textwidth]{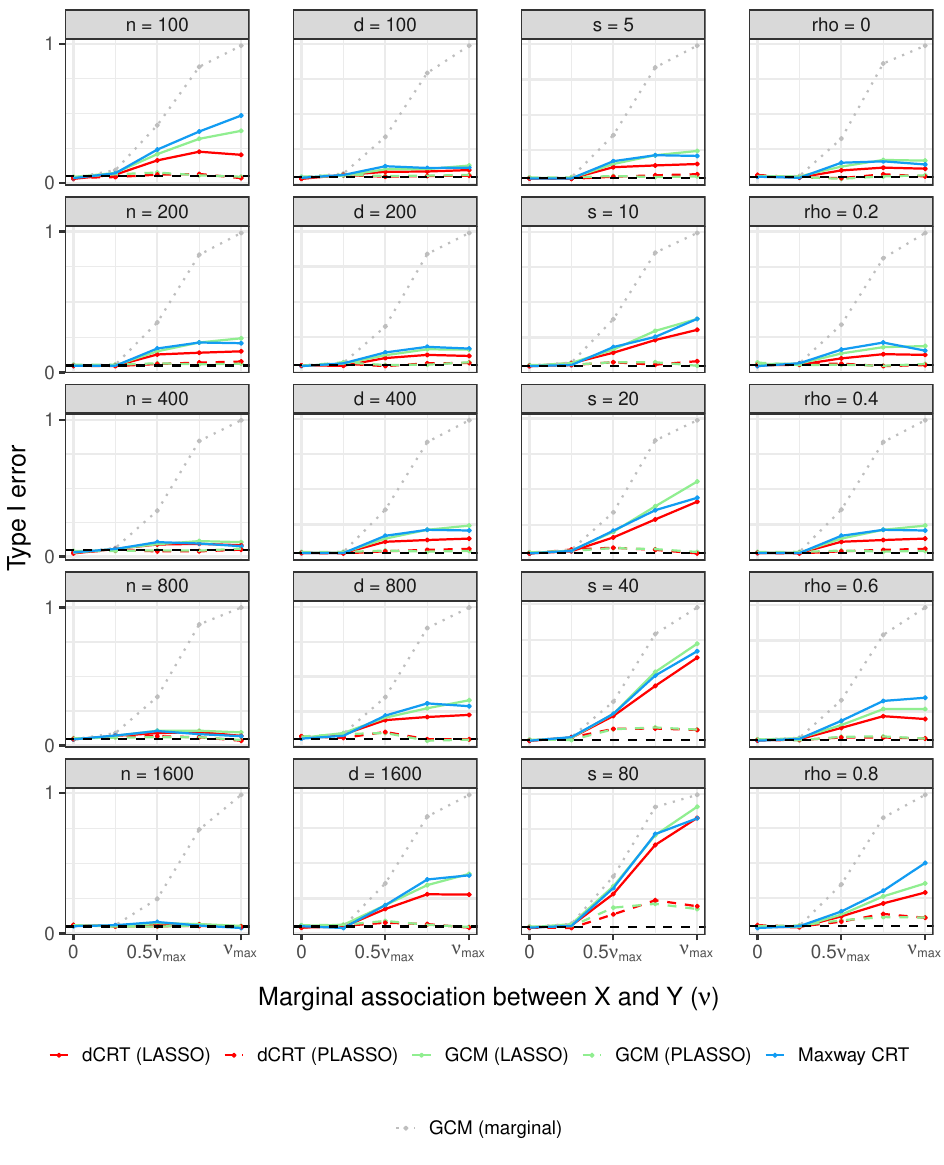}
	\caption{Type-I error in the binary semi-supervised setting.}
	\label{fig:binomial_semi-supervised_null}
\end{figure}

\begin{figure}[!ht]
	\centering
	\includegraphics[width = \textwidth]{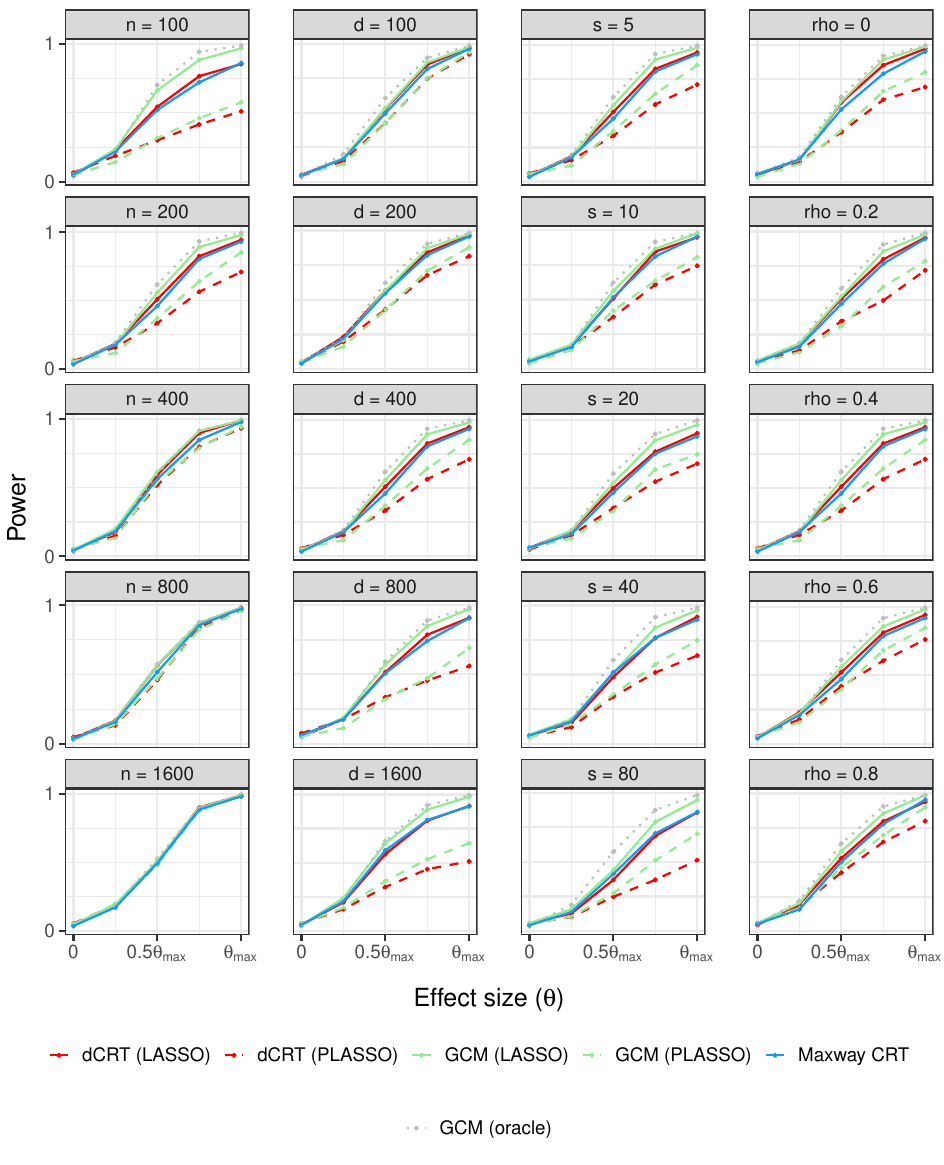}
	\caption{Power in the binary semi-supervised setting.}
	\label{fig:binomial_semi-supervised_alternative}
\end{figure}



\clearpage

\section{Proofs of conditional convergence results} \label{sec:conditional-convergence-proofs}

In this section, we present the proofs of the conditional convergence results from Appendix~\ref{sec:conditional-convergence-results}. We proceed by stating and proving necessary lemmas in Section~\ref{sec:conditional-convergence-lemmas} and then proving the convergence results themselves in Section~\ref{subsec:conditional-convergence-proofs}.

\subsection{Auxiliary lemmas} \label{sec:conditional-convergence-lemmas}

First we state a few known results for the reader's convenience.

\begin{lemma}[\cite{Durrett2010}, Theorem 2.3.2]\label{lem:sub_subseq}
A sequence of random variables $W_n$ converges to a limit $W$ in probability if and only if every subsequence of $W_n$ has a further subsequence that converges to $W$ almost surely.
\end{lemma}



\begin{lemma}[Conditional Markov inequality, \cite{Davidson2003}, Theorem 10.17]\label{lem:cond_markov}
Let $W$ be a random variable and let $\mathcal F$ be a $\sigma$-algebra. If for some $q > 0$ we have $\E[|W|^q] < \infty$, then for any $\epsilon$ we have
	\begin{align*}
		\P(|W|\geq\epsilon|\mathcal{F})\leq \frac{\E[|W|^q|\mathcal{F}]}{\epsilon^q} \quad \text{almost surely}.
	\end{align*}
\end{lemma}

\begin{lemma}[Conditional H\"older inequality, \cite{Swanson2019}, Theorem 6.60]\label{lem:cond_holder}
	Let $W_1$ and $W_2$ be random variables and let $\mathcal F$ be a $\sigma$-algebra. If for some $q_1, q_2 \in (1,\infty)$ with $\frac{1}{q_1} + \frac{1}{q_2} = 1$ we have $\E[|W_1|^{q_1}], \E[|W_2|^{q_2}] < \infty$, then
	\begin{align*}
		\E[|W_1 W_2| \mid \mathcal F] \leq (\E[|W_1|^{q_1} \mid \mathcal F])^{1/q_1}(\E[|W_2|^{q_2} \mid \mathcal F])^{1/q_2} \quad \text{almost surely}.
	\end{align*}
\end{lemma}

\begin{lemma}[\cite{RL2005}, Lemma 11.2.1]\label{lem:conver_quantile}
	Suppose $W_n \convd W$. If for given $\alpha \in (0,1)$ the CDF of $W$ is continuous and strictly increasing at $\Q_{\alpha}[W]$, then
	\begin{align*}
		\Q_{\alpha}[W_n] \rightarrow \Q_{\alpha}[W].
	\end{align*}
\end{lemma}

Next we establish that, without of loss of generality, all random variables, $\sigma$-algebras, and conditional expectations in a triangular array may be viewed as being defined on a common probability space.
\begin{lemma}[Embedding into a single probability space]\label{lem:embedding}
	Consider a sequence of probability spaces $\{(\P_n,\Omega_n,\mathcal{G}_n),n \geq 1\}$. For each $n$, let $\{W_{i,n}\}_{i \geq 1}$ be a collection of integrable random variables defined on $(\P_n,\Omega_n,\mathcal{G}_n)$ and let $\mathcal F_n \subseteq \mathcal G_n$ be a $\sigma$-algebra. Then there exists a single probability space $(\widetilde{\P}, \widetilde{\Omega}, \widetilde{\mathcal G})$, random variables $\{\widetilde W_{i,n}\}_{i,n \geq 1}$ on $(\widetilde{\P}, \widetilde{\Omega}, \widetilde{\mathcal G})$, and $\sigma$-fields $\widetilde{\mathcal F}_n \subseteq \widetilde{\mathcal G}$ for $n \geq 1$, such that for each $n$, the joint distribution of $(\{W_{i,n}\}_{i \geq 1}, \{\E[W_{i,n}\mid\mathcal F_n]\}_{i \geq 1})$ on $(\P_n,\Omega_n,\mathcal{G}_n)$ coincides with that of $(\{\widetilde W_{i,n}\}_{i \geq 1}, \{\E[\widetilde W_{i,n}\mid \widetilde{\mathcal F}_n]\}_{i \geq 1})$ on $(\widetilde{\P}, \widetilde{\Omega}, \widetilde{\mathcal G})$.
\end{lemma}

\begin{proof}
	Define the Cartesian product $\widetilde \Omega \equiv \prod_{n=1}^{\infty}\Omega_n$, the $\sigma$-algebra $\widetilde {\mathcal G}$ generated by measurable cylinders $\prod_{n = 1}^\infty A_n$  for $A_n \in \mathcal G_n$ and $A_n = \mathcal G_n$ for all but finitely many $n$, and the infinite product measure $\widetilde{\P}$ on the measurable space $(\widetilde \Omega, \widetilde{\mathcal G})$ \citep{Saeki1996}. On this probability space, define $\sigma$-algebras 
	\begin{equation}
		\widetilde{\mathcal F}_n \equiv \{\mathcal G_1 \times \cdots \times \mathcal G_{n-1} \times A_n \times\mathcal G_{n+1} \times \cdots: A_n \in \mathcal F_n\}		
	\end{equation}
	and random variables
	\begin{equation}
	\widetilde W_{in}(\omega) \equiv W_{in}(\omega_n)
	\label{eq:random-variables-embedding}
	\end{equation}
	for each $i, n \geq 1$. Next we claim that for each $i, n \geq 1$, the random variable
	\begin{equation}
	\E_{\widetilde{\P}}[\widetilde W_{in}|\widetilde{\mathcal F}_n](\omega) \equiv \E_{\P_n}[W_{in}|\mathcal F_n](\omega_n) \quad \text{for each } \omega \in \widetilde \Omega
	\label{eq:sigma-fields-embedding}
	\end{equation}
	is in fact a version of the conditional expectation $\E_{\widetilde{\P}}[\widetilde W_{in}|\widetilde{\mathcal F}_n]$. Indeed, it suffices to check that for each $A \equiv \mathcal G_1 \times \cdots \times \mathcal G_{n-1} \times A_n \times\mathcal G_{n+1} \times \cdots \in \widetilde{\mathcal F}_n$ we have 
	\begin{equation}
		\begin{split}
	\int_{A} \E_{\widetilde{\P}}[\widetilde W_{in}|\widetilde{\mathcal F}_n](\omega)d\widetilde{\P}(\omega) &\equiv \int_{A} \E_{\P_n}[W_{in}|\mathcal F_n](\omega_n) \mathrm{d}\widetilde{\P}(\omega) \\
	&= \int_{\prod_{n' \neq n} \Omega_{n'}} \int_{A_n} \E_{\P_n}[W_{in}|\mathcal F_n](\omega_n) \mathrm{d}\widetilde{\P}_n(\omega_n)\mathrm{d}\prod_{n'\neq n}\P_{n'}(\omega_{n'}) \\
	&\equiv \int_{\prod_{n' \neq n} \Omega_{n'}} \int_{A_n} W_{in}(\omega_n) \mathrm{d}\widetilde{\P}_n(\omega_n)\mathrm{d}\prod_{n'\neq n}\P_{n'}(\omega_{n'}) \\
	&= \int_{A} W_{in}(\omega_n) \mathrm{d}\widetilde{\P}(\omega) \\
	&\equiv \int_{A} \widetilde W_{in}(\omega)d\widetilde{\P}(\omega).
		\end{split}
	\end{equation}
	From the $\omega$-wise embeddings~\eqref{eq:random-variables-embedding} and~\eqref{eq:sigma-fields-embedding}, it is easy to verify the claimed equality between the joint distributions on $(\P_n,\Omega_n,\mathcal{G}_n)$ and $(\widetilde{\P}, \widetilde{\Omega}, \widetilde{\mathcal G})$.
\end{proof}

Finally, we state a conditional version of the truncated weak law of large numbers:
\begin{lemma}\label{lem:cond_wlln}
	For each $n$, let $W_{in},1\leq i\leq n$ be a set of random variables independent conditionally on $\mathcal{F}_n$. Let $b_n>0$ with $b_n\rightarrow\infty$ and let $\Bar W_{in}=W_{in}\indicator(|W_{in}|\leq b_n)$. Suppose that as $n\rightarrow\infty$ we have
	\begin{enumerate}
		\item $\sum_{i=1}^{n}\P[|W_{in}|>b_n|\mathcal{F}_n]\convp0$ and
		\item $b_n^{-2}\sum_{i=1}^n\E[\Bar W_{in}^2|\mathcal{F}_n]\convp0$.
	\end{enumerate}
	If we set $S_n\equiv\sum_{i = 1}^n W_{in}$ and $a_n \equiv \sum_{i=1}^n\E[\Bar W_{in}]$ then
	\begin{align*}
		\frac{S_n - a_n}{b_n} \mid \mathcal{F}_n \convpp 0.
	\end{align*}
\end{lemma}
\begin{proof}
	Let $\Bar{S}_n \equiv \sum_{i = 1}^n \Bar W_{in}$. We first write
	\begin{align*}
		\P\left[\left|\frac{S_n-a_n}{b_n}\right|>\epsilon\bigg|\mathcal{F}_n\right]\leq \P\left[S_n\neq\Bar{S}_n|\mathcal{F}_n\right]+\P\left[\left|\frac{\Bar{S}_n-a_n}{b_n}\right|>\epsilon\bigg|\mathcal{F}_n\right].
	\end{align*}
	To estimate the first term, we note that
	\begin{align*}
		\P[S_n\neq\Bar{S}_n|\mathcal{F}_n]\leq \P[\cup_{i=1}^n\{\Bar W_{in}\neq W_{in}\}|\mathcal{F}_n] \leq \sum_{i=1}^n\P(|W_{in}|>b_n|\mathcal{F}_n)\convp0
	\end{align*}
	by the first assumption. For the second term, we note that conditional Markov's inequality (Lemma \ref{lem:cond_markov}), $a_n=\E[\Bar{S}_n|\mathcal{F}_n]$ implies that
	\begin{align*}
		\P\left[\left|\frac{\Bar{S}_n-a_n}{b_n}\right|>\epsilon|\mathcal{F}_n\right]
		&
		\leq \epsilon^{-2}\E\left[\left|\frac{\Bar{S}_n-a_n}{b_n}\right|^2\bigg|\mathcal{F}_n\right]\\
		&
		=\epsilon^{-2}b_n^{-2}\mathrm{Var}[\Bar{S}_n|\mathcal{F}_n]\\
		&
		=(b_n\epsilon)^{-2}\sum_{i=1}^n\mathrm{Var}[\Bar W_{in}|\mathcal{F}_n]\\
		&
		\leq (b_n\epsilon)^{-2}\sum_{i=1}^{n}\E\left[\Bar{X}^2_{in}|\mathcal{F}_n\right]\convp0,
	\end{align*}
	where the convergence in the last line is given by the second assumption. This completes the proof.
\end{proof}

	
	\subsection{Proofs of conditional convergence results} \label{subsec:conditional-convergence-proofs}
	
	\begin{proof}[Proof of Theorem~\ref{thm:cond_polya}] 
		This proof generalizes the argument of \cite[Lemma 2.11]{VDV1998} to allow for conditioning. Fix $\epsilon > 0$, and choose an integer $k \geq 2/\epsilon$. Because the CDF of $W$ is continuous, it follows that 
		\begin{equation*}
			\P[W \leq \Q_{i/k}[W]] = i/k \quad \text{for each }0 \leq i \leq k.
		\end{equation*}
		Fix $t \in \R$, and suppose $\Q_{\frac{i-1}{k}}[W] \leq t \leq \Q_{\frac i k}[W]$. It follows that 
		\begin{equation*}
			\P[W_n \leq \Q_{\frac{i-1}{k}}[W] \mid \mathcal F_n] - \frac{i}{k} \leq \P[W_n \leq t \mid \mathcal F_n] - \P[W \leq t] \leq \P[W_n \leq \Q_{\frac i k}[W] \mid \mathcal F_n] - \frac{i-1}{k}.
		\end{equation*}
		Therefore, for all $t \in \R$, we have
		\begin{equation*}
			\begin{split}
				|\P[W_n \leq t \mid \mathcal F_n] - \P[W \leq t]| &\leq \sup_{0 \leq i \leq k} \left|\P[W_n \leq \Q_{i/k}[W] \mid \mathcal F_n] - \frac{i}{k}\right| + \frac 1 k \\
				&=  \sup_{0 \leq i \leq k} \left|\P[W_n \leq \Q_{\frac i k}[W] \mid \mathcal F_n] - \mathbb P[W \leq \Q_{\frac i k}[W]]\right| + \frac{1}{k},
			\end{split}
		\end{equation*}
		so that
		\begin{equation*}
			\sup_{t \in \R}|\P[W_n \leq t \mid \mathcal F_n] - \P[W \leq t]| \leq \sup_{0 \leq i \leq k} \left|\P[W_n \leq \Q_{\frac i k}[W] \mid \mathcal F_n] - \mathbb P[W \leq \Q_{\frac i k}[W]]\right| + \frac{1}{k}.
		\end{equation*}
		By assumption, we have
		\begin{equation}
			\P\left[\left|\P[W_n \leq \Q_{\frac i k}[W] \mid \mathcal F_n] - \mathbb P[W \leq \Q_{\frac i k}[W]]\right| > \frac{1}{k(k+1)}\right] \rightarrow 0. 
		\end{equation}
		Therefore, 
		\begin{equation*}
			\begin{split}
				&\P\left[\sup_{t \in \R}|\P[W_n \leq t \mid \mathcal F_n] - \P[W \leq t]| > \epsilon \right] \\
				&\quad\leq \P\left[\sup_{t \in \R}|\P[W_n \leq t \mid \mathcal F_n] - \P[W \leq t]| > \frac{2}{k} \right] \\
				&\quad\leq \P\left[\sup_{0 \leq i \leq k} \left|\P[W_n \leq \Q_{\frac i k}[W] \mid \mathcal F_n] - \mathbb P[W \leq \Q_{\frac i k}[W]]\right| > \frac{1}{k} \right] \\
				&\quad\leq \sum_{i = 0}^k  \P\left[ \left|\P[W_n \leq \Q_{\frac i k}[W] \mid \mathcal F_n] - \mathbb P[W \leq \Q_{\frac i k}[W]]\right| > \frac{1}{k(k+1)} \right] \\
				&\quad \rightarrow 0.
			\end{split}
		\end{equation*}
		This completes the proof.
	\end{proof}
	
	\begin{proof}[Proof of Theorem \ref{thm:cond_slutsky}]
		Fix $t \in \R$. Letting $F(t') \equiv \P[W \leq t']$ be the CDF of $W$, Theorem \ref{thm:cond_polya} gives 
		\begin{align*}
			\left|\P\left[W_n\leq \frac{t - b_n}{a_n}\mid\mathcal{F}_n\right]-F\left(\frac{t - b_n}{a_n}\right)\right| \leq	\sup_{t' \in \R}|\P[W_n\leq t'|\mathcal{F}_n]-F(t')|\convp0.
		\end{align*}
		By the continuous mapping theorem, we have $F\left(\frac{t - b_n}{a_n}\right) \convp F(t)$, so that
		\begin{equation*}
			\P\left[W_n\leq \frac{t - b_n}{a_n}\mid\mathcal{F}_n\right] \convp F(t).
		\end{equation*}
		Noting that $\P[a_n \leq 0 | \mathcal F_n]$ is a sequence of nonnegative random variables whose expectations converge to zero, it follows that $\P[a_n \leq 0 | \mathcal F_n] \convp 0$ and so
		\begin{equation*}
			\begin{split}
				\P\left[a_nW_n + b_n \leq t \mid\mathcal{F}_n\right] &= \P\left[a_nW_n + b_n \leq t, a_n > 0\mid\mathcal{F}_n\right] + o_p(1) \\
				&= \P\left[W_n\leq \frac{t - b_n}{a_n}, a_n > 0\mid\mathcal{F}_n\right] + o_p(1) \\
				&= \P\left[W_n\leq \frac{t - b_n}{a_n}\mid\mathcal{F}_n\right] + o_p(1) \\
				&\convp F(t),
			\end{split}
		\end{equation*}
		as desired.
	\end{proof}

	\begin{proof}[Proof of Theorem \ref{thm:wlln_cond}] 
		We apply Lemma \ref{lem:cond_wlln} with $b_n=n$. We first verify the first assumption in Lemma \ref{lem:cond_wlln} by conditional Markov's inequality (Lemma \ref{lem:cond_markov}):
		\begin{align*}
			\sum_{i=1}^n\P(|W_{in}|>n|\mathcal{F}_n)\leq \sum_{i=1}^n\frac{\E(|W_{in}|^{1+\delta}|\mathcal{F}_n)}{n^{1+\delta}}\convp0.
		\end{align*}
		For the second condition, we have
		\begin{align*}
			\frac{1}{n^{2}}\sum_{i=1}^n\E[W_{in}^2\indicator(|W_{in}|\leq n)|\mathcal{F}_n]
			&
			\leq \frac{1}{n^2}\sum_{i=1}^n\E[|W_{in}|^{1+\delta}n^{1-\delta}\indicator(|W_{in}|\leq n)]\\
			&
			= \frac{1}{n^{1+\delta}}\sum_{i=1}^n\E[|W_{in}|^{1+\delta}\indicator(|W_{in}|\leq n)|\mathcal{F}_n]\\
			&
			\leq \frac{1}{n^{1+\delta}}\sum_{i=1}^n\E[|W_{in}|^{1+\delta}|\mathcal{F}_n]\convp0.
		\end{align*}
		Therefore, Lemma~\ref{lem:cond_wlln} yields
		\begin{align*}
			\left.\frac1n\sum_{i=1}^n (W_{in}-\E[W_{in}\indicator(|W_{in}|\leq n) \mid \mathcal{F}_n]) \ \right|\ \mathcal{F}_n \convpp 0.
		\end{align*}
		By conditional Slutsky (Theorem~\ref{thm:cond_slutsky}), it now suffices to show that
		\begin{align*}
			\frac{1}{n}\sum_{i=1}^n\E\left[W_{in}\indicator(|W_{in}| > n) \mid \mathcal{F}_n\right] \convp 0.
		\end{align*}
		To see this, applying conditional Markov's and H\"older's inequalities (Lemmas~\ref{lem:cond_markov} and~\ref{lem:cond_holder}, respectively) we obtain
		\begin{align*}
			\left|\frac{1}{n}\sum_{i=1}^n\E\left[W_{in}\indicator(|W_{in}| > n)|\mathcal{F}_n\right]\right|
			&
			\leq \frac{1}{n}\sum_{i=1}^n\E\left[|W_{in}|\indicator(|W_{in}| > n)|\mathcal{F}_n\right]\\
			&
			\leq\frac{1}{n}\sum_{i=1}^n\{\E[|W_{in}|^{1+\delta}|\mathcal{F}_n]\}^{1/(1+\delta)}\{\P[|W_{in}|>n|\mathcal{F}_n]\}^{\delta/(1+\delta)} \\
			&
			\leq \frac{1}{n}\sum_{i=1}^n\{\E[|W_{in}|^{1+\delta}|\mathcal{F}_n]\}^{1/(1+\delta)}\left\{\frac{\E(|W_{in}|^{1+\delta}|\mathcal{F}_n)}{n^{1+\delta}}\right\}^{\delta/(1+\delta)}\\
			&
			=\frac{1}{{n^{1+\delta}}}\sum_{i=1}^n\E[|W_{in}|^{1+\delta}|\mathcal{F}_n]\\
			&
			\convp0,
		\end{align*}
		where the last convergence is by assumption. Finally, we verify that the condition~\eqref{eq:wlln_cond_sufficient} is sufficient for the conditional WLLN assumption~\eqref{eq:wlln_cond_assumption} by noting that it implies
		\begin{align*}
			\frac{1}{n^{1+\delta}} \sum_{i = 1}^n \E[|W_{in}|^{1+\delta} \mid \mathcal{F}_n]\leq \frac{\sup_{1\leq i\leq n}\E[|W_{in}|^{1+\delta} \mid \mathcal{F}_n]n}{n^{1+\delta}}=\frac{\sup_{1\leq i\leq n}\E[|W_{in}|^{1+\delta} \mid \mathcal{F}_n]}{n^{\delta}} \convp 0.
		\end{align*}
		This completes the proof.
	\end{proof}
	
	\begin{proof}[Proof of Theorem \ref{thm:conditional-clt}] 
		Without loss of generality, we assume $\E[W_{in}|\mathcal{F}_n]=0$ and that all random variables and $\sigma$-algebras are defined on a common probability space $(\P, \Omega, \mathcal G)$ (Lemma~\ref{lem:embedding}). Let $\mathcal B(\R^n)$ be the Borel $\sigma$-algebra on $\R^n$. Let $\kappa_{n}$ be a regular conditional distribution of $(W_{1n}, \dots, W_{nn})$ given $\mathcal{F}_{n}$ \citep[Theorem 8.37]{Lista2017}, i.e. a function $\kappa_{n}: \Omega \times \mathcal B(\R^n) \rightarrow [0,\infty]$ such that $\omega \mapsto \kappa_{n}(\omega, B)$ is measurable for each $B \in \mathcal B(\R^n)$, $B \mapsto \kappa_{n}(\omega, B)$ is a $\sigma$-finite measure on $\R^n$ for each $\omega \in \Omega$, and 
		\begin{align*}
			\kappa_{n}(\omega, B) = \P[(W_{1n}, \dots, W_{nn})\in B|\mathcal{F}_{n}](\omega), \quad \text{for almost all } \omega \in \Omega \text{ and all} \ B \in \mathcal B(\R^n).
		\end{align*}
		For each $n$ and each $\omega \in \Omega$, let $(\widetilde W_{1n}(\omega), \dots, \widetilde W_{nn}(\omega))$ be a draw from the measure $\kappa_n(\omega,\cdot)$. By \citet[Theorem 8.38]{Lista2017}, we have for each $n$ that
		\begin{equation*}
			\left(\sum_{i = 1}^n \V[\widetilde W_{in}(\omega)]\right)^{-(2+\delta)/2}\sum_{i = 1}^n \E[|\widetilde W_{in}(\omega)|^{2+\delta}] \overset{a.s.}= \frac{1}{S_{n}^{2+\delta}(\omega)} \sum_{i = 1}^{n} \E[|W_{in}|^{2+\delta} \mid \mathcal{F}_{n}](\omega).
		\end{equation*}
		Now, let $\{n_k\}_{k \geq 1}$ be a subsequence of $\mathbb N$. By the conditional Lyapunov assumption~\eqref{eq:conditional-lyapunov} and Lemma~\ref{lem:sub_subseq}, there is a further subsequence $n_{k_j}$ such that
		\begin{equation}
			\frac{1}{S_{n_{k_j}}^{2+\delta}} \sum_{i = 1}^{n_{k_j}} \E[|W_{in_{k_j}}|^{2+\delta} \mid \mathcal{F}_{n_{k_j}}] \convas 0.
		\end{equation}
		Hence, it follows that 
		\begin{equation}
			\left(\sum_{i = 1}^{n_{k_j}} \V[\widetilde W_{in_{k_j}}(\omega)]\right)^{-(2+\delta)/2}\sum_{i = 1}^{n_{k_j}} \E[|\widetilde W_{i{n_{k_j}}}(\omega)|^{2+\delta}] \rightarrow 0 \quad \text{for almost every } \omega \in \Omega.
		\end{equation}
		Applying the usual Lyapunov CLT to the triangular array $\{\widetilde W_{in_{k_j}}(\omega)\}_{i,n_{k_j}}$, we find that
		\begin{equation}
			\left(\sum_{i = 1}^{n_{k_j}} \V[\widetilde W_{in_{k_j}}(\omega)]\right)^{-1/2}\sum_{i = 1}^{n_{k_j}} \widetilde W_{in_{k_j}}(\omega) \convd N(0,1) \quad \text{for almost every } \omega \in \Omega,
		\end{equation}
		and therefore that, for each $t \in \R$, we have 
		\begin{equation}
			\mathbb P\left[\left(\sum_{i = 1}^{n_{k_j}} \V[\widetilde W_{in_{k_j}}(\omega)]\right)^{-1/2}\sum_{i = 1}^{n_{k_j}} \widetilde W_{in_{k_j}}(\omega) \leq t\right] \rightarrow \Phi(t) \quad \text{for almost every } \omega \in \Omega.
		\end{equation}
		Using \citet[Theorem 8.38]{Lista2017} again, it follows that for each $t \in \R$, we have
		\begin{equation}
			\mathbb P\left[\frac{1}{S_{n_{k_j}}}\sum_{i = 1}^{n_{k_j}} W_{in_{k_j}} \leq t \mid \mathcal F_{n_{k_j}}\right] \convas \Phi(t).
		\end{equation}
		Applying Lemma~\ref{lem:sub_subseq}, it follows that
		\begin{equation}
			\mathbb P\left[\frac{1}{S_{n}}\sum_{i = 1}^{n} W_{in} \leq t \mid \mathcal F_{n}\right] \convp \Phi(t),
		\end{equation}
		as desired.
	\end{proof}
	
	\begin{proof}[Proof of Lemma \ref{lem:conditional-convergence-to-quantile}]
		
		Without loss of generality, we assume that all random variables and $\sigma$-algebras are defined on a common probability space $(\P, \Omega, \mathcal G)$ (Lemma~\ref{lem:embedding}). Let $\mathcal B(\R)$ be the Borel $\sigma$-algebra on $\R$. Let $\kappa_{n}$ be a regular conditional distribution of $W_n$ given $\mathcal{F}_{n}$ \citep[Theorem 8.29]{Lista2017}, i.e. a function $\kappa_{n}: \Omega \times \mathcal B(\R) \rightarrow [0,\infty]$ such that $\omega \mapsto \kappa_{n}(\omega, B)$ is measurable for each $B \in \mathcal B(\R)$, $B \mapsto \kappa_{n}(\omega, B)$ is a $\sigma$-finite measure on $\R^n$ for each $\omega \in \Omega$, and 
		\begin{align*}
			\kappa_{n}(\omega, B) = \P[W_{n} \in B|\mathcal{F}_{n}](\omega), \quad \text{for almost all } \omega \in \Omega \text{ and all} \ B \in \mathcal B(\R).
		\end{align*}
		Now, let $\{n_k\}_{k \geq 1}$ be a subsequence of $\mathbb N$. By conditional Polya's theorem (Theorem~\ref{thm:cond_polya}), we have 
		\begin{equation}
			\sup_{t \in \R}|\P[W_{n}\leq t|\mathcal{F}_n]-\P[W\leq t]| \convp0.
		\end{equation}
		Hence, by Lemma~\ref{lem:sub_subseq} there is a further subsequence $n_{k_j}$ such that
		\begin{equation}
			\sup_{t \in \R}|\P[W_{n_{k_j}}\leq t|\mathcal{F}_{n_{k_j}}](\omega)-\P(W\leq t)|\rightarrow0,\ \text{for almost all } \omega \in \Omega.
		\end{equation}	
		It follows that 
		\begin{equation}
			\sup_{t \in \R}|\kappa_{n_{k_j}}(\omega, (-\infty, t])-\P(W\leq t)|\rightarrow0,\ \text{for almost all } \omega \in \Omega,
		\end{equation}
		i.e.
		\begin{equation}
			\kappa_{n_{k_j}}(\omega, \cdot) \convd W \ \text{for almost all } \omega \in \Omega.
		\end{equation}
		Hence, by Lemma~\ref{lem:conver_quantile}, it follows that
		\begin{equation}
			\Q_{\alpha}[\kappa_{n_{k_j}}(\omega, \cdot)] \rightarrow \Q_{\alpha}[W] \ \text{for almost all } \omega \in \Omega,
		\end{equation}
		and therefore
		\begin{equation}
			\Q_{\alpha}[W_{n_{k_j}}|\mathcal{F}_{n_{k_j}}] \convas \Q_{\alpha}[W].
		\end{equation}
		Applying Lemma~\ref{lem:sub_subseq} again, we conclude that
		\begin{equation}
			\Q_{\alpha}[W_{n}|\mathcal{F}_{n}] \convp \Q_{\alpha}[W],
		\end{equation}
		as desired.
	\end{proof}

\end{document}